\documentclass{lmcs}
\pdfoutput=1

\usepackage{lastpage}
\lmcsdoi{20}{1}{4}
\lmcsheading{}{\pageref{LastPage}}{}{}%
{Sep.~08,~2022}{Jan.~19,~2024}{}

\pdfoutput=1


\usepackage{hyperref}

\usepackage{wrapfig}
\usepackage{float}

\usepackage{graphicx}
\usepackage{tikz}
\usepackage{subfigure}
\usepackage{pgf}
\usepackage{gastex}
\usepackage{adjustbox}

\usetikzlibrary{arrows,
	automata,
	calc,
	decorations,
	decorations.pathreplacing,
	positioning,
	shapes}

\usepackage[utf8]{inputenc}

\usepackage{amsmath}
\usepackage{amsfonts}
\usepackage{amssymb}
\usepackage{eulervm}

\newcommand{\MEC}[1]{\mathit{MEC}(#1)}

\newcommand{\effuncov}{\mathsf{eff}_{\mathsf{unc}}}
\newcommand{\effunc}{\effuncov}
\newcommand{\effcov}{\mathsf{eff}_{\mathsf{cov}}}

\newcommand{\freq}[2]{\mathit{freq}_{#1}(#2)}

\newcommand{\Nat}{\mathbb{N}}
\newcommand{\Real}{\mathbb{R}}

\newcommand{\Rat}{\mathbb{Q}}
\newcommand{\Rational}{\Rat}

\def\Pr{\mathrm{Pr}}

\newcommand{\eqdef}{\stackrel{\text{\rm \tiny def}}{=}}

\newcommand{\init}{\mathsf{init}}

\newcommand{\StAct}{\mathit{StAct}}
\newcommand{\SA}{\StAct}

\newcommand{\Act}{\mathit{Act}}

\newcommand{\Acc}{\mathit{Acc}}

\newcommand{\residual}[2]{\mathit{res}(#1,#2)}

\newcommand{\yes}{\mathit{yes}}
\newcommand{\no}{\mathit{no}}
\newcommand{\choice}{\mathit{choice}}

\newcommand{\sched}{\mathfrak{S}}

\newcommand{\tsched}{\mathfrak{T}}
\newcommand{\usched}{\mathfrak{U}}
\newcommand{\vsched}{\mathfrak{V}}

\newcommand{\Ssched}{\Sigma}

\newcommand{\Usched}{\Upsilon}

\newcommand{\last}{\mathit{last}}

\newcommand{\Cause}{\mathsf{Cause}}
\newcommand{\rCause}{\mathsf{rCause}}

\newcommand{\Effect}{\mathsf{Eff}}
\newcommand{\Eff}{\Effect}
\newcommand{\rEffect}{\mathsf{rEff}}
\newcommand{\rEff}{\rEffect}
\newcommand{\effect}{\mathsf{eff}}
\newcommand{\eff}{\effect}
\newcommand{\noeff}{\mathsf{noeff}}
\newcommand{\CanCause}{\mathsf{CanCause}}
\newcommand{\CanCau}{\CanCause}

\newcommand{\noeffc}{\noeff_{\mathsf{fp}}}
\newcommand{\noefffp}{\noeffc}
\newcommand{\noeffbot}{\noeff_{\mathsf{tn}}}
\newcommand{\noefftn}{\noeffbot}

\newcommand{\precision}{\operatorname{\mathit{precision}}}
\newcommand{\recall}{\operatorname{\mathit{recall}}}
\newcommand{\fscore}{\mathit{fscore}}

\newcommand{\relcov}{\recall}
\newcommand{\ratiocov}{\mathit{covrat}}
\newcommand{\covratio}{\ratiocov}
\newcommand{\covrat}{\ratiocov}

\newcommand{\ratio}[3]{\mathit{ratio}^{#1}_{#2}(#3)}

\newcommand{\wminMDP}[2]{#1_{[#2]}}
\newcommand{\wminMDPmax}[2]{#1^{\max}_{[#2]}}

\newcommand{\Until}{ \, \mathrm{U}\, }
\newcommand{\until}{\Until}

\newcommand{\PTIME}{\mathrm{P}}

\newcommand{\NP}{\mathrm{NP}}
\newcommand{\PFNP}{\mathrm{PF}^{\NP}}
\newcommand{\coNP}{\mathrm{coNP}}

\newcommand{\PSPACE}{\mathrm{PSPACE}}

\def\cA{\mathcal{A}}

\def\cC{\mathcal{C}}
\def\cD{\mathcal{D}}
\def\cE{\mathcal{E}}

\def\cG{\mathcal{G}}

\def\cK{\mathcal{K}}

\def\cM{\mathcal{M}}
\def\cN{\mathcal{N}}

\newcommand{\Ende}{\hfill $\lhd$}

\usepackage{thmtools}

\usepackage{xcolor}

\begin{document}

\title[Probability-raising causality in MDPs]{Foundations of probability-raising causality\texorpdfstring{\\}{} in Markov decision processes
  }
\titlecomment{This is the extended version of the conference version at FoSSaCS 2022 \cite{fossacs2022}}
\thanks{This work was funded by DFG grant 389792660 as part of TRR~248 -- CPEC (see \url{https://perspicuous-computing.science}, the Cluster of Excellence EXC 2050/1 (CeTI, project ID 390696704, as part of Germany’s Excellence Strategy), DFG-projects BA-1679/11-1 and BA-1679/12-1.}
\author[C.~Baier]{Christel Baier\lmcsorcid{0000-0002-5321-9343}}
\author[J.~Piribauer]{Jakob Piribauer\lmcsorcid{0000-0003-4829-0476}}
\author[R.~Ziemek]{Robin Ziemek\lmcsorcid{0000-0002-8490-1433}}
\address{Technische Universit\"at Dresden}
\email{christel.baier@tu-dresden.de,
	jakob.piribauer@tu-dresden.de, 
	robin.ziemek@tu-dresden.de}

\begin{abstract}
	This work introduces a novel cause-effect relation in Markov decision processes using the probability-raising principle.
	Initially,  sets of states as causes and effects are considered, which is subsequently extended to regular path properties as effects and then as causes.
	The paper lays the mathematical foundations and analyzes the algorithmic properties of these cause-effect relations.
	This includes algorithms for checking cause conditions given an effect and deciding the existence of probability-raising causes.
	As the definition allows for sub-optimal coverage properties, quality measures for causes inspired by concepts of statistical analysis are studied.
	These include recall, coverage ratio and f-score.
	The computational complexity for finding optimal causes with respect to these measures is analyzed.
\end{abstract}

\maketitle

\section{Introduction}

In  recent years, scientific and technological advancement in computer science and engineering led to an ever increasing influence of computer systems on our everyday lives.
A lot of decisions which were historically done by humans are now in the hands of intelligent systems.
At the same time, these systems grow more and more complex, and thus, harder to understand.
This poses a huge challenge in the development of reliable and trustworthy systems.
Therefore, an important task of computer science today is the development of comprehensive and versatile ways to understand modern software and cyber physical systems.

The area of formal verification aims to prove the correctness of a system with respect to a specification.
While the formal verification process can provide guarantees on the behavior of a system, such a result alone does not tell much about the inner workings of the system.
To give some additional insight, counterexamples, invariants or related certificates as a form of verifiable justification that a system does or does not behave according to a specification have been studied extensively (see e.g., \cite{MaPn95,CGP99,Namjoshi01}).
These kinds of certificates, however, do not allow us to understand the system behavior to a full extend.
In epistemic terms, the outcome of model checking applied to a system and a specification provides \emph{knowledge that} a system satisfies a specification (or not) in terms of	an assertion (whether the system satisfies the specification) and a justification (certificate or counterexample) to increase the belief in the result.
However, model checking usually does not provide \emph{understanding} on why a system behaves in a certain way. 
This establishes a desideratum for a more comprehensive understanding of \emph{why} a system satisfies or violates a specification.
Explications of the system are needed to assess \emph{how} different components influence its behavior and performance.
Causal relations between occurring events during the execution of a system can constitute a strong tool to form such an understanding.
Moreover, causality is fundamental for determining moral responsibility \cite{ChocklerH04, BrahamvanHees2012} or legal accountability \cite{FeigenbaumHJJWW11}, and ultimately fosters user acceptance through an increased level of transparency \cite{Miller17}.

The majority of prior work in this direction relies on causality notions based on Lewis' counterfactual principle \cite{Lewis1973} which states that the effect would not have occurred if the cause would not have happened.
A prominent formalization of the counterfactual principle is given by Halpern and Pearl \cite{HalpernP2001} via structural equation models.
This inspired formal definitions of causality and related notions of blameworthiness and responsibility in Kripke and game structures as well as reactive systems
(see, e.g., \cite{ChocklerHK2008,BeerBCOT12,Chockler16,YazdanpanahDas16,FriedenbergHalpern19,YazdanpanahDasJamAleLog19,IJCAI21,NorineATVA22}).

A lot of systems are to a certain extend influenced by probabilistic events.
Thus, a branch of formal methods is studying probabilistic models such as Markov chains (MCs) which are purely probabilistic or Markov decision processes (MDPs) which combine non-determinism and probabilistic choice.
This gives rise to another approach to the concept of causality in a probabilistic setting,
since the statement of counterfactual reasoning can be interpreted more gently in a probabilistic setting:
Instead of saying ``the effect would not have occurred, if the cause had not happened'', we can say
``the probability of the effect would have been lower if the cause would not have occurred''.
This interpretation leads to the widely accepted \emph{probability-raising principle} which also has its roots in philosophy \cite{Reichenbach56,Suppes70,Eells91,Hitchcock-Handbook} 
and has been refined by Pearl \cite{Pearl09} for causal and probabilistic reasoning in intelligent systems.
The different notions of probability-raising cause-effect relations discussed in the literature share the following two main principles:
\begin{description}
	\item [(C1)] Causes raise the probabilities for their effects,
	informally expressed by the requirement
	``$\Pr( \, \text{effect} \, | \, \text{cause} \, ) >
	\Pr( \, \text{effect} \, )$''.
	\item [(C2)] Causes must happen before their effects.
\end{description}
Despite the huge amount of work on probabilistic causation in other disciplines, research on probability-raising causality in the context of formal methods is comparably rare and has concentrated on the purely probabilistic setting in Markov chains (see, e.g., \cite{KleinbergM2009,Kleinberg2012,ISSE22} and the discussion of related work below).
To the best of our knowledge, probabilistic causation for probabilistic operational models with nondeterminism has not been studied before.

In this work, we formalize principles \textbf{(C1)} and \textbf{(C2)} for Markov decision processes.
We start in a basic setting by focusing on reachability properties where both effect and cause are sets of states.
Later, we naturally extend this framework by considering the effect to be an $\omega$-regular path property while causes can either still be state-based or $\omega$-regular co-safety path properties.

As the probability-raising is inherent in the MDP, we require \textbf{(C1)} under every scheduler.
Thus, the cause-effect relation holds for every resolution of non-deterministic choices.
We consider two natural ways to interpret condition \textbf{(C1)}:
On one hand, the probability-raising property can be required locally for each element of the cause.
This results in a strict property which requires that after each execution showing the cause the probability of effect has been raised.
Such causes are \emph{strict probability-raising (SPR) causes} in our framework. 
This interpretation is especially suited when the task is to identify system states that have to be avoided for lowering the effect probability.
On the other hand, we can treat the cause globally as a unit in \textbf{(C1)} leading to the notion of \emph{global probability-raising (GPR) cause}.
This way, the causal relation can also be formulated between properties without considering individual elements or executions to be causal.
Considering the cause as a whole is also better suited when further constraints are imposed on the candidates for cause sets.
E.g. if the set of non-terminal states of the given MDP is partitioned into sets of states $S_i$ under control of different agents $i$, $1\leq i \leq k$ and the task is to identify which agent's decisions might cause the effect, only the subsets of $S_1,\ldots,S_k$ are candidates for causes.
Furthermore, GPR causes are more appropriate when causes are used for monitoring purposes under partial observability as then the cause candidates are sets of indistinguishable states.

Despite the distinction between strict and global probability-raising causality, different causes can still vary substantially in how well they predict the effect.
Within Markov decision processes this intuitively coincides with how well the executions exhibiting the cause cover executions showing the effect.
However, solely focusing on broader coverage might lead to more trivial causal relations.
In order to take this trade off into account, we take inspiration from measures for binary classifiers used in statistical analysis (see, e.g., \cite{Powers-fscore}) to introduce quality measures for causes.
These allow us to compare different causes and to look for optimal causes: 
The \emph{recall} captures the probability that the effect is indeed preceded by the cause.
The \emph{coverage-ratio} is the fraction of the probability that cause and effect are observed and the probability of the effect without a cause.
The \emph{f-score}, a widely used quality measure for binary classifiers, is the harmonic mean of recall and precision, where the precision is the probability that the cause is followed by the effect.
Finally, we address the computation of quality measures which can be represented as algebraic functions.

\subsection*{Contributions.}
In this work we build mathematical and algorithmic foundations for probabilistic causation in Markov decision processes based on the principles \textbf{(C1)} and \textbf{(C2)}.
In the setting where the effect is represented as a set of terminal states, we introduce strict and global probability-raising cause sets in MDPs (Section~\ref{sec:SPR-GPR}).
Algorithms are provided to check whether given cause and effect sets satisfy (one of) the probability-raising conditions (Section~\ref{sec:SPR_check} and \ref{sec:check-GPR}) and to check the existence of a cause set for a given effect set (Section~\ref{sec:SPR_check}).
In order to evaluate the coverage properties of a cause, we subsequently introduce the above-mentioned quality measures (Section~\ref{sec:acc-measures}).
We give algorithms for computing these values for given cause-effect sets (Section~\ref{sec:comp-acc-measures-fixed-cause}) and characterize the computational complexity of finding optimal cause sets with respect to the different measures (Section~\ref{sec:opt-PR-causes}).
We then extend the setting to $\omega$-regular effects (Section~\ref{sub:states_as_causes}), and evaluate how established properties transfer to this setting.
We observe that in this extension SPR causes can be viewed as a collection of GPR causes.
Finally we discuss the case where  causes are also  path properties, namely,  $\omega$-regular co-safety properties and investigate how this more general perspective affects cause-effect relations (Section~\ref{sub:co-safety_as_causes}).
Here, the class of potential cause candidates is greatly increased.
Table \ref{tabelle:summary-of-results} summarizes our complexity results.

\newcommand{\polytime}{poly-time}
\newcommand{\polyspace}{poly-space}

\begin{table}[t]
	\begin{center} 
		\caption{Complexity results for MDPs and Markov chains (MC) 
			with fixed effects. The $\omega$-regular effects are given as deterministic Rabin automata, $\omega$-regular co-safety properties as causes are given as deterministic finite automata accepting good prefixes.
		}
		\begin{adjustbox}{max width=\textwidth}
			\begin{tabular}{c||c|c|c|c|c|c}
				\multicolumn{7}{c}{}\\
				\multicolumn{7}{c}{\bf{sets of states as causes and effects}}\\
				\multicolumn{1}{c}{\phantom{ GPR }} &
				\multicolumn{1}{c}{\phantom{ in $\PSPACE$ }} &
				\multicolumn{1}{c}{\phantom{covratio}} &
				\multicolumn{1}{c}{\phantom{ covratio }} &
				\multicolumn{1}{c}{\phantom{covratio}} &
				\multicolumn{1}{c}{\phantom{ recall = covratio }} &
				\multicolumn{1}{c}{\phantom{covratio}}
				\\
				&
				\multicolumn{4}{c|}{for fixed set $\Cause$}
				&
				\multicolumn{2}{c}{find optimal cause}
				\\
				
				&
				
				&
				\multicolumn{3}{c|}{compute quality values}
				\\[-2ex]
				&
				\multicolumn{1}{c}{%
					\begin{tabular}{c} \\[-2.5ex]
						check PR \\[-0.3ex] condition \\[0.5ex] \end{tabular}}
				&
				\multicolumn{3}{c}{}
				&
				\multicolumn{1}{c|}{%
					\begin{tabular}{r} \\[-2.5ex]
						covratio-optimal \\[-0.3ex] = recall-optimal \\[0.5ex] \end{tabular}}
				&
				\multicolumn{1}{c}{}
				\\[-3.8ex]
				&
				&
				\multicolumn{3}{c|}{(recall, covratio, f-score)} 
				&
				\phantom{recall = covratio} & f-score-optimal 
				\\[1ex]
				\hline
				\hline
				SPR &
				$\in \PTIME$ &
				\multicolumn{3}{c|}{\polytime} &
				\polytime
				&
				\begin{tabular}{c}
					\\[-2ex]
					\polyspace  \\
					\polytime \ for MC \\
					threshold problem $\in \NP\cap \coNP$ 
					
					\\[0.5ex]
				\end{tabular}
				\\
				\hline
				GPR &
				\begin{tabular}{c}
					\\[-2ex]
					$\in \coNP$  \\
					and $\in \PTIME$ for MC
					\\[0.5ex]
				\end{tabular}
				&
				\multicolumn{3}{c|}{\polytime} &
				\multicolumn{2}{c}{%
					\begin{tabular}{c}
						\\[-2ex]
						\polyspace \\
						threshold 
						problems $\in \Sigma^P_2$ and $\NP$-hard 
						\\
						and 
						$\NP$-complete for MC
						\\[0ex]
				\end{tabular}} \\
				\multicolumn{7}{c}{}\\
				\multicolumn{7}{c}{\bf{sets of states as causes and $\omega$-regular effects}}\\
				\multicolumn{1}{c}{\phantom{ GPR }} &
				\multicolumn{1}{c}{\phantom{ in $\PSPACE$ }} &
				\multicolumn{1}{c}{\phantom{covratio}} &
				\multicolumn{1}{c}{\phantom{ covratio }} &
				\multicolumn{1}{c}{\phantom{covratio}} &
				\multicolumn{1}{c}{\phantom{ recall = covratio }} &
				\multicolumn{1}{c}{\phantom{covratio}}
				\\
				&
				\multicolumn{4}{c|}{for fixed set $\Cause$}
				&
				\multicolumn{2}{c}{find optimal cause}
				\\
				
				&
				
				&
				\multicolumn{3}{c|}{compute quality values}
				\\[-2ex]
				&
				\multicolumn{1}{c}{%
					\begin{tabular}{c} \\[-2.5ex]
						check PR \\[-0.3ex] condition \\[0.5ex] \end{tabular}}
				&
				\multicolumn{3}{c}{}
				&
				\multicolumn{1}{c|}{%
					\begin{tabular}{r} \\[-2.5ex]
						covratio-optimal \\[-0.3ex] = recall-optimal \\[0.5ex] \end{tabular}}
				&
				\multicolumn{1}{c}{}
				\\[-3.8ex]
				&
				&
				\multicolumn{3}{c|}{(recall, covratio, f-score)} 
				&
				\phantom{recall = covratio} & f-score-optimal 
				\\[1ex]
				\hline
				\hline
				SPR &
				\begin{tabular}{c}
					$\in \coNP$  \\
					and $\in \PTIME$ for MC 
				\end{tabular}
				&
				\multicolumn{3}{c|}{\polytime} &%
				\begin{tabular}{c}
					\\[-2ex]
					$\PFNP$ (as def. in \cite{Selman1994})\\
					threshold 
					problems $\in \coNP$ 
				\end{tabular} 
				&%
				\begin{tabular}{c}
					\\[-2ex]
					\polyspace \\
					threshold 
					problem $\in \Sigma^P_2$ 
				\end{tabular}   \\
				\hline
				GPR &
				\begin{tabular}{c}
					\\[-2ex]
					$\in \coNP$  \\
					and $\in \PTIME$ for MC 
					\\[0.5ex]
				\end{tabular}
				&
				\multicolumn{3}{c|}{\polytime} &
				\multicolumn{2}{c}{%
					\begin{tabular}{c}
						\\[-2ex]
						\polyspace \\
						threshold 
						problems $\in \Sigma^P_2$ \\
						$\NP$-hard  
				\end{tabular}}   \\
				\multicolumn{7}{c}{}\\
				\multicolumn{7}{c}{\bf{$\omega$-regular co-safety properties as causes and $\omega$-regular effects}}\\
				\multicolumn{1}{c}{\phantom{ GPR }} &
				\multicolumn{1}{c}{\phantom{ in $\PSPACE$ }} &
				\multicolumn{1}{c}{\phantom{covratio}} &
				\multicolumn{1}{c}{\phantom{ covratio }} &
				\multicolumn{1}{c}{\phantom{covratio}} &
				\multicolumn{1}{c}{\phantom{ recall = covratio }} &
				\multicolumn{1}{c}{\phantom{covratio}}
				\\
				&
				\multicolumn{4}{c|}{for fixed cause property $\rCause$}
				&
				\multicolumn{2}{c}{find optimal cause}
				\\
				
				&
				
				&
				\multicolumn{3}{c|}{compute quality values}
				\\[-2ex]
				&
				\multicolumn{1}{c}{%
					\begin{tabular}{c} \\[-2.5ex]
						check PR \\[-0.3ex] condition \\[0.5ex] \end{tabular}}
				&
				\multicolumn{3}{c}{}
				&
				\multicolumn{1}{c|}{%
					\begin{tabular}{r} \\[-2.5ex]
						covratio-optimal \\[-0.3ex] = recall-optimal \\[0.5ex] \end{tabular}}
				&
				\multicolumn{1}{c}{}
				\\[-3.8ex]
				&
				&
				\multicolumn{3}{c|}{(recall, covratio, f-score)} 
				&
				\phantom{recall = covratio} & f-score-optimal 
				\\[1ex]
				\hline
				\hline
				SPR &
				\begin{tabular}{c}
					difficulty illustrated \\
					in Example \ref{ex:no-co-safety-SPR-check}
				\end{tabular} &
				\multicolumn{3}{c|}{\polytime} &
				\begin{tabular}{c}
					optimal cause known,\\
					computation unclear 
				\end{tabular}
				&
				open
				\\
				\hline
				GPR &
				\begin{tabular}{c}
					\\[-2ex]
					$\in \coNP$  \\
					and $\in \PTIME$ for MC
					\\[0.5ex]
				\end{tabular}
				&
				\multicolumn{3}{c|}{\polytime} &
				\begin{tabular}{c}
					in general, \\
					no
					optimal causes
				\end{tabular}
				& open  
			\end{tabular}
		\end{adjustbox}
	\end{center}
	\label{tabelle:summary-of-results}
\end{table}


\subsection*{Related work}

\label{sec:related-work}

Previous work in the direction of probabilistic causation in stochastic operational models has mainly concentrated on Markov chains.
Kleinberg \cite{KleinbergM2009,Kleinberg2012} introduced \emph{prima facie causes} in finite Markov chains where both causes and effects are formalized as PCTL state formulae, and thus they can be seen as sets of states as in our approach.
The correspondence of Kleinberg's PCTL constraints for prima facie causes and the strict probability-raising condition formalized using conditional probabilities has been worked out in the survey article \cite{ICALP21}.
Our notion of strict probability-raising causes interpreted in Markov chains corresponds to Kleinberg's prima facie causes with the exception of a minimality condition forbidding redundant elements in our definition.
\'{A}brah\'{a}m et al \cite{Abraham-QEST2018} introduces a hyperlogic for Markov chains and gives a formalization of probabilistic causation in Markov chains as a hyperproperty, which is consistent with Kleinberg's prima facie causes, and with strict probability-raising causes up to minimality.
Cause-effect relations in Markov chains where effects are $\omega$-regular properties and the causes are sets of paths have been introduced in \cite{ISSE22}.
These relations rely on a strict probability-raising condition, but use a probability threshold $p$ instead of directly requiring probability-raising.
Therefore, \cite{ISSE22} permits a non-strict inequality in the PR condition with the consequence that causes always exist, which is not the case for our notions.
However, a minimal good prefix of a co-safety strict probability-raising cause in a Markov chains corresponds to a probability-raising path in \cite{ISSE22}.

The survey article \cite{ICALP21} introduces notions of global probability-raising causes for Markov chains, where causes and effects can be path properties.
\cite{ICALP21}'s notion of \emph{reachability causes} in Markov chains directly corresponds to our notion of PR causes, the only difference being that \cite{ICALP21} deals with a relaxed minimality condition and requires that the cause set is reachable without visiting an effect state before. 
The latter is inherent in our approach as we suppose that all states are reachable and the effect states are terminal.
On the other hand if we restrict \cite{ICALP21}'s notion of global PR-cause to $\omega$-regular effects and co-safety causes, this corresponds to our notion of a co-safety GPR cause with the exception of the minimality condition.
The same can be said about the correspondence of local PR-causes from \cite{ICALP21} and co-safety SPR causes.

To the best of our knowledge, probabilistic causation in MDPs has not been studied before.
The only work in this direction we are aware of is the recent paper by Dimitrova et al \cite{DimFinkbeinerTorfah-ATVA2020} on a hyperlogic, called PHL, for MDPs.
While the paper focuses on the foundation of PHL, it contains an example illustrating how action causality can be formalized as a PHL formula.
Roughly, the presented formula expresses that taking a specific action $\alpha$ increases the probability for reaching effect states.
Thus, it also relies on the probability-raising principle, but compares the ``effect probabilities'' under different schedulers (which either schedule $\alpha$ or not) rather than comparing probabilities under the same scheduler as in our PR condition.
However, to some extent our notions of PR causes can reason about action causality as well.

There has also been work on causality-based explanations of counterexamples in probabilistic models \cite{LeitnerFischer-Leue-SAFECOMP11,LeitnerFischer-PhD15}.
The underlying causality notion of this work, however, relies on the non-probabilistic counterfactual principle rather than the probability-raising condition.
The same applies to the notions of forward and backward responsibility in stochastic games in extensive form introduced in the recent work \cite{IJCAI21}. 


\section{Preliminaries}

\label{sec:prelim}
Throughout the paper, we will assume familiarity with basic concepts of Markov decision processes. 
Here, we present a brief summary of the notations used in the paper.
For more details, we refer to \cite{Puterman,BaierK2008,Kallenberg20}.

\subsection{Markov decision processes}
A \emph{Markov decision process (MDP)} constitutes a tuple $\cM =(S, \Act,P,\init)$ where
$S$ is a finite set of states,
$\Act$  a finite set of actions,
$\init \in S$  a unique initial state and
$P : S \times \Act \times S \to [0,1]$ the transition probability function such that $\sum_{t\in S}P(s, \alpha,t) \in \{0,1 \}$ for all states $s \in S$ and actions $\alpha \in \Act$.

For $\alpha\in \Act$ and $T\subseteq S$, $P(s,\alpha,T)$ is a shortform notation for $\sum_{t\in t} P(s,\alpha,t)$.
An action $\alpha$ is \emph{enabled} in state $s \in S$ if $\sum_{t\in S}P(s, \alpha,t)=1$. 
Then, $\Act(s) = \{\alpha \in \Act \mid \alpha \text{ is enabled in } s\}$.
A state $t$ is \emph{terminal} if $\Act(t) = \emptyset$.
A Markov chain (MC) is a special case of MDP where $\Act$ is a singleton (we then write $P(s,t)$ rather than $P(s,\alpha,t)$).
A \emph{path} in an MDP $\cM$ is a (finite or infinite) alternating sequence $\pi=s_0 \, \alpha_0 \, s_1 \, \alpha_1 \, s_2 \dots \in (S \times \Act)^*\times S \cup (S \times \Act)^\omega$ such that  $P(s_i, \alpha_{i},s_{i+1})>0$ for all indices $i$.
A path is called maximal if it is infinite or finite and ends in a terminal state.
If $\pi$ is a finite path in $\cM$ then $\last(\pi)$ denotes the last state of $\pi$.
That is, if $\pi = s_0\, \alpha_0 \ldots \alpha_{n-1}\, s_n$ then $\last(\pi)=s_n$.

A \emph{(randomized) scheduler} $\sched$ is a function that maps each finite non-maximal path $s_0 \alpha_0  \dots s_n$ to a distribution over $\Act(s_n)$. 
$\sched$ is called deterministic if $\sched(\pi)$ is a Dirac distribution for all finite non-maximal paths $\pi$.
If the chosen action only depends on the last state of the path, $\sched$ is called \emph{memoryless}. 
We write MR for the class of memoryless (randomized)  and MD for the class of memoryless deterministic schedulers. 
\emph{Finite-memory} schedulers are those that are representable by a finite-state automaton.
A path $\pi$ is said to be a $\sched$-path if $\sched(s_0\, \alpha_0 \ldots \alpha_{i-1} \, s_i)(\alpha_i) >0$ for each $i\in \{0,\ldots,n{-}1\}$. 
Given a path $\pi=s_0\, \alpha_0 \, \dots \, \alpha_{n-1} \, s_n$, the \emph{residual scheduler} $\residual{\sched}{\pi}$ of $\sched$ after $\pi$ is defined by $\residual{\sched}{\pi}(\zeta) = \sched (\pi \circ \zeta)$ for all finite paths $\zeta$ starting in $s_n$.
Here, $\pi \circ \zeta$ denotes the concatenation of the paths $\pi$ and $\zeta$.
Intuitively speaking, $\residual{\sched}{\pi}$ behaves like $\sched$ after $\pi$ has already been seen.

A scheduler $\sched$ of $\cM$ induces a (possibly infinite) Markov chain.
We write $\Pr^{\sched}_{\cM,s}$ for the standard probability measure on measurable sets of maximal paths in the Markov chain induced by $\sched$ with initial state $s$.
If $\varphi$ is a measurable set of maximal paths, then $\Pr^{\max}_{\cM,s}(\varphi)$ and $\Pr^{\min}_{\cM,s}(\varphi)$ denote the supremum resp. infimum of the probabilities for $\varphi$ under all schedulers.
We use the abbreviation $\Pr^{\sched}_{\cM}=\Pr^{\sched}_{\cM,\init}$ and notations $\Pr^{\max}_{\cM}$ and $\Pr^{\min}_{\cM}$ for extremal probabilities.
Analogous notations will be used for expectations.
So, if $f$ is a random variable, then, e.g., $\mathrm{E}^{\sched}_{\cM}(f)$ denotes the expectation of $f$ under $\sched$ and $\mathrm{E}^{\max}_{\cM}(f)$ its supremum over all schedulers.

We use LTL-like temporal modalities such as $\Diamond$ (eventually) and $\Until$ (until) to denote path properties.
For $X,T \subseteq S$ the formula $X \Until T$ is satisfied by paths $\pi = s_0 s_1 \dots $ such that there exists $j \geq 0$ such that for all $i<j: s_i \in X$ and $s_j \in T$
and $\Diamond T = S \Until T$.
It is well-known that $\Pr^{\min}_{\cM}(X \Until T)$ and
$\Pr^{\max}_{\cM}(X \Until T)$ and corresponding optimal MD-schedulers are computable in polynomial time.

We also use conditional probabilities in MDPs cf. \cite{TACAS14-condprob,Maercker-PhD20}.
For two measurable sets of maximal paths $\varphi$ and $\psi$ we have 
\[\Pr^{\max}_{\cM,s}(\varphi \mid \psi) = \max_{\sched}\ \Pr^\sched_{\cM,s}(\varphi \mid \psi) = \max_{\sched}\ \frac{\Pr^\sched_{\cM,s}(\varphi \wedge \psi)}{\Pr^\sched_{\cM,s}(\psi)},\]
where $\sched$ ranges over all scheduler for which $\Pr^\sched(\psi) > 0$.
We define $\Pr^{\min}_{\cM,s}(\varphi \mid \psi)$ likewise.
If both $\varphi$ and $\psi$ are reachability properties then maximal conditional probabilities are computable in polynomial time \cite{TACAS14-condprob}.
The proposed algorithm for maximal conditional probabilities $\Pr^{\max}_{\cM,\init}(\lozenge G \mid \lozenge F)$ relies on a model transformation generating a new MDP $\cN$ that distinguishes the modes "before $G$ and $F$" (where $\cN$ essentially behaves as $\cM$ with additional reset transitions from end components to the initial state), "before  $G$,  after $F$" (where $\cN$ maximizes the probability to reach $G$), "before $F$, after $G$"  (where $\cN$ maximizes the probability to reach $F$).
Essentially the reset transitions serve to ``discard'' paths that never reach $G$ and $F$.
For further details we refer to \cite{TACAS14-condprob, MBKK17, Maercker-PhD20}.

If $s\in S$ and $\alpha\in \Act(s)$, then $(s,\alpha)$ is said to be a state-action pair of $\cM$.
Given a scheduler $\sched$ for $\cM$, the expected frequencies
(i.e., expected number of occurrences in maximal paths) of 
state action-pairs $(s, \alpha)$, states $s\in S$ and state-sets $T \subseteq S$
under $\sched$ are defined by:
\begin{align*}
	\freq{\sched}{s,\alpha} & \ \ \eqdef \ \
	\mathrm{E}^{\sched}_{\cM}(\text{number of visits to $s$ in which $\alpha$ is taken}),\\
	\freq{\sched}{s} &\ \ \eqdef
	\sum\nolimits_{\alpha \in \Act(s)} \freq{\sched}{s,\alpha},
	\qquad	\freq{\sched}{T} \eqdef \sum\nolimits_{s \in T} \freq{\sched}{s}.
\end{align*}

An \emph{end component} (EC) of an MDP $\cM$ is a strongly connected sub-MDP containing at least one state-action pair. ECs will be often identified with the set of their state-action pairs. An EC $\cE$ is called maximal (abbreviated MEC) if there is no proper superset $\cE'$ of (the set of state-action pairs of) $\cE$ which is an EC.

\subsection{MR-scheduler in MDPs without ECs}
The following preliminary lemma is folklore (see, e.g., \cite[Theorem 9.16]{Kallenberg20}) and used in the paper in the following form.

\begin{lem}[From general schedulers to MR-schedulers in MDPs without ECs]
	\label{lem:from-general-to-MR-schedulers}
	Consider an MDP  $\cM=(S,\Act,P,\init)$ without end components.
	Then, for each scheduler  $\sched$ for $\cM$,
	there exists an MR-scheduler $\tsched$ such that:
	\begin{center}
		$\Pr^{\sched}_{\cM}(\Diamond t) \ = \
		\Pr^{\tsched}_{\cM}(\Diamond t)$
		\ \ for each terminal state $t$.
	\end{center}
\end{lem}  

As a consequence we can build linear combinations of scheduler in such MDPs.

\begin{lem}[Convex combination of MR-schedulers]
	\label{lem:convex}
	Let $\cM$ be an MDP without end components and let
	$\sched$ and $\tsched$ be schedulers for $\cM$ and
	$\lambda$ a real number in the open interval $]0,1[$.
	Then, there exists an MR-scheduler $\usched$
	such that:
	\begin{center}
		$\Pr^{\usched}_{\cM}(\Diamond t) \ = \
		\lambda \cdot \Pr^{\sched}_{\cM}(\Diamond t) \ + \
		(1{-}\lambda) \cdot \Pr^{\tsched}_{\cM}(\Diamond t)$ 
	\end{center}
	for each terminal state $t$.
\end{lem}

\begin{proof}
	Thanks to Lemma \ref{lem:from-general-to-MR-schedulers}
	we may suppose that $\sched$ and $\tsched$ are MR-schedulers.
	Let
	\[
	f_* \ \ = \ \
	\lambda \cdot \freq{\sched}{*} \ + \ (1{-}\lambda)\cdot \freq{\tsched}{*}
	\]
	where $*$ stands for a state or a state-action pair in $\cM$.
	Let $\usched$ be an MR-scheduler
	defined by
	$
	\usched(s)(\alpha)
	= \frac{f_{s,\alpha}}{f_s}
	$
	for each non-terminal state $s$ where $f_s>0$ and
	each action $\alpha \in \Act(s)$.
	If $f_s=0$ then $\usched$ selects an arbitrary distribution
	over $\Act(s)$.
	
	Using 
	Lemma \ref{lem:from-general-to-MR-schedulers}
	we then obtain
	$f_* = \freq{\usched}{*}$ where $*$ ranges over all states
	and state-action pairs in $\cM$.
	But this yields:
	\begin{eqnarray*}
		\Pr^{\usched}_{\cM}(\Diamond t) & = & f_t \ = \ 
		\lambda \cdot \freq{\sched}{t} +
		(1{-}\lambda) \cdot \freq{\tsched}{t}
		\ = \
		\lambda \cdot \Pr^{\sched}_{\cM}(\Diamond t) +
		(1{-}\lambda) \cdot \Pr^{\tsched}_{\cM}(\Diamond t)
	\end{eqnarray*}
	for each terminal state $t$.
\end{proof}
Let $\cM,\sched,\tsched,\lambda$ be as in Lemma \ref{lem:convex}.
Then, the notation $\lambda \sched \oplus (1{-}\lambda)\tsched$	will be used to denote any MR-scheduler $\usched$ as in	Lemma \ref{lem:convex}.

\subsection{MEC-quotient}\label{app:MEC-quotient}

We recall the definition of the MEC-quotient, which is a standard concept for the analysis of MDPs \cite{deAlfaro1997}.
Intuitively, the MEC-quotient of an MDP collapses all maximal end components, ignoring the actions in the end component while keeping outgoing transitions.
More concretely, we use a modified version with an additional trap state as in \cite{LICS18-SSP} that serves to mimic behaviors inside an end component of the original MDP.

\begin{defi}[MEC-quotient of an MDP]
	\label{def:MEC-contraction} 
	Let $\cM = (S,\Act,P,\init)$ be an MDP with end components.  
	Let $\cE_1,\ldots,\cE_k$ be the maximal end components (MECs) of $\cM$.
	We may suppose without loss of generality that enabled actions of states 
	are pairwise disjoint, i.e., whenever $s_1,s_2$ are states
	in $\cM$ with $s_1\not= s_2$ then
	$\Act_{\cM}(s_1)\cap \Act_{\cM}(s_2)=\varnothing$.
	This permits to consider $\cE_i$ as a subset of $\Act$.
	Let $U_i$ denote the set of states that belong to $\cE_i$
	and let $U=U_1\cup \ldots \cup U_k$.
	The \emph{MEC-quotient of $\cM$} is the MDP $\cN = (S',\Act',P',\init')$
	and the function $\iota : S \to S'$ are defined as follows.
	\begin{itemize}
		\item
		The state space $S'$ is $S \setminus U \cup \{s_{\cE_1},\ldots,s_{\cE_k},\bot\}$
		where $s_{\cE_1},\ldots,s_{\cE_k},\bot$ are pairwise distinct fresh states.
		\item
		The function $\iota$ is given by $\iota(s)=s$ if $s\in S \setminus U$
		and $\iota(u)=s_{\cE_i}$ if $u\in U_i$.
		\item
		The initial state of $\cN$ is $\init' = \iota(\init)$.
		\item
		The action set $\Act'$ is $\Act \cup \{\tau\}$ where
		$\tau$ is a fresh action symbol.
		\item
		The set of actions enabled in state $s\in S'$ of $\cN$
		and the transition probabilities are as follows:
		\begin{itemize}
			\item
			If $s$ is a state of $\cM$ that does not belong to an MEC of $\cM$
			(i.e., $s\in S \cap S'$) then 
			$\Act_{\cN}(s)=\Act_{\cM}(s)$ and
			$P'(s,\alpha,s')=P(s,\alpha,\iota^{-1}(s'))$
			for all $s'\in S'$ and $\alpha \in \Act_{\cM}(s)$.
			\item
			If $s=s_{\cE_i}$ is a state representing MEC $\cE_i$ of $\cM$
			then (as we view $\cE_i$ as a set of actions):   
			\[
			\Act_{\cN}\bigl(s_{\cE_i}\bigr) \ = \
			\bigcup_{u\in U_i} (\Act_{\cM}(u) \setminus \cE_i)
			\cup \{\tau\}
			\]
			The $\tau$-action is the deterministic transition to
			the fresh state $\bot$, i.e.:
			$
			P'(s_{\cE_i},\tau,\bot)=1.
			$
			For $u\in U_i$ and $\alpha \in \Act_{\cM}(u) \setminus \cE_i$ we set $P'(s_{\cE_i},\alpha,s')=P(u,\alpha,\iota^{-1}(s'))$
			for all $s'\in S'$. 
			\item
			The state $\bot$ is terminal, i.e.,
			$\Act_{\cN}(\bot)=\varnothing$.\Ende
		\end{itemize}
	\end{itemize}
\end{defi}

Thus, each terminal state of $\cM$ is terminal in its MEC-quotient $\cN$ too.
Vice versa, every terminal state of $\cN$ is either a terminal state of
$\cM$ or $\bot$.
Moreover, $\cN$ has no end components, which implies that
under every scheduler
$\tsched$ for $\cN$, a terminal state will be reached with probability 1.
In Section~\ref{sec:check-GPR}, we use the notation $\noeffbot$ rather than $\bot$.

The original MDP and its MEC-quotient have been found to be connected by the following lemma (see also \cite{deAlfaro1997, Alfaro-CONCUR99}).
For the sake of completeness we present the proof for our version of the MEC-quotient.

\begin{lem}[Correspondence of an MDP and its MEC-quotient]
	\label{lem:MEC-contraction-terminal-states-final} 
	Let $\cM$ be an MDP and $\cN$ its MEC-quotient.
	Then, for each scheduler $\sched$ for $\cM$ there is a scheduler
	$\tsched$ for $\cN$ such that
	\begin{align}
		\label{contraction-property}
		\text{$\Pr^{\sched}_{\cM}(\Diamond t) \ = \
			\Pr^{\tsched}_{\cN}(\Diamond t)$ \
			for each terminal state $t$ of $\cM$}  
	\end{align}  
	and vice versa.
	Moreover, if \eqref{contraction-property} holds then
	$\Pr^{\tsched}_{\cN}(\Diamond \bot)$ equals the probability
	for $\sched$ to generate an infinite path in $\cM$ that
	eventually enters and stays forever in an end component.
\end{lem}  

\begin{proof}
	Given a scheduler $\tsched$ for $\cN$, we pick an
	MD-scheduler $\usched$ such that $\usched(u) \in \cE_i$ for each $u\in U_i$.
	Then, the corresponding scheduler $\sched$ for $\cM$ behaves as
	$\tsched$ as long as $\tsched$ does not choose the
	$\tau$-transition to $\bot$. As soon as $\tsched$ schedules $\tau$
	then $\sched$ behaves as $\usched$ from this moment on.
	
	Vice versa, given a scheduler $\sched$ for $\cM$ then
	a corresponding scheduler $\tsched$ for $\cN$ mimics
	$\sched$ as long as $\sched$ has not visited a state
	belong to an end component $\cE_i$ of $\cM$.
	Scheduler $\tsched$ ignores $\sched$'s transitions inside
	an MEC $\cE_i$ and takes
	$\beta \in \bigcup_{u\in U_i} (\Act_{\cM}(u) \setminus \cE_i)$
	with the same probability
	as $\sched$ leaves $\cE_i$.
	With the remaining probability mass, $\sched$ stays forever inside $\cE_i$,
	which is mimicked by $\tsched$ by taking the $\tau$-transition to $\bot$.
	
	For the formal definition of $\tsched$, we use the following notation.
	For simplicity, let us assume that $\init \notin U_1 \cup \ldots \cup U_k$.
	This yields $\init= \init'$.
	Consider a finite path $\pi= s_0\, \alpha_0 \, s_1 \, \alpha_1 \ldots \alpha_{m-1} \, s_m$
	in $\cM$ with $s_0=\init$, then $\pi_{\cN}$ is the path in $\cN$ resulting from
	by replacing each maximal path fragment
	$s_h \alpha_h \ldots  \alpha_{j-1} s_j$
	consisting of actions inside an $\cE_i$ with state $s_{\cE_i}$.
	(Here, maximality means if $h > 0$ then $\alpha_{h-1}\notin \cE_i$
	and if $j<m$ then $\alpha_{j+1}\notin \cE_i$.)
	Furthermore, let $p_{\pi}^{\sched}$ denote the probability
	for $\sched$ to generate the path $\pi$ when starting in the
	first state of $\pi$.

	Let $\rho$ be a finite path in $\cN$ with first state $\init$ (recall that we suppose that $\cM$'s initial state
	does not belong to an MEC, which yields $\init=\init'$)
	and $\last(\rho)\not= \bot$.
	Then, $\Pi_{\rho}$ denotes the set of finite paths
	$\pi= s_0\, \alpha_0 \, s_1 \, \alpha_1 \ldots \alpha_{m-1} \, s_m$
	in $\cM$ such that
	(i) $\pi_{\cN}=\rho$ and
	(ii) if $s_m \in U_i$ then $\alpha_{m-1}\notin \cE_i$.
	The formal definition of scheduler $\tsched$ is as follows.
	Let $\rho$ be a finite path in $\cN$ where the last state $s$ of $\rho$ is
	non-terminal.
	If $s$ is a state of $\cM$ and does not belong to an MEC of $\cM$
	and $\beta \in \Act_{\cM}(s)$
	then:
	\[
	\tsched(\rho)(\beta)
	\ \ = \ \ \sum_{\pi\in \Pi_{\rho}} p_{\pi}^{\sched} \cdot \sched(\pi)(\beta)
	\]
	If $s = s_{\cE_i}$ and $\beta \in \Act_{\cN}\bigl(s_{\cE_i}\bigr)\setminus \{\tau\}$
	then
	\[
	\tsched(\rho)(\beta)
	\ \ = \ \
	\sum_{\pi\in \Pi_{\rho}} p_{\pi}^{\sched} \cdot
	\Pr^{\residual{\sched}{\pi}}_{\cM,\last(\pi)}\bigl( \
	\text{``leave $\cE_i$ via action $\beta$''} \ \bigr) 
	\]
	where ``leave $\cE_i$ via action $\beta$'' means the existence of a prefix
	whose action sequence consists of actions inside $\cE_i$ followed by action
	$\beta$. The last state of this prefix, however,
	could be a state of $U_i$. (Note $\beta \in \Act_{\cN}(s_{\cE_i})$ means that
	$\beta$ could have reached a state outside $U_i$, but
	there might be states inside $U_i$ that are accessible via $\beta$.)
	Similarly,  
	\[
	\tsched(\rho)(\tau)
	\ \ = \ \
	\sum_{\pi\in \Pi_{\rho}} p_{\pi}^{\sched} \cdot
	\Pr^{\residual{\sched}{\pi}}_{\cM,\last(\pi)}\bigl( \
	\text{``stay forever in $\cE_i$''} \ \bigr)
	\]
	where ``stay forever in $\cE_i$'' means that only actions
	inside $\cE_i$ are performed.
	By induction on the length of $\rho$ we obtain: 
	\[
	p_{\rho}^{\tsched} \ \ = \ \ \sum_{\pi\in \Pi_{\rho}} p_{\pi}^{\sched}
	\]
	But this yields 
	$\Pr^{\sched}_{\cM}(\Diamond t)$ = $\Pr^{\tsched}_{\cN}(\Diamond t)$
	for each terminal state $t$ of $\cM$.
	Moreover the probability under $\sched$ to eventually enter and stay forever in $\cE_i$
	equals the probability for $\tsched$ to
	reach the terminal state $\bot$ via a path of the form
	$\rho \, \tau \, \bot$ where $\last(\rho)=s_{\cE_i}$.
\end{proof}

\subsection{Automata and \texorpdfstring{$\omega$}{ω}-regular languages}
In order have a representation of an $\omega$-regular language, we use \emph{deterministic Rabin automata} (DRA).
A DRA constitutes a tuple $\mathcal{A} = (Q, \Sigma, q_0, \delta, \mathrm{Acc})$ where $Q$ is a finite set of states, $\Sigma$ an alphabet, $q_0$ the initial state, $\delta: Q \times \Sigma \to Q$ the transition function and $\mathrm{Acc} \subseteq 2^Q \times 2^Q$ the acceptance set.
The \emph{run} of $\mathcal{A}$ on a word $w=w_0 w_1 \dots \in \Sigma^\omega$ is the sequence $\tau = q_0 q_1 \ldots$ of states such that $\delta(q_i,w_i)=q_{i+1}$ for all $i$.
It is \emph{accepting} if there exists a pair $(L,K) \in \mathrm{Acc}$ such that $L$ is only visited finitely often and $K$ is visited infinitely often by $\tau$.
The language $\mathcal{L}(\mathcal{A})$ is the set of all words $w \in \Sigma^\omega$ on which the run of $\mathcal{A}$ is accepting.

A \emph{good prefix} $\pi$ for an $\omega$-regular language $\mathcal{L}$ is a finite word such that all infinite extensions of $\pi$ belong to $\mathcal{L}$.
An $\omega$-regular language $\mathcal{L}$ is called a  \emph{co-safety language} if all words in $\mathcal{L}$ have a prefix that is a good prefix for $\mathcal{L}$.
A {co-safety language} $\mathcal{L}$  is uniquely determined by the regular set of minimal good prefixes of words in $\mathcal{L}$.

The regular language of minimal good prefixes of a co-safety $\mathcal{L}$ which uniquely determines $\mathcal{L}$ can be represented by a \emph{deterministic finite automaton} (DFA).
A DFA constitutes a tuple $\mathcal{A} = (Q, \Sigma, q_0, \delta, \mathrm{Acc})$ where $Q$ is a finite set of states, $\Sigma$ an alphabet, $q_0$ the initial state, $\delta: Q \times \Sigma \to Q$ the transition function and $\mathrm{Acc} \subseteq Q$ the acceptance set.
The \emph{run} of $\mathcal{A}$ on a finite word $w=w_0 w_1 \dots w_n$ is the sequence $\tau = q_0 q_1 \ldots q_n$ of states such that $\delta(q_i,w_i)=q_{i+1}$ for all $i$.
It is \emph{accepting} if $q_n \in \mathrm{Acc}$.
The language $\mathcal{L}(\mathcal{A})$ is the set of all words $w \in \Sigma^*$ on which the run of $\mathcal{A}$ is accepting.

Given an MDP $\cM = (S, \Act, P, \init)$ and a DFA $\cA = (Q, \Sigma, q_0, \delta, \Acc)$ with $\Sigma \subseteq S^*$ we define the \emph{product MDP} $\cM \otimes \cA = (S \times Q, \Act, P', \init')$ with $P'(<\!\! s,q \!\!>, \alpha, <\!\! t,r \!\!>) = P(s, \alpha, t)$ if $r = \delta(q, s)$ and $0$ otherwise, and $\init' = \delta(q_0, \init)$.
The same construction works for the product of an MDP with a DRA.
The difference comes from the acceptance condition encoded in the second components of states of the product MDP.


\section{Strict and global probability-raising causes}

\label{sec:SPR-GPR}

Our contribution starts by providing novel formal definitions for cause-effect relations
in MDPs which rely on the probability-raising (PR) principle
$P(E \mid C) > P(C)$
\textbf{(C1)} which states that the probability of the effect is higher after the cause.
Additionally, we include temporal priority of causes \textbf{(C2)}, stating that causes must happen before the effect.
Here, we focus on the case where both causes and effects
are state properties, i.e., sets of states.

In the sequel, let
$\cM = (S,\Act,P,\init)$ be an MDP
and $\Effect \subseteq S \setminus \{\init\}$ 
a nonempty set of terminal states.
As the effect set is fixed, the assumption that all effect states are terminal contributes to the temporal priority \textbf{(C2)}.
We may also assume that every state $s\in S$ is reachable from $\init$.

We consider two variants of the probability-raising condition: the global setting treats the set $\Cause$ as a unit, while the strict view requires \textbf{(C1)} for all states in $\Cause$ individually.

\begin{defi}[Global and strict probability-raising cause (GPR/SPR cause)]
	\label{def:GPR} 
	\label{def-PR-causes} 
	\label{def:causes}
	Let $\cM$ and $\Effect$ be as above and
	$\Cause$ a nonempty subset of $S \setminus \Effect$ such that for each $c\in \Cause$ we have $\Pr^{\max}_{\cM}( (\neg \Cause) \Until c ) >0$.
	Then, $\Cause$ is said to be a
	\emph{GPR cause} for
	$\Effect$ iff the following condition (G) holds:
	\begin{description}
		\item [(G)]
		For each scheduler $\sched$ where
		$\Pr^{\sched}_{\cM}( \Diamond \Cause) >0$:
		\begin{align*}
			\label{GPR}  
			\Pr^{\sched}_{\cM}(\ \Diamond \Effect \ | \ \Diamond \Cause \ )
			\ > \ \Pr^{\sched}_{\cM}(\Diamond \Effect).
			\tag{\text{GPR}}
		\end{align*}
	\end{description}
	$\Cause$ is called
	an \emph{SPR cause} for
	$\Effect$ iff the following condition (S) holds:
	\begin{description}
		\item [(S)]
		For each state $c\in \Cause$ and each
		scheduler $\sched$ where $\Pr^{\sched}_{\cM}( (\neg \Cause) \Until c ) >0$:
		\begin{align*}
			\label{SPR}  
			\Pr^{\sched}_{\cM}(\ \Diamond \Effect \ | \
			(\neg \Cause) \Until c \ )
			\ > \ \Pr^{\sched}_{\cM}(\Diamond \Effect).
			\tag{\text{SPR}}
		\end{align*}   
	\end{description}
\end{defi}

Note that we only consider sets $\Cause$ as PR cause when each state in $c\in \Cause$ is accessible from $\init$ without traversing other states in $\Cause$.
This can be seen as a minimality condition ensuring that a cause does not contain redundant elements.
However, we could omit this requirement 
without affecting the covered effects (events where an effect state is reached after visiting a cause state) or uncovered effects (events where an effect state is reached without visiting a cause state before).
More concretely, whenever a set $C \subseteq S \setminus \Effect$ satisfies conditions (G) or (S) then the set of states $c\in C$ where $\cM$ has a path from $\init$ satisfying $(\neg C)\Until c$ is a GPR resp. an SPR cause.
On the other hand the set $\Cause$ is disjoint of the effect $\Eff$ to ensure temporal priority \textbf{(C2)}.


\subsection{Examples and simple properties of probability-raising causes}
\label{sec: simple properties}

We first observe that SPR and GPR causes cannot contain the initial state 
$\init$, since otherwise an equality instead of an inequality would hold in \eqref{GPR}
and \eqref{SPR}.
Furthermore as a direct consequence of the definitions and using the equivalence of
the LTL formulas
$\Diamond \Cause$ and
$(\neg \Cause) \until \Cause$
we obtain:

\begin{lem}[Singleton PR causes]
	\label{global-vs-strict-causes}
	If $\Cause$ is a singleton then
	$\Cause$ is a SPR cause for $\Effect$
	if and only if
	$\Cause$ is a GPR cause for $\Effect$.
\end{lem}

The direction from SPR cause to GPR cause even holds in general as the event $\Diamond \Cause$ can be expressed as a disjoint union of all events $(\neg \Cause) \Until c$ where $c\in \Cause$.
Therefore, the probability for covered effects
$\Pr^{\sched}_{\cM}(\ \Diamond \Effect \ | \ \Diamond \Cause \ )$
is a weighted average of the probabilities
$\Pr^{\sched}_{\cM}(\ \Diamond \Effect \ | \ (\neg \Cause)\Until c \ )$
for $c\in \Cause$, which yields:

\begin{lem}
	\label{lemma:strict-implies-global} 
	Every SPR cause for $\Effect$ is a GPR cause for $\Effect$.
\end{lem}
\begin{proof}
	Assume that $\Cause$ is a SPR cause for $\Effect$ in $\cM$ and let $\sched$ be a scheduler that reaches $\Cause$ with positive probability.
	Further, let 
	\[
	C_{\sched}\eqdef \{c\in \Cause\ \mid \ \Pr^{\sched}_{\cM}((\neg \Cause )\Until c)>0\} 
	\quad \text{and} \quad
	m\eqdef \min_{c\in C_{\sched}} \Pr^{\sched}_{\cM}(\ \Diamond \Effect \ | \
	(\neg \Cause) \Until c \ ).
	\]
	As $\Cause$ is a SPR cause, $ m> \Pr^{\sched}_{\cM}(\Diamond \Effect)$.
	The set of $\sched$-paths satisfying $\Diamond \Cause$ is the disjoint union of the sets of $\sched$-paths satisfying $(\neg \Cause) \Until c$ with $c \in C_{\sched}$.
	Hence,
	\[
	\Pr^{\sched}_{\cM}(\Diamond \Effect \mid  \Diamond \Cause)=\frac{\sum_{c\in C_{\sched}}\Pr^{\sched}_{\cM}(\Diamond \Effect  \mid
		(\neg \Cause) \Until c  )\cdot \Pr^{\sched}_{\cM}((\neg \Cause) \Until c  )}{\sum_{c\in C_{\sched}}\Pr^{\sched}_{\cM}((\neg \Cause) \Until c) }\geq m.
	\]
	As $m>\Pr^{\sched}_{\cM}(\Diamond \Effect)$, the GPR condition \eqref{GPR} is satisfied under $\sched$. 
\end{proof}

\begin{figure}[t]
	\centering
	\begin{minipage}{0.6\textwidth}
		\centering
		\resizebox{\textwidth}{!}{

\begin{tikzpicture}[->,>=stealth',shorten >=1pt,auto ,node distance=0.5cm, thick]
	\node[scale=1, state] (s0) {$\init$};
	\node[scale=1, state] (c1) [below left = 1 of s0] {$c_1$};
	\node[scale=1, state] (c2) [right =2 of c1] {$c_2$};
	\node[scale=1, state] (e) [left =2 of c1] {$\effect$};
	\node[scale=1, state] (f) [right = 2 of c2] {$\noeff$};
	
	\draw[<-] (s0) --++(-0.55,0.55);
	\draw (s0) -- (c1) node[below=0.2, pos=0.3,scale=1] {$1/3$};
	\draw (s0) -- (c2) node[below=0.2, pos=0.3,scale=1] {$1/3$};
	\draw (s0) -- (e) node[above, pos=0.5,scale=1] {$1/12$};
	\draw (s0) -- (f) node[pos=0.5,scale=1] {$1/4$};
	\draw (c1) -- (e) node[pos=0.3,above,scale=1] {$1$};
	\draw (c2) -- (f) node[below, pos=0.5,scale=1] {$3/4$};
	\draw (c2) to [out=200, in=340] (e) node[below right = .5, scale=1] {$1/4$};
\end{tikzpicture}
		}
		\caption{A MC allowing for non-strict GPR causes}
		\label{fig: non-strict-global-causes-in-MC}
	\end{minipage}\hspace{25pt}
	\begin{minipage}{0.3\textwidth}
		\centering
		\resizebox{0.28\textwidth}{!}{
	\begin{tikzpicture}
		[scale=1,->,>=stealth',auto ,node distance=0.5cm, thick]
		\tikzstyle{round}=[thin,draw=black,circle]
		
		\node[scale=1, state] (init) {$\init$};
		\node[scale=1, state, below=1 of init] (g) {$\eff$};
		
		\draw[<-] (init) --++(-0.55,0.55);
		\draw[color=black,->] (init) edge node [pos=0.5, right] {$1$} (g) ;
	\end{tikzpicture}
		}
		\caption{A MC with no PR cause}
		\label{fig: no-prob-raising-cause}
	\end{minipage}
\end{figure}

\begin{exa}[Non-strict GPR cause]
	\label{ex:non-strict-global-causes-in-MC}  
	{\rm
		Consider the Markov chain $\cM$ depicted in Figure \ref{fig: non-strict-global-causes-in-MC}
		where the nodes represent states and the directed edges represent transitions labeled with their respective probabilities.
		Let $\Effect=\{\effect\}$.
		Then,
		\begin{center}
			$\Pr_{\cM}(\Diamond \Effect) = \frac{1}{3} \, + \, \frac{1}{3}\cdot \frac{1}{4} \, + \, \frac{1}{12} = \frac{1}{2},\ \ $
			$ \Pr_{\cM}( \Diamond \Effect  |  \Diamond c_1  ) = \Pr_{\cM,c_1}(\Diamond \effect) = 1 \text{ and}\ \ $
			$ \Pr_{\cM}( \Diamond \Effect  |  \Diamond c_2  ) = \Pr_{\cM,c_2}(\Diamond \effect) =\frac{1}{4}.$
		\end{center}       
		Thus, $\{c_1\}$ is both an SPR and a GPR cause for $\Effect$, while $\{c_2\}$
		is not.
		The set $\Cause=\{c_1,c_2\}$ is a
		non-strict GPR cause for $\Effect$
		as:
		\begin{center}
			$\Pr_{\cM}(\ \Diamond \Effect \ | \ \Diamond \Cause \ ) 
			=
			(\frac{1}{3}+\frac{1}{3}\cdot \frac{1}{4}) / (\frac{1}{3}+\frac{1}{3})
			=
			(\frac{5}{12})/(\frac{2}{3})
			=
			\frac{5}{8}
			>
			\frac{1}{2}
			=
			\Pr_{\cM}(\Diamond \Eff)$.
		\end{center}
		Non-strictness follows from the fact that
		the SPR condition does not hold for
		state $c_2$.\Ende		
	}
\end{exa}

\begin{exa}[Probability-raising causes might not exist]
	\label{ex:no-prob-raising-cause}  
	{\rm  
		PR causes might not
		exist, even if $\cM$ is a Markov chain.
		This applies, e.g., to the MC in Figure \ref{fig: no-prob-raising-cause} and the effect set $\Effect=\{\effect\}$.
		The only cause candidate is the singleton $\{\init\}$.
		However, the strict inequality in \eqref{GPR} or \eqref{SPR} forbids $\{\init\}$ to be a PR cause.
		The same phenomenon occurs if all non-terminal states of a MC
		reach the effect states with the same probability.
		In such cases, however, the non-existence of PR causes
		is well justified as the events $\Diamond \Effect$ and
		$\Diamond \Cause$
		are stochastically independent for every set
		$\Cause \subseteq S \setminus \Effect$.
		\Ende
	}
\end{exa}

\begin{rem}[Memory needed for refuting PR condition]
	\label{rem:memory_necessary}
	\label{rem:memory-needed}
	Let $\cM$ be the MDP in Figure~\ref{fig:memory-needed}, where the notation is similar to Example \ref{ex:non-strict-global-causes-in-MC} with the addition of actions $\alpha, \beta$ and $\gamma$.
	Let $\Cause = \{c\}$
	and $\Effect=\{\eff\}$.
	Only state $s$ has a nondeterministic choice.
	$\Cause$ is not an PR cause. To see this,
	regard the finite-memory deterministic scheduler $\tsched$ that
	schedules $\beta$ only for the first visit of $s$ and $\alpha$ for
	the second visit of $s$.
	Then:
	\begin{align*}
		\Pr^{\tsched}_{\cM}(\Diamond \eff)
		\ = \ \frac{1}{2} \cdot \frac{1}{2} \, + \, 
		\frac{1}{2} \cdot \frac{1}{2} \cdot 1 \cdot \frac{1}{4}
		\ = \ \frac{5}{16}
		\ > \ \frac{1}{4} \ = \
		\Pr^{\tsched}_{\cM}(\Diamond \eff |\Diamond c)
	\end{align*}  
	Denote the MR schedulers reaching $c$ with positive probability as $\sched_{\lambda}$ with $\sched_{\lambda}(s)(\alpha)$ $=$ $\lambda$ and $\sched_{\lambda}(s)(\beta)=1{-}\lambda$  for some $\lambda \in \, [0,1[$. Then, $\Pr^{\sched_{\lambda}}_{\cM,s}(\Diamond \eff)  >0$ and:
	\begin{align*}
		\Pr^{\sched_{\lambda}}_{\cM}(\Diamond \eff)
		\ = \
		\frac{1}{2} \cdot \Pr^{\sched_{\lambda}}_{\cM,s}(\Diamond \eff)
		\ < \ 
		\Pr^{\sched_{\lambda}}_{\cM,s}(\Diamond \eff)
		\ = \
		\Pr^{\sched_{\lambda}}_{\cM,c}(\Diamond \eff)
		\ = \
		\Pr^{\sched_{\lambda}}_{\cM}(\Diamond \eff |\Diamond c)
	\end{align*}
	Thus, the SPR/GPR condition holds for $\Cause$
	and $\Effect$ 
	under all memoryless schedulers reaching $\Cause$ with positive probability,
	although $\Cause$ is not an PR cause.
	\Ende
\end{rem}

\begin{rem}[Randomization needed for refuting PR condition]
	\label{rem:randomization-needed}
	Consider the MDP $\cM$ in \mbox{Figure \ref{fig:randomization-needed}}.
	Let $\Effect = \{\effuncov,\effcov\}$ and $\Cause = \{c\}$.
	Here the state $\effuncov$ is not covered by the cause whereas $\effcov$ is, hence their names.
	The two MD-schedulers $\sched_{\alpha}$ and $\sched_{\beta}$ 
	which select $\alpha$ resp. $\beta$ for the initial
	state $\init$ are the only deterministic schedulers.
	As $\sched_{\alpha}$ does not reach $c$, it is irrelevant for
	the PR conditions.
	On the other hand $\sched_{\beta}$ satisfies \eqref{SPR} and \eqref{GPR} since
	\begin{align*}
		\Pr^{\sched_{\beta}}_{\cM}(\Diamond \Effect|\Diamond c)
		= \frac{1}{2} > \frac{1}{4}=
		\Pr^{\sched_{\beta}}_{\cM}(\Diamond \Effect).
	\end{align*}
	The MR scheduler $\tsched$ which selects $\alpha$ and $\beta$
	with probability $\frac{1}{2}$ in $\init$
	also reaches $c$ with positive probability but violates
	\eqref{SPR} and \eqref{GPR} as
	\begin{align*}
		\Pr^{\tsched}_{\cM}(\Diamond \Effect|\Diamond c)
		= \frac{1}{2} < \frac{5}{8}
		= \frac{1}{2} + \frac{1}{2} \cdot \frac{1}{2} \cdot \frac{1}{2}
		= \Pr^{\tsched}_{\cM}(\Diamond \Effect).
	\end{align*}
\end{rem}

\begin{figure}[t]
	\centering
	\begin{minipage}{0.41\textwidth}
		\centering
		\resizebox{\textwidth}{!}{

\begin{tikzpicture}[scale=1,->,>=stealth',auto ,node distance=0.5cm, thick]
	\tikzstyle{round}=[thin,draw=black,circle]
	
	\node[scale=1, state] (init) {$\init$};
	\node[scale=1, state, below=1.5 of init] (noeff) {$\noeff$};
	\node[scale=1, state, right=1.5 of init] (s) {$s$};
	\node[scale=1, state, below = 1.5 of s] (eff) {$\eff$};
	\node[scale=1, state, right=1.5 of eff] (c) {$c$};
	
	\draw[<-] (init) --++(-0.55,0.55);
	\draw[color=black ,->] (init) edge  node [pos=0.5,above] {$\gamma \mid 1/2$} (s) ;
	\draw[color=black ,->] (init)  edge  node [pos=0.5, left] {$\gamma \mid 1/2$} (noeff) ;
	\draw[color=black ,->] (s) edge  node [anchor=center] (n1) {} node [pos=0.8,right] {$3/4$} (noeff) ;
	\draw[color=black,->] (s) edge[out=260, in=100] node [anchor=center] (m1) {} node [pos=0.85,left] {$1/4$} (eff) ;
	\draw[color=black,->] (s) edge[out=280, in=80] node [anchor=center] (m2) {} node [pos=0.85,right] {$1/2$} (eff) ;
	\draw[color=black,->] (s) edge node [anchor=center] (n2) {} node [pos=0.8,left] {$1/2$} (c) ;
	\draw[color=black,->] (c) edge[out=90, in=0] node [pos=0.5, right] {$\gamma \mid 1$} (s) ;
	\draw[color=black , very thick, -] (n1.center) edge [bend right=45] node [pos=0.25] {$\alpha$} (m1.center);
	\draw[color=black, very thick, -] (m2.center) edge [bend right=45] node [pos=0.25] {$\beta$} (n2.center);
	
\end{tikzpicture}
		}
		\caption{MDP $\cM$ from Remark \ref{rem:memory-needed}}
		\label{fig:memory-needed}
	\end{minipage}\hspace{25pt}
	\begin{minipage}{0.41\textwidth}
		\centering
		\resizebox{\textwidth}{!}{

\begin{tikzpicture}[scale=1,->,>=stealth',auto ,node distance=0.5cm, thick]
	\tikzstyle{round}=[thin,draw=black,circle]
	
	\node[scale=1, state] (init) {$\init$};
	\node[scale=1, state, below=1.25 of init] (pre) {$\effuncov$};
	\node[scale=1, state, right=3.75 of init] (c) {$c$};
	\node[scale=1, state, right=1.25 of pre] (t) {$\noeff$};
	\node[scale=1, state, below=1.25 of c] (post) {$\effcov$};
	
	\draw[<-] (init) --++(-0.55,0.55);
	\draw[color=black ,->] (init) edge  node [very near start, anchor=center] (n5) {} node [pos=0.5,above] {$1/2$} (c) ;
	\draw[color=black ,->] (init)  edge  node [very near start, anchor=center] (n0) {} node [pos=0.5, left] {$\alpha \mid 1$} (pre) ;
	\draw[color=black ,->] (c) edge  node [near start, anchor=center] (m5) {} node [pos=0.5,right] {$1/2$} (post) ;
	\draw[color=black,->] (init) edge node [near start, anchor=center] (n6) {} node [pos=0.5,right] {$1/2$} (t) ;
	\draw[color=black , very thick, -] (n6.center) edge [bend right=45] node [pos=0.3] {$\beta$} (n5.center);
	
	\draw[color=black ,->] (c)  edge  node [near start, anchor=center] (m6) {} node [pos=0.5,left] {$1/2$} (t) ;
	\draw[color=black, very thick, -] (m6.center) edge [bend right=45] node [pos=0.3] {$\tau$} (m5.center);
	
\end{tikzpicture}
		}
		\caption{MDP $\cM$ from Remark \ref{rem:randomization-needed}}
		\label{fig:randomization-needed}
	\end{minipage}
\end{figure}

\begin{rem}[Cause-effect relations for regular classes of schedulers]
	\label{remark:regular-classes-of-schedulers}
	The definitions of PR causes in MDPs impose constraints for all
	schedulers reaching a cause state.
	This condition is fairly strong and can often lead to the phenomenon
	that no PR cause exists.
	E.g, in order for $\Cause$ to not be a PR cause of any kind it suffices if there is a scheduler $\sched$ which reaches a cause state but also reaches $\Eff$ from the initial state with some extreme probability of $0$ or $1$.
	Replacing $\cM$ with an MDP resulting from the
	synchronous parallel composition of $\cM$ with a
	deterministic finite automaton
	representing a regular constraint on the scheduled state-action sequences
	(e.g., ``alternate between actions $\alpha$ and $\beta$ in state $s$''
	or ``take $\alpha$ on every third visit to state $s$ and actions $\beta$ or $\gamma$ otherwise'')
	leads to a weaker notion of PR causality.
	In such a construction the considered class of schedulers can be significantly reduced.
	This can be useful to obtain more detailed
	information on cause-effect relationships in special scenarios,
	be it at design time where multiple scenarios (regular classes of schedulers)
	are considered or
	for a post-hoc analysis where one seeks for the causes
	of an occurred effect and where information about the scheduled actions
	is extractable from log files or the information gathered by a monitor.
	\Ende
\end{rem}

\begin{rem}[Action PR causality]
	\label{remark:action-causality}
	Our notions of PR causes are purely state-based with PR conditions that compare probabilities under the same scheduler.
	However, in combination with model transformations, the proposed notions of PR causes are also applicable for reasoning about other forms of PR causality.
	
	Suppose, the task is to check whether taking action $\alpha$ in state $s$ raises the effect probabilities compared to never scheduling $\alpha$ in state $s$.
	This form of action causality was discussed in an example in \cite{DimFinkbeinerTorfah-ATVA2020}.
	We argue that we can deal with this kind of causality too.
	
	For this we assume there are no cycles in $\cM$ containing $s$.
	This is a strong assumption as we do not want to force the action $\alpha$ to be taken in the first visit to $s$ or to be always taken when visiting $s$.
	Therefore, our framework can not handle this kind of action causality if the state in question is part of a cycle.
	
	Let $\cM_0$ and $\cM_1$ be copies of $\cM$ with the following modifications: 
	In $\cM_0$, the only enabled action of state $s$ is $\alpha$, while in $\cM_1$ the enabled actions of state $s$ are the elements of $\Act_{\cM}(s)\setminus \{\alpha\}$.
	The MDP $\cN$ then has a fresh initial state $\init$ which transitions with equal probabilities $1/2$ to the copies of $s$ in $\cM_0$ and $\cM_1$.
	The action $\alpha$ raises the effect probability in $\cM$ if and only if for all scheduler $\sched$ of $\cN$ the copy of $s$ in $\cM_0$ satisfies \eqref{SPR} for the union of effect sets of both copies in $\cN$.
	This idea can be generalized to check whether scheduler classes satisfying a regular constraint have higher effect probability compared to all other schedulers.
	In this case, we deal with an MDP $\cN$ as above where $\cM_0$ and $\cM_1$ are defined as the synchronous product of deterministic finite automata and $\cM$.
	
	To demonstrate this consider the MDP $\cM$ from Figure \ref{fig:action1}.
	We are interested whether taking $\alpha$ in $s$ raises the probability to reach the effect state $\eff$.
	The constructed MDP $\cN$ with two adapted copies of $\cM$ is depicted in Figure \ref{fig:action2}.
	For all scheduler $\sched$ of $\cN$ the state $s_0$ satisfies \eqref{SPR} by
	\begin{align*}
		\Pr_{\cN}^\sched(\lozenge \{\eff_0, \eff_1\} \mid \lozenge s_0) = 1/4 > 1/8 = \Pr^{\sched}_\cN(\lozenge \{\eff_0, \eff_1\}),
	\end{align*}
	which means that the action $\alpha$ does indeed raise the probability of reaching $\eff$ in $\cM$.
	\Ende
\end{rem}

\begin{figure}[t]
	\centering
	\begin{minipage}{0.4\textwidth}
		\centering
		\resizebox*{0.68\textwidth}{!}{\begin{tikzpicture}
	[scale=1,->,>=stealth',auto ,node distance=0.5cm, thick]
	
	\node[scale=1, state] (init) {$\init$};
	\node[scale=1, state, right=0.8 of init] (s) {$s$};
	\node[scale=1, state, below=0.7 of s] (t) {$t$};
	\node[scale=1, state, below=2.1 of init] (eff) {$\eff$};
	\node[scale=1, ellipse, draw, below right=1 of t] (noeff) {$\noeff$};
	
	\draw[<-] (init) --++(-0.55,0.55);
	
	\draw[color=black, ->] (init) edge node[pos=0.2, right] {$1/10$} (eff);
	\draw[color=black, ->] (init) edge node[pos=0.5, above] {$9/10$} (s);
	
	\draw[color=black, ->] (s) edge node[pos=0.5] {$\alpha$} (t);
	\draw[color=black, ->] (s) edge[out=325, in=90] node[pos=0.5] {$\beta$} (noeff);
	
	\draw[color=black, ->] (t) edge node[pos=0.15, left] {$1/4$} (eff);
	\draw[color=black, ->] (t) edge node[pos=0.15, right] {$3/4$} (noeff);
\end{tikzpicture}}
		\caption{MDP $\cM$ from Remark \ref{remark:action-causality}}
		\label{fig:action1}
	\end{minipage}
	\hfill
	\begin{minipage}{0.59\textwidth}
		\centering
		\resizebox{0.945\textwidth}{!}{\begin{tikzpicture}
	[scale=1,->,>=stealth',auto ,node distance=0.5cm, thick]
	
	\node[scale=1, state] (help) {$\init$};
	
	\draw[<-] (help) --++(-0.55,0.55);

	\node[scale=1, state, left=1.7 of help] (s0) {$s_0$};
	\node[scale=1, state, left=0.8 of s0] (init0) {$\init_0$};
	\node[scale=1, state, below=0.7 of s0] (t0) {$t_0$};
	\node[scale=1, state, below=2.1 of init0] (eff0) {$\eff_0$};
	\node[scale=1, ellipse, draw, below right=1 of t0] (noeff0) {$\noeff_0$};
	
	\draw[color=black, ->] (init0) edge node[pos=0.2, right] {$1/10$} (eff0);
	\draw[color=black, ->] (init0) edge node[pos=0.5, above] {$9/10$} (s0);
	
	\draw[color=black, ->] (s0) edge node[pos=0.5] {$\alpha$} (t0);
	
	\draw[color=black, ->] (t0) edge node[pos=0.15, left] {$1/4$} (eff0);
	\draw[color=black, ->] (t0) edge node[pos=0.15, right] {$3/4$} (noeff0);

	\node[scale=1, state, right=1.7 of help] (s1) {$s_1$};
	\node[scale=1, state, right=0.8 of s1] (init1) {$\init_1$};
	\node[scale=1, state, below=0.7 of s1] (t1) {$t_1$};
	\node[scale=1, state, below=2.1 of init1] (eff1) {$\eff_1$};
	\node[scale=1, ellipse, draw, below left=1 of t1] (noeff1) {$\noeff_1$};
	
	\draw[color=black, ->] (init1) edge node[pos=0.2, left] {$1/10$} (eff1);
	\draw[color=black, ->] (init1) edge node[pos=0.5, above] {$9/10$} (s1);
	
	\draw[color=black, ->] (s1) edge[out=215, in=90] node[pos=0.5, above left] {$\beta$} (noeff1);
	
	\draw[color=black, ->] (t1) edge node[pos=0.15, right] {$1/4$} (eff1);
	\draw[color=black, ->] (t1) edge node[pos=0.15, left] {$3/4$} (noeff1);
	
	\draw[color=black, ->] (help) edge node[pos=0.5, above] {$1/2$} (s0);
	\draw[color=black, ->] (help) edge node[pos=0.5, above] {$1/2$} (s1);
\end{tikzpicture}}
		\caption{$\cN$ with the two copies $\cM_0$ and $\cM_1$}
		\label{fig:action2}
	\end{minipage}	
\end{figure}


\section{Checking the existence of PR causes and the PR conditions}

\label{sec:check}

We now turn to algorithms for checking whether a given set $\Cause$ is an SPR or GPR cause for $\Effect$ in $\cM$.
Since the minimality condition (for all $c \in \Cause: \Pr^{\max}_\cM(\neg \Cause \until c) > 0$) of PR causes is verifiable by standard model checking techniques in polynomial time, we concentrate on checking the probability-raising conditions (S) and (G).
In the special case where $\cM$ is a Markov chain, both conditions \eqref{SPR} and \eqref{GPR} can be checked in polynomial time by computing the corresponding probabilities.
Thus, the interesting case is checking the PR conditions in MDPs. 
In case of SPR causality this is closely related to the existence of PR causes and decidable in polynomial time (Section~\ref{sec:check-SPR}), while checking the GPR condition is more complex and polynomially reducible to (the non-solvability of) a quadratic constraint system (Section~\ref{sec:check-GPR}).

We start with the preliminary consideration that for both conditions (S) and (G), it suffices to consider a class of worst-case schedulers, which are minimizing the probability to reach an effect state from every cause state.
For this we transform the MDP in question according to the cause candidate in question.

\begin{nota}
	[MDP with minimal effect probabilities from cause candidates]
	\label{notation:MDP-mit-min-prob-ab-cause-candidates}  
	If $C \subseteq S$ then we write $\wminMDP{\cM}{C}$ for the MDP resulting from $\cM$ by removing all enabled actions and transitions of the states in $C$.
	Instead, $\wminMDP{\cM}{C}$ has a fresh action $\gamma$ which is enabled exactly in the states $s \in C$ with the transition probabilities $P_{\wminMDP{\cM}{C}}(s,\gamma,\eff)=\Pr^{\min}_{\cM,s}(\Diamond \Effect)$ and $P_{\wminMDP{\cM}{C}}(s,\gamma,\noeff)=1{-}\Pr^{\min}_{\cM,s}(\Diamond \Effect)$.
	Here, $\eff$ is a fixed state in $\Effect$ and $\noeff$ a (possibly fresh) terminal state not in $\Effect$. 
	We write $\wminMDP{\cM}{c}$ if $C=\{c\}$ is a singleton.
	\Ende
\end{nota}

	As an example for the model transformation consider the abstract MDP $\cM$ from Figure \ref{fig:M_c-raw} for the singleton set $C = \{c\}$.
	The transformed MDP $\wminMDP{\cM}{c}$ is seen in Figure \ref{fig:M_c-transformed}, where a fresh state $\noeff$ is added.
	Furthermore, all outgoing transitions from $c$ are deleted a replaced by a fresh action $\gamma$ with exactly two transitions corresponding to $P_{\wminMDP{\cM}{c}}(c,\gamma,\eff)=\Pr^{\min}_{\cM,c}(\Diamond \Effect)$ and $P_{\wminMDP{\cM}{c}}(c,\gamma,\noeff)=1{-}\Pr^{\min}_{\cM,c}(\Diamond \Effect)$.
	\begin{figure}[t]
		\centering
		\begin{minipage}{0.43\textwidth}
			\centering
			\resizebox{0.9\textwidth}{!}{
				\begin{tikzpicture}[scale=1,->,>=stealth',auto ,node distance=0.5cm, thick]
	\tikzstyle{round}=[thin,draw=black,circle]
	
	\node[scale=1, state] (init) {$\init$};
	\node[scale=5, ellipse, draw, dotted, below right=0.1 of init] (MDP) {\phantom{text}};
	\node[scale=1, above=-1 of MDP] (MDPtext) {MDP};
	\node[scale=0.2, circle, draw, below=0.5 of init] (a) {\phantom{a}};
	\node[scale=0.2, circle, draw, below right=0.5 of init] (b) {\phantom{a}};
	\node[scale=0.2, below right=2 of b] (c) {\phantom{a}};
	\node[scale=0.2, circle, draw, below=0.5 of c] (e) {\phantom{a}};
	\node[scale=0.2, circle, draw, right=1 of c] (d) {\phantom{a}};
	\node[scale=1, state, below=1 of b] (cause) {$c$};
	\node[scale=1, state, below=0.2 of MDP] (eff1) {$\eff$};
	\node[scale=1, state, right=1 of eff1] (eff2) {$\eff$};
	
	\draw[<-] (init) --++(-0.55,0.55);
	\draw[color=black ,->] (init) edge  (a) ;
	\draw[color=black ,->] (init) edge  (b) ;
	
	\draw[->] (a) --++(0.35,-0.35);
	\draw[->] (a) --++(0.5,0);
	
	\draw[->] (b) --++(0.35,-0.35);
	\draw[->] (b) --++(0.35,0.35);
	
	\draw[<-] (cause) --++(0,01);
	\draw[->] (cause) --++(0.65,0.65);
	\draw[color=black ,->] (cause) edge  (e) ;
	
	\draw[<-] (e) --++(0,0.5);
	\draw[->] (e) edge (eff1);
	\draw[->] (e) edge (d);
	
	\draw[->] (d) edge (eff2);
	\draw[<-] (d) --++(-0.35,0.35);
	\draw[->] (d) --++(0.35,0.35);
	
\end{tikzpicture}
			}
			\caption{MDP $\cM$}
			\label{fig:M_c-raw}
		\end{minipage}
		\hfill
		\begin{minipage}{0.52\textwidth}
			\centering
			\resizebox{0.9\textwidth}{!}{
				\begin{tikzpicture}[scale=1,->,>=stealth',auto ,node distance=0.5cm, thick]
	\tikzstyle{round}=[thin,draw=black,circle]
	
	\node[scale=1, state] (init) {$\init$};
	\node[scale=5, ellipse, draw, dotted, below right=0.1 of init] (MDP) {\phantom{text}};
	\node[scale=1, above=-1 of MDP] (MDPtext) {MDP};
	\node[scale=0.2, circle, draw, below=0.5 of init] (a) {\phantom{a}};
	\node[scale=0.2, circle, draw, below right=0.5 of init] (b) {\phantom{a}};
	\node[scale=0.2, below right=2 of b] (c) {\phantom{a}};
	\node[scale=0.2, circle, draw, below=0.4 of c] (e) {\phantom{a}};
	\node[scale=0.2, circle, draw, right=1 of c] (d) {\phantom{a}};
	\node[scale=1, state, below=1 of b] (cause) {$c$};
	\node[scale=1, state, below=0.2 of MDP] (eff1) {$\eff$};
	\node[scale=1, ellipse, draw, left=1.5 of eff1] (noeff) {$\noeff$};
	\node[scale=1, state, right=1 of eff1] (eff2) {$\eff$};
	
	\draw[<-] (init) --++(-0.55,0.55);
	\draw[color=black ,->] (init) edge  (a) ;
	\draw[color=black ,->] (init) edge  (b) ;
	
	\draw[->] (a) --++(0.35,-0.35);
	\draw[->] (a) --++(0.5,0);
	
	\draw[->] (b) --++(0.35,-0.35);
	\draw[->] (b) --++(0.35,0.35);
	
	\draw[<-] (cause) --++(0,01);
	\draw[->, color=red] (cause) edge node[color=black, scale=0.8, pos=0.8, left] {$\Pr^{\min}_c(\lozenge \Eff)$} (eff1);
	\draw[->, color=red] (cause) edge node[color=black, scale=0.8, pos=0.8, left] {$1-\Pr^{\min}_c(\lozenge \Eff)$}(noeff);

	\draw[<-] (e) --++(0,0.5);
	\draw[->] (e) edge (eff1);
	\draw[->] (e) edge (d);
	
	\draw[->] (d) edge (eff2);
	\draw[<-] (d) --++(-0.35,0.35);
	\draw[->] (d) --++(0.35,0.35);
	
\end{tikzpicture}
			}
			\caption{Transformed MDP $\wminMDP{\cM}{c}$}
			\label{fig:M_c-transformed}
		\end{minipage}
	\end{figure}

\begin{lem}
	\label{lemma:wmin-criterion-PR-causes}    
	Let $\cM=(S,\Act,P,\init)$ be an MDP and $\Effect\subseteq S$ a set of terminal states.
	Let $\Cause \subseteq S\setminus \Effect$. Then, $\Cause$ is an SPR cause (resp. a GPR cause) for $\Effect$ in $\cM$ if and only if $\Cause$ is an SPR cause (resp. a GPR cause) for $\Effect$ in $\wminMDP{\cM}{\Cause}$.  
\end{lem}

Obviously, for all $c \in \Cause: \Pr^{\max}_\cM(\neg \Cause \until c) > 0$ holds for $\Cause$ in $\cM$ if and only if
it holds for $\Cause$ in $\wminMDP{\cM}{\Cause}$.  
Furthermore, it is clear all SPR resp. GPR causes of $\cM$
are also SPR resp. GPR causes in $\wminMDP{\cM}{\Cause}$. 
So, it remains to prove the converse direction.
This will be done
in Lemma \ref{app:lem:criterion-strict-prob-raising} for SPR causes and
in Lemma \ref{app:lem:criterion-global-prob-raising} for GPR causes.

\begin{lem}%
	[Criterion for strict probability-raising causes]
	\label{app:lem:criterion-strict-prob-raising}
	Suppose $\Cause$ is
	an SPR cause for $\Effect$ in $\wminMDP{\cM}{\Cause}$.
	Then, $\Cause$ is an SPR cause for $\Effect$ in $\cM$.
\end{lem}

\begin{proof}
	We show that $\Cause$ is an SPR cause in $\cM$ by showing (S) for all states in $\Cause$.
	Thus, we fix a state $c\in \Cause$.
	Recall also that we assume the states in $\Effect$ to be terminal.
	Let $\psi_c=(\neg \Cause) \Until c$, $w_c = \Pr_{\cM, c}^{\min}(\Diamond \Eff)$ and
	let $\Usched_c$ denote the set of all schedulers $\usched$
	for $\cM$ such that
	\begin{itemize}
		\item
		$\Pr_{\cM}^{\usched}(\psi_c)>0$ and
		\item
		$\Pr^{\residual{\usched}{\pi}}_{\cM,c}(\Diamond \Effect)=w_c$
		for each finite $\usched$-path $\pi$ from $\init$ to $c$.
	\end{itemize}
	Clearly, $\Pr^{\usched}_{\cM}( \Diamond c \wedge \Diamond \Effect)
	= \Pr_\cM^\usched(\Diamond c)\cdot w_c$
	for $\usched \in \Usched_c$.
	As $\Cause$ is an SPR cause in $\wminMDP{\cM}{\Cause}$
	we have:
	\begin{align}
		\label{SPR-1} 
		w_c \ > \ \Pr^{\usched}_{\cM}(\Diamond \Effect)
		\quad \text{for all schedulers $\usched \in \Usched_c$}. 
	\end{align}
	The task is to prove that (S)
	holds for $c$ and all schedulers of $\cM$ with
	$\Pr^{\sched}_{\cM}(\psi_c) >0$.
	
	Suppose $\sched$ is a scheduler for $\cM$ with
	$\Pr^{\sched}_{\cM}(\psi_c) >0$. Then $\Pr^{\sched}_{\cM}(\psi_c \wedge \Diamond \Eff)
	\geqslant  
	\Pr^{\sched}_{\cM}(\psi_c)\cdot w_c.$
	Moreover, there
	exists a scheduler $\usched =\usched_{\sched} \in \Usched_c$
	with
	\begin{center}
		$\Pr^{\sched}_{\cM}(\psi_c) = \Pr^{\usched}_{\cM}(\psi_c)$
		\ \ and \ \
		$\Pr^{\sched}_{\cM}((\neg \psi_c) \wedge \Diamond \Effect)=
		\Pr_\cM^\usched((\neg \psi_c) \wedge \Diamond \Effect)$.
	\end{center}
	To see this, consider the scheduler $\usched$ that
	behaves as $\sched$ as long as $c$ is not reached.
	As soon as $\usched$ has reached $c$, scheduler $\usched$
	switches mode and
	behaves
	as an MD-scheduler minimizing the
	probability to reach an effect state.
	The SPR condition (S)
	holds for $c$ and $\sched$ if and only if
	\begin{align}
		\label{SPR-2}
		\frac{\Pr^{\sched}_{\cM}( \psi_c \wedge \Diamond \Effect)}
		{\Pr^{\sched}_{\cM}(\psi_c)}
		\ \ > \ \
		\Pr^{\sched}_{\cM}(\Diamond \Effect)
	\end{align}
	Using $
		\Pr^{\sched}_{\cM}(\Diamond \Effect)
		 = 
		\Pr^{\sched}_{\cM}( \psi_c \wedge \Diamond \Effect)
		\ + \
		\Pr^{\sched}_{\cM}(
		(\neg \psi_c) \wedge \Diamond \Effect ),$
	we can equivalently convert condition \eqref{SPR-2} 
	for $c$ and $\sched$ to
	\begin{align}
		\label{SPR-3}
		\Pr^{\sched}_{\cM}(\psi_c \wedge \Diamond \Effect)
		\cdot
		\frac{1-\Pr^{\sched}_{\cM}(\psi_c)}{\Pr^{\sched}_{\cM}(\psi_c)}
		\ \ > \ \
		\Pr^{\sched}_{\cM}((\neg \psi_c) \wedge \Diamond \Effect )
	\end{align}
	The remaining
	task is now to derive \eqref{SPR-3} from \eqref{SPR-1}.	
	Applying \eqref{SPR-1} to scheduler $\usched=\usched_{\sched}$ yields:
	\[
	\begin{array}{lcl}         
		w_c & \ \ > \ \ &
		\Pr^{\usched}_{\cM}(\psi_c \wedge  \Diamond \Effect)
		\ + \
		\Pr^{\usched}_{\cM}((\neg \psi_c) \wedge  \Diamond \Effect)
		\\
		\\[0ex]
		& = &
		\Pr^{\sched}_{\cM}(\psi_c)\cdot w_c
		\ + \
		\Pr^{\sched}_{\cM}((\neg \psi_c) \wedge  \Diamond \Effect).
	\end{array}   
	\]
	We conclude:
	\begin{eqnarray*}
		\Pr^{\sched}_{\cM}(\psi_c \wedge \Diamond \Effect)
		\cdot
		\frac{1-\Pr^{\sched}_{\cM}(\psi_c)}{\Pr^{\sched}_{\cM}(\psi_c)}
		& \ \geqslant \ &
		\Pr^{\sched}_{\cM}(\psi_c) \cdot w_c
		\cdot
		\frac{1-\Pr^{\sched}_{\cM}(\psi_c)}{\Pr^{\sched}_{\cM}(\psi_c)}
		\\
		\\[0ex]
		& = &
		\bigl(1-\Pr^{\sched}_{\cM}(\psi_c)\bigr)\cdot w_c
		\\[1ex]
		& > &
		\Pr^{\sched}_{\cM}((\neg \psi_c) \wedge  \Diamond \Effect).
	\end{eqnarray*}
	Thus, \eqref{SPR-3} holds for $c$ and $\sched$.
\end{proof}

\begin{lem}[Criterion for GPR causes]
	\label{app:lem:criterion-global-prob-raising}
	Suppose $\Cause$ is
	a GPR cause for $\Effect$ in $\wminMDP{\cM}{\Cause}$.
	Then, $\Cause$ is a GPR cause for $\Effect$ in $\cM$.
\end{lem}

\begin{proof}
	From the assumption that $\Cause$ is
	a GPR cause for $\Effect$ in $\wminMDP{\cM}{\Cause}$, we can conclude that 
	the GPR condition \eqref{GPR} holds
	for all schedulers $\sched$
	that satisfy
	\[\Pr_{\cM}^{\sched}(\Diamond \Cause)>0 \qquad \text{and} \qquad\Pr^{\residual{\sched}{\pi}}_{\cM,c}(\Diamond \Effect)
	\ \ = \ \ \Pr^{\min}_{\cM,c}(\Diamond \Effect)
	\]
	for each finite $\sched$-path from the initial state
	$\init$ to a state $c \in \Cause$.
	To prove that \eqref{GPR}
	holds
	for all schedulers $\sched$ that satisfy $\Pr_{\cM}^{\sched}(\Diamond \Cause)>0$, we introduce the following notation:
	We write
	\begin{itemize}
		\item
		$\Ssched_{>0}$ for the set of all schedulers $\sched$ such that
		$\Pr_{\cM}^{\sched}(\Diamond \Cause)>0$,
		\item
		$\Ssched_{>0,\min}$ for the set of all schedulers
		with $\Pr_{\cM}^{\sched}(\Diamond \Cause)>0$ such that
		\[
		\Pr^{\residual{\sched}{\pi}}_{\cM,c}(\Diamond \Effect)
		\ \ = \ \ \Pr^{\min}_{\cM,c}(\Diamond \Effect)
		\]
		for each finite $\sched$-path from the initial state
		$\init$ to a state $c \in \Cause$.
	\end{itemize} 
	It now suffices to show
	that for each scheduler $\sched \in \Ssched_{>0}$ there exists a scheduler
	$\sched' \in \Ssched_{>0,\min}$ such that
	if \eqref{GPR} holds for $\sched'$ then \eqref{GPR} holds for $\sched$.
	So, let $\sched\in\Ssched_{>0}$. 
	
	For $c\in \Cause$, let $\Pi_c$ denote the set of
	finite paths
	$\pi=s_0 \, \alpha_0\,  s_1 \, \alpha_1 \ldots \alpha_{n-1}\, s_n$
	with $s_0=\init$, $s_n=c$ and
	$\{s_0,\ldots,s_{n-1} \}\cap (\Cause \cup \Effect) =\varnothing$.
	Let
	\[
	w_{\pi}^{\sched} \ = \ \Pr^{\residual{\sched}{\pi}}_{\cM,c}(\Diamond \Effect)
	\]
	Furthermore, let $p_{\pi}^{\sched}$ denote the probability for
	(the cylinder set of) $\pi$ under scheduler $\sched$.
	Then
	\[
	\Pr^{\sched}_{\cM}((\neg \Cause)\Until c)
	\ = \ \sum_{\pi \in \Pi_c} p_{\pi}^{\sched}.
	\]
	Moreover:
	\[
	\Pr^{\sched}_{\cM}(\Diamond \Effect)
	\ \ = \ \ 
	\Pr^{\sched}_{\cM}(\neg \Cause \Until \Effect) \ + 
	\sum_{c\in \Cause} \sum_{\pi \in \Pi_c}
	\!\!\! p_{\pi}^{\sched} \cdot w_{\pi}^{\sched}\text{ and}
	\]
	\[
	\Pr^{\sched}_{\cM}(\ \Diamond \Effect \ | \ \Diamond \Cause \ )
	\ \ = \ \
	\frac{1}{\Pr_{\cM}^{\sched}(\Diamond \Cause)} \cdot 
	\sum_{c\in \Cause} \
	\sum_{\pi\in \Pi_c} p_{\pi}^{\sched} \cdot w_{\pi}^{\sched}      
	\]
	Thus, the condition \eqref{GPR}
	holds for the scheduler
	$\sched \in \Ssched_{>0}$ if and only if
	\begin{eqnarray*}
		\Pr^{\sched}_{\cM}(\neg \Cause \Until \Effect)
		+
		\sum_{c\in \Cause}\sum_{\pi \in \Pi_c} p_{\pi}^{\sched}\cdot w_{\pi}^{\sched}
		&  <  &
		\frac{1}{\Pr_{\cM}^{\sched}(\Diamond \Cause)} \cdot \!\!\!
		\sum_{c\in \Cause} \sum_{\pi\in \Pi_c} p_{\pi}^{\sched} \cdot w_{\pi}^{\sched}.
	\end{eqnarray*}
	The latter is equivalent to:
	\begin{align*}
		\Pr_{\cM}^{\sched}(\Diamond \Cause) \! \cdot \! \Pr^{\sched}_{\cM}(\neg \Cause \Until \Effect)
		+ 
		\Pr_{\cM}^{\sched}(\Diamond \Cause) \! \cdot \!\!\!\!\!\!
		\sum_{c\in \Cause}  \sum_{\pi \in \Pi_c} \!\!\! p_{\pi}^{\sched}\cdot w_{\pi}^{\sched} 
		< \!\!\!\!\!\!
		\sum_{c\in \Cause}\sum_{\pi \in \Pi_c} \!\!\! p_{\pi}^{\sched}\cdot w_{\pi}^{\sched},
	\end{align*}
	which again is equivalent to:
	\begin{align}
		\label{GPR-2}
		\Pr_{\cM}^{\sched}(\Diamond \Cause) \cdot \Pr^{\sched}_{\cM}(\neg \Cause \Until \Effect) 
		< 
		\bigl(1-\Pr_{\cM}^{\sched}(\Diamond \Cause)\bigr) \cdot \!\!\!
		\sum_{c\in \Cause} \
		\sum_{\pi \in \Pi_c} p_{\pi}^{\sched}\cdot w_{\pi}^{\sched}.
	\end{align}
	Pick an MD-scheduler
	$\tsched$ that minimizes the probability to reach $\Effect$
	from every state.
	In particular, $w_c = w_\pi^{\tsched} \leqslant w_\pi^{\sched}$
	for every state $c\in \Cause$ and every path $\pi\in \Pi_c$
	(recall that $w_c =  \Pr^{\min}_{\cM,c}(\Diamond \Effect)$).
	Moreover, the scheduler $\sched$ can be transformed into a
	scheduler $\sched_{\tsched}\in \Ssched_{>0,\min}$
	that is ``equivalent'' to $\sched$ with respect to the
	global probability-raising condition.
	More concretely, let $\sched_\tsched$ denote the scheduler 
	that behaves as $\sched$ as long as $\sched$ has not yet visited a state
	in $\Cause$ and behaves as $\tsched$ as soon as a state in $\Cause$ has been
	reached. Thus,
	$p_{\pi}^{\sched}=p_{\pi}^{\sched_{\tsched}}$
	and
	$\residual{\sched_{\tsched}}{\pi}=\tsched$
	for each $\pi\in \Pi_c$.
	This yields that the probability to reach $c\in \Cause$ from
	$\init$ is the same under
	$\sched$ and $\sched_{\tsched}$, i.e.,
	$\Pr^{\sched}_{\cM}(\Diamond c)=
	\Pr^{\sched_{\tsched}}_{\cM}(\Diamond c)$.
	Therefore $\Pr^{\sched}_{\cM}(\Diamond \Cause)=\Pr^{\sched_{\tsched}}_{\cM}(\Diamond \Cause)$.
	The latter implies that $\sched_{\tsched}\in \Ssched_{>0}$,
	and hence $\sched_{\tsched}\in \Ssched_{>0,\min}$.
	Moreover, $\sched$ and $\sched_{\tsched}$ reach $\Effect$ without visiting $\Cause$ with the same
	probability, i.e., $\Pr^{\sched}_{\cM}(\neg \Cause \Until \Effect)=\Pr^{\sched_{\tsched}}_{\cM}(\neg \Cause \Until \Effect)$.
	But this yields: if \eqref{GPR-2} holds for $\sched_{\tsched}$ then
	\eqref{GPR-2} holds for $\sched$.
	As \eqref{GPR-2} holds for $\sched_{\tsched}$ by assumption, this completes the proof.
\end{proof}

\subsection{Checking the strict probability-raising condition and the existence of causes}

\label{sec:check-SPR}
\label{sec:SPR_check}
\label{sec:check-SPR-condition}

The basis of both checking the existence of PR causes or checking the SPR condition (S) for a given cause candidate is the following polynomial time algorithm to check whether (S)
holds in a given state $c$ of $\cM$
for all schedulers $\sched$ with $\Pr^{\sched}_{\cM}(\Diamond c)>0$:

\begin{algo}
	\label{alg:SPR-check} \ 
	\begin{description}
		\item [Input] state $c \in S$, set of terminal states $\Eff \subseteq S$
		\item [Task] Decide whether \eqref{SPR} holds in $c$ for all schedulers $\sched$.
	\end{description}
	\begin{enumerate}
		\item [0.] Compute $q_s =\Pr^{\max}_{\wminMDP{\cM}{c},s}(\Diamond \Effect)$ and $w_c = \Pr^{\min}_{\cM,c}(\Diamond \Eff)$ for each state $s$ in $\wminMDP{\cM}{c}$ .
		
		\item [1.] If $q_{\init} < w_c$, then return
		``yes, \eqref{SPR} holds for $c$''.
		
		\item [2.] If $q_{\init} > w_c$, then  return
		``no, \eqref{SPR} does not hold for $c$''.
		
		\item [3.]
		Suppose $q_{\init} = w_c$. Let
		$A(s) = \{\alpha \in \Act_{\wminMDP{\cM}{c}}(s) \mid q_s = \sum_{t\in \wminMDP{S}{c}} P_{\wminMDP{\cM}{c}}(s,\alpha,t)\cdot q_t\}$
		for each non-terminal state $s$.
		Let $\wminMDPmax{\cM}{c}$
		denote the sub-MDP of $\wminMDP{\cM}{c}$
		induced by the state-action pairs $(s,\alpha)$ where
		$\alpha \in A(s)$.
		\begin{enumerate}
			\item [3.1]
			If $c$ is reachable from $\init$ in
			$\wminMDPmax{\cM}{c}$,
			then return
			\mbox{``no,  \eqref{SPR} does not hold for $c$''.}
			
			\item [3.2]
			If $c$ is not reachable from $\init$
			in $\wminMDPmax{\cM}{c}$,
			then return ``yes,  \eqref{SPR} holds for $c$''.
		\end{enumerate}
	\end{enumerate}  
\end{algo}

As the construction of the MDP $\wminMDP{\cM}{c}$ suggests, the two values compared by the algorithm are instances of worst-case scheduler.
On one hand, the probability to reach $\Eff$ starting in $c$ is minimized, while it is maximized if $c$ was not seen yet.
If in such a scenario we have case 1. $q_\init < w_c$ then $c$ obviously satisfies \eqref{SPR}.
In the case 2. $q_\init > w_c$ we can build a scheduler which refuses \eqref{SPR} for $c$.
Lastly, in the corner case 3. $q_\init = w_c$ a treatment by a reachability analysis is needed, as seen in the following Example \ref{ex:M_c_max}.

\begin{exa}
	\label{ex:M_c_max}
	For the transformation to $\wminMDPmax{\cM}{c}$ consider $\wminMDP{\cM}{c}$ from Figure \ref{fig:M_c_max-raw}.
	For $\Cause = \{c\}$ we are in case 3. of Algorithm \ref{alg:SPR-check} as $q_\init = \Pr^{\max}_{\wminMDP{\cM}{c},s}(\Diamond \Effect) = 1/4 = \Pr^{\min}_{\cM,c}(\Diamond \Eff) = w_c$.
	The only non-deterministic choice is in the state $\init$.
	We have $A(\init) = \{\alpha\}$ since $\alpha$ is the only maximizing action for $\lozenge \eff$ in $\init$.
	Thus, in the resulting MDP $\wminMDPmax{\cM}{c}$, depicted in Figure \ref{fig:M_c_max-transformed}, all other actions in $\init$ are deleted.
	We are actually in case 3.2 as $c$ is not reachable from the initial state in $\wminMDPmax{\cM}{c}$ which means that \eqref{SPR} holds for $c$.
	\Ende
	\begin{figure}[t]
		\centering
		\begin{minipage}{0.4\textwidth}
			\centering
			\resizebox{\textwidth}{!}{
				\begin{tikzpicture}
	[scale=1,->,>=stealth',auto ,node distance=0.5cm, thick]
	\tikzstyle{round}=[thin,draw=black,circle]
	
	\node[scale=1, state] (init) {$\init$};
	\node[scale=1, state, below=1.5 of init] (noeff) {$\noeff$};
	\node[scale=1, state, left = 1.5 of noeff] (s) {$s$};
	\node[scale=1, state, right=1.5 of noeff] (c) {$c$};
	\node[scale=1, state, below=1 of noeff] (eff) {$\eff$};
	
	\draw[<-] (init) --++(-0.55,0.55);
	
	\draw[color=black ,->] (init) edge  node[pos=0.3 ,anchor=center] (a1) {} node[pos=0.5,above left] {$1/2$} (s) ;
	\draw[color=black,->] (init) edge[out=250, in=110] node[pos=0.5 ,anchor=center] (a2) {} node[pos=0.7,left] {$1/2$} (noeff) ;
	\draw[color=black,->] (init) edge[out=290, in=70] node[pos=0.5 ,anchor=center] (b1) {} node [pos=0.7,right] {$1/2$} (noeff) ;
	\draw[color=black,->] (init) edge node[pos=0.3 ,anchor=center] (b2) {} node [pos=0.5,above right] {$1/2$} (c) ;
	
	\draw[color=black,->] (s) edge node[pos=0.5,above] {$1/2$} (noeff) ;
	\draw[color=black,->] (s) edge node[pos=0.6,above] {$1/2$} (eff) ;
	
	\draw[color=black,->] (c) edge node [pos=0.5,above] {$3/4$} (noeff) ;
	\draw[color=black,->] (c) edge node [pos=0.6,above] {$1/4$} (eff) ;
	
	\draw[color=black , very thick, -] (a1.center) edge [bend right=45] node [pos=0.3] {$\alpha$} (a2.center);
	\draw[color=black , very thick, -] (b1.center) edge [bend right=45] node [pos=0] {$\beta$} (b2.center);
\end{tikzpicture}
			}
			\caption{Example MDP $\wminMDP{\cM}{c}$}
			\label{fig:M_c_max-raw}
		\end{minipage}
		\hspace*{40pt}
		\begin{minipage}{0.4\textwidth}
			\centering
			\resizebox{\textwidth}{!}{
				\begin{tikzpicture}
	[scale=1,->,>=stealth',auto ,node distance=0.5cm, thick]
	\tikzstyle{round}=[thin,draw=black,circle]
	
	\node[scale=1, state] (init) {$\init$};
	\node[scale=1, state, below=1.5 of init] (noeff) {$\noeff$};
	\node[scale=1, state, left = 1.5 of noeff] (s) {$s$};
	\node[scale=1, state, right=1.5 of noeff] (c) {$c$};
	\node[scale=1, state, below=1 of noeff] (eff) {$\eff$};
	
	\draw[<-] (init) --++(-0.55,0.55);
	
	\draw[color=black ,->] (init) edge  node[pos=0.3 ,anchor=center] (a1) {} node[pos=0.5,above left] {$1/2$} (s) ;
	\draw[color=black,->] (init) edge[out=250, in=110] node[pos=0.5 ,anchor=center] (a2) {} node[pos=0.7,left] {$1/2$} (noeff) ;
	
	\draw[color=black,->] (s) edge node[pos=0.5,above] {$1/2$} (noeff) ;
	\draw[color=black,->] (s) edge node[pos=0.6,above] {$1/2$} (eff) ;
	
	\draw[color=black,->] (c) edge node [pos=0.5,above] {$3/4$} (noeff) ;
	\draw[color=black,->] (c) edge node [pos=0.6,above] {$1/4$} (eff) ;
	
	\draw[color=black , very thick, -] (a1.center) edge [bend right=45] node [pos=0.3] {$\alpha$} (a2.center);
\end{tikzpicture}
			}
			\caption{Transformed MDP $\wminMDPmax{\cM}{c}$}
			\label{fig:M_c_max-transformed}
		\end{minipage}
	\end{figure}
\end{exa}

\begin{lem} \label{soundness-SPR-algo}
	Algorithm \ref{alg:SPR-check} is sound and runs in polynomial time.
\end{lem}

\begin{proof} 
	First, we show the soundness of Algorithm \ref{alg:SPR-check}.
	By the virtue of Lemma \ref{lemma:wmin-criterion-PR-causes} stating the soundness of the transformation $\cM$ to $\wminMDP{\cM}{c}$ it
	suffices to show that Algorithm \ref{alg:SPR-check} returns the correct
	answers ``yes'' or ``no'' when the task is to check whether the
	singleton $\Cause = \{c\}$ is an SPR cause in $\cN=\wminMDP{\cM}{c}$.
	Recall the notation $q_s =\Pr^{\max}_{\wminMDP{\cM}{c},s}(\Diamond \Effect)$.
	We abbreviate $q=q_\init$.
	Note that $(\neg \Cause) \Until c$ is equivalent to $\Diamond c$ as $c \in \Cause$.
	
	For every scheduler $\sched$ of $\cN$ we have
	$\Pr^{\sched}_{\cN,c}(\Diamond \Effect)=w_c$.
	Thus, $\Pr^{\sched}_{\cN}( \Diamond \Effect \ | \ \Diamond c)=w_c$
	if $\sched$ is a scheduler of $\cN$ with $\Pr^\sched_\cN(\Diamond c)>0$.
	
	Algorithm \ref{alg:SPR-check} correctly answers ``no'' (case 2 or 3.1)
	if $w_c=0$.
	Suppose that $w_c>0$. 
	Thus, the SPR condition for $c$ reduces to $\Pr^\sched_\cN(\Diamond \Effect) < w_c$ for all schedulers $\sched$ of $\cN$ with $\Pr^\sched_\cN(\Diamond c)>0$.
	\begin{enumerate}
		\item[1.] 
		of Algorithm \ref{alg:SPR-check} (i.e., if $q < w_c$), the answer ``yes'' is sound
		as $\Pr^{\max}_{\cN}(\Diamond \Effect) = q < w_c$.
		
		\item[2.] 
		(i.e., if $q > w_c$)
		Let $\tsched$ be an MD-scheduler 
		with $\Pr^{\tsched}_{\cN,s}(\Diamond \Effect) = q_s$
		for each state $s$ and pick an MD-scheduler $\sched$ with
		$\Pr^{\sched}_{\cN}(\Diamond c)>0$.
		It is no restriction to suppose that
		$\tsched$ and $\sched$ realize the same end components of $\cN$.
		(Note that if state $s$ belongs to an end component
		that is realized by $\tsched$ then $s$ contained in a bottom strongly
		connected component of the Markov chain induced by $\tsched$.
		But then
		$q_s=0$, i.e., no effect state is reachable from $s$ in $\cN$.
		Recall that all effect states are terminal and thus not contained in
		end components.
		But then we can safely assume that $\tsched$ and $\sched$ schedule
		the same action for state $s$.)
		Let $\lambda$ be any real number with $1 > \lambda > \frac{w_c}{q}$
		and let $\cK$ denote the sub-MDP of $\cN$ with state space $S$
		where the enabled
		actions of state $s$ are the actions scheduled for $s$ under one of the
		schedulers $\tsched$ or $\sched$.
		Let now $\usched$ be the MR-scheduler
		$\lambda \tsched \oplus (1{-}\lambda)\sched$
		defined as in Lemma \ref{lem:convex}
		for the EC-free MDP resulting from $\cK$ when collapsing
		$\cK$'s end components into a single terminal state.
		For the states belonging to an end component of $\cK$,
		$\usched$ schedules the same action as $\tsched$ and $\sched$.
		Then, $\Pr^{\usched}_{\cN}(\Diamond t)=\lambda \Pr^{\tsched}_{\cN}(\Diamond t)+(1{-}\lambda) \Pr^{\sched}_{\cN}(\Diamond t)$
		for all terminal states $t$ of $\cN$ and $t=c$.
		Hence:
		\begin{center}
			$\Pr_{\cN}^{\usched}(\Diamond c)
			\ \ \geqslant \ \
			(1{-}\lambda)\cdot \Pr^{\sched}_{\cM}(\Diamond c) \ \ > \ \ 0$
			and
			\\[1ex]
			$
			\Pr^{\usched}_{\cN}(\Diamond \Effect)
			\ \ \geqslant \ \
			\lambda \cdot \Pr^{\tsched}_{\cM}(\Diamond \Effect)
			\ \ = \ \ \lambda \cdot q \ \ > \ \ w_c
			$
		\end{center}
		Thus, scheduler $\usched$ is a witness
		why \eqref{SPR} does not hold for $c$.
		
		\item[3.1] 
		Pick an MD-scheduler $\sched$ of $\wminMDPmax{\cM}{c}$ such that
		$c$ is reachable from $\init$ via $\sched$ and
		$\Pr^{\sched}_{\cN,s}(\Diamond \Effect)=q_s$ for all states.
		Thus, \eqref{SPR} does not hold for $c$ and scheduler $\sched$.
		\item[3.2]	
		We have $\Pr^{\sched}_{\cN}(\Diamond c)=0$ for all schedulers $\sched$ for $\cN$
		with $\Pr^{\sched}_{\cN}(\Diamond \Effect)=q=w_c$.
		But then $\Pr^{\sched}_{\cN}(\Diamond c)>0$ implies
		$\Pr^{\sched}_{\cN}(\Diamond \Effect) < w_c$ 
		as required in \eqref{SPR}.
	\end{enumerate}	
	The polynomial runtime of Algorithm \ref{alg:SPR-check}
	follows from the fact that minimal and maximal reachability probabilities and hence also the MDPs $\cN=\wminMDP{\cM}{c}$ and its sub-MDP $\wminMDPmax{\cM}{c}$ can be computed in polynomial time.
\end{proof}


By applying Algorithm \ref{alg:SPR-check} to all states $c \in \Cause$ and standard algorithms to check the existence of a path satisfying $(\neg \Cause) \Until c$ for every state $c\in \Cause$, we obtain:

\begin{thm}[Checking SPR causes]
	\label{thm:SPR-check-complexity}
	The problem ``given $\cM$, $\Cause$ and $\Effect$, check whether $\Cause$ is a SPR cause for $\Effect$ in $\cM$'' is solvable in polynomial-time.
\end{thm}  


\begin{rem}[Memory requirements for (S)]
	\label{MR-sufficient-SRP} 
	As the soundness proof for Algorithm \ref{alg:SPR-check}
	shows:
	If $\Cause$ does not satisfy (S),
	then there is an MR-scheduler $\sched$ for
	$\wminMDP{\cM}{\Cause}$ witnessing the violation of \eqref{SPR}.
	Scheduler $\sched$ 
	corresponds to a finite-memory (randomized) scheduler $\tsched$
	with two memory cells for $\cM$:
	``before $\Cause$'' (where $\tsched$ behaves as $\sched$)
	and ``after $\Cause$'' (where $\tsched$ behaves as an MD-scheduler minimizing
	the effect probability).
	\Ende
\end{rem}


\begin{lem}[Criterion for the existence of PR causes]
	\label{lem:existence-check}
	Let $\cM$ be an MDP and $\Effect$ a nonempty set of states. The following statements are equivalent:
	\begin{enumerate}
		\item[(a)]
		$\Effect$ has an SPR cause in $\cM$,
		\item[(b)]
		$\Effect$ has a GPR cause in $\cM$,
		\item[(c)]
		there is a state $c_0\in S \setminus \Effect$ such that
		the singleton $\{c_0\}$ is an SPR cause (and therefore a GPR cause) for $\Effect$ in $\cM$.
	\end{enumerate}
	Thus, the existence of SPR and GPR causes can be checked
	with Algorithm \ref{alg:SPR-check} in polynomial time.
\end{lem}  

\begin{proof}
	Obviously, statement (c) implies statements (a) and (b).
	The implication ``(a) $\Longrightarrow$ (b)'' follows from
	Lemma \ref{lemma:strict-implies-global}.
	We now turn to the proof of ``(b) $\Longrightarrow$ (c)''.
	For this, we assume that we are given a
	GPR cause $\Cause$ for $\Effect$ in $\cM$.
	For $c\in \Cause$, let $w_c=\Pr^{\min}_{\cM,c}(\Diamond \Effect)$.
	Pick a state $c_0\in \Cause$ such that
	$w_{c_0} = \max \{ w_c : c \in \Cause\}$.
	For every scheduler $\sched$ for $\cM$ that minimizes the
	effect probability whenever it visits a state in $\Cause$,
	and visits $\Cause$ with positive probability,
	the conditional probability
	$\Pr^{\sched}_{\cM}(\Diamond \Effect |\Diamond \Cause)$ is
	a weighted average of the values $w_c$, $c \in \Cause$,
	and thus bounded by $w_{c_0}$.
	Using Lemma \ref{lemma:wmin-criterion-PR-causes} we see that it is sufficient to only consider the minimal probabilities $w_c=\Pr^{\min}_{\cM,c}(\Diamond \Effect)$.
	Thus, we conclude that $\{c_0\}$ is both an SPR and a GPR cause for $\Effect$.
\end{proof}


\subsection{Checking the global probability-raising condition}

\label{sec:check-GPR-condition}
\label{sec:check-GPR}

Throughout this section, we suppose that both the effect set $\Effect$ and the cause candidate $\Cause$ are fixed disjoint subsets of the state space of the MDP $\cM=(S,\Act,P,\init)$, and address the task to check whether $\Cause$ is a global probability-raising cause for $\Effect$ in $\cM$. As the minimality condition (for all $c \in \Cause: \Pr^{\max}_\cM(\neg \Cause \until c) > 0)$ can be checked in polynomial time using a standard graph algorithm, we will concentrate on an algorithm to check the probability-raising condition \eqref{GPR}.
We start by stating the main results of this section.

\begin{thm}
	\label{thm:checking-GPR-in-poly-space}
	Given $\cM$, $\Cause$ and $\Effect$, deciding whether $\Cause$ is a GPR cause for $\Effect$ in $\cM$ can be done in $\coNP$.
\end{thm}

In order to provide an algorithm, we perform a model transformation  after which the violation of \eqref{GPR} by a scheduler $\sched$ can be expressed solely in terms of the expected frequencies of the state-action pairs of the transformed MDP under $\sched$. 
This allows us to express the existence of a scheduler witnessing the non-causality of $\Cause$ in terms of the satisfiability of a quadratic constraint system.
Thus, we can restrict the quantification in (G) to MR-schedulers in the transformed model.
We trace back the memory requirements to $\wminMDP{\cM}{\Cause}$ and to the original MDP $\cM$ yielding the second main result.

\begin{thm}
	\label{thm:MR-sufficient-GPR}
	Let $\cM$ be an MDP with effect set $\Effect$ as before and
	$\Cause$ a set of non-effect states which satisfies for all $c \in \Cause: \Pr^{\max}_\cM(\neg \Cause \until c) > 0$.
	If $\Cause$ is not a GPR cause for $\Effect$, then
	there is an MR-scheduler for $\wminMDP{\cM}{\Cause}$ refuting the GPR
	condition for $\Cause$
	in $\wminMDP{\cM}{\Cause}$
	and a finite-memory scheduler for $\cM$ with two memory cells
	refuting the GPR
	condition for $\Cause$
	in $\cM$.
\end{thm}


The remainder of this section is concerned with the proofs of both
Theorem \ref{thm:checking-GPR-in-poly-space} and
Theorem \ref{thm:MR-sufficient-GPR}.
For this, we suppose that $\Cause$ satisfies for all $c \in \Cause: \Pr^{\max}_\cM(\neg \Cause \until c) > 0$ which can be checked preemptively in polynomial time as argued before.

\paragraph*{\bf Checking the GPR condition (Proof of Theorem \ref{thm:checking-GPR-in-poly-space}). }
We will start with a polynomial-time model transformation into a kind of ``canonical form'' after which we can make the following assumptions when checking the GPR condition of $\Cause$ for $\Effect$

\begin{description}
	\item [(A1)]
	$\Effect=\{\effuncov,\effcov\}$ consists of two terminal states.
	
	\item [(A2)] %
	For every $c\in \Cause$, there is a single enabled
	action $\Act(c)=\{\gamma\}$, and
	there is $w_c\in [0,1]\cap \Rational$ such that
	$P(c,\gamma,\effcov)=w_c$ and
	$P(c,\gamma,\noeffc)=1{-}w_c$,
	where $\noeffc$ is a terminal non-effect state
	and $\noeffc$ and $\effcov$ are only accessible via
	$\gamma$-transition from the $c\in \Cause$.
	
	\item [(A3)]
	$\cM$ has no end components and
	there is a further terminal state $\noeffbot$
	and an action $\tau$ such that 
	$\tau \in \Act(s)$ implies $P(s,\tau,\noeffbot)=1$.
\end{description}
The terminal states $\effunc$, $\effcov$, $\noeffc$ and $\noeffbot$ 
are pairwise distinct. $\cM$ can have further terminal
states representing true negatives. 
However, these can be identified with $\noeffbot$.

Intuitively, $\effcov$ stands for covered effects (``$\Effect$ after $\Cause$'')
and can be seen as a true positive,
while $\effuncov$ represents the uncovered effects (``$\Effect$ without preceding $\Cause$'')
and corresponds to a false negative.
Let $\sched$ be a scheduler in $\cM$. Note that 
$\Pr^{\sched}_{\cM}((\neg \Cause)\until \Effect)=
\Pr^{\sched}_{\cM}(\Diamond \effuncov)$
and
$\Pr^{\sched}_{\cM}( \Diamond (\Cause \wedge \Diamond \Effect))=
\Pr^{\sched}_{\cM}(\Diamond \effcov)$.
As the cause states can not reach each other we also have
$\Pr^{\sched}_{\cM}((\neg \Cause)\Until c)= \Pr^{\sched}_{\cM}(\Diamond c)$
for each $c\in \Cause$.
The intuitive meaning of $\noeffc$ is a false positive (``no effect after $\Cause$''), while $\noeffbot$ stands for true negatives where neither the effect nor the cause is observed.
Note that
$\Pr^{\sched}_{\cM}( \Diamond (\Cause \wedge \neg \Diamond \Effect))= \Pr^{\sched}_{\cM}(\Diamond \noeffc)$
and
$\Pr^{\sched}_{\cM}( \neg \Diamond \Cause \wedge \neg \Diamond \Effect))= \Pr^{\sched}_{\cM}(\Diamond \noefftn)$.

\paragraph*{Establishing assumptions (A1)-(A3):}
We justify the assumptions as we can transform $\cM$ into a new MDP of the same asymptotic size satisfying the above assumptions.
Thanks to Lemma \ref{lemma:wmin-criterion-PR-causes},
we may suppose that $\cM=\wminMDP{\cM}{\Cause}$ without changing the satisfaction of (G).
Thus, from cause states $c \in \Cause$ there are only two outgoing transitions, either to a terminal effect state $\eff$ with probability $\Pr^{\min}_c(\lozenge \Eff)$ or to a terminal non-effect state $\noeff$ with the remaining probability (see Notation \ref{notation:MDP-mit-min-prob-ab-cause-candidates}).
We then may rename the effect state $\eff$ and the non-effect state $\noeff$ reachable from $\Cause$ into $\effcov$ and $\noeffc$, respectively.
Furthermore, we collapse all other effect states into a single state $\effunc$ and all true negative states into a single state $\noefftn$.
Similarly, by renaming and possibly duplicating terminal states we also suppose that $\noeffc$ has no other incoming transitions than the $\gamma$-transitions from the states in $\Cause$.
This ensures (A1) and (A2).
For (A3) consider the set $T$ of terminal states in the MDP obtained so far.
We remove all non-trivial end components by switching to the MEC-quotient \cite{deAlfaro1997}, i.e., we collapse all states that belong to the same MEC $\cE$ into a single state $s_{\cE}$ representing the MEC while ignoring the actions inside $\cE$.
Additionally, we add a fresh deterministic $\tau$-transition from the states $s_{\cE}$ representing MECs to $\noeffbot$
(i.e., $P(s_{\cE},\tau,\noeffbot) = 1$).
The $\tau$-transitions from states $s_{\cE}$ to $\noeffbot$ can then be used to mimic cases where the scheduler of the original MDP enters the end component $\cE$ and stays there forever.

In particular, consider the MEC-quotient $\cN$ of $\wminMDP{\cM}{\Cause}$ (see Definition \ref{app:MEC-quotient}). 
Let $\noeffbot$ be the state to which we add a $\tau$-transition with probability $1$ from each MEC that we collapse in the MEC-quotient.
That is, $\noeffbot=\bot$ with the notations of Definition \ref{def:MEC-contraction}.

	We demonstrate these transformations on the abstract MDP $\cM$ from Figure \ref{fig:establishing-assumptions-abstract-raw}, where the dotted circles correspond to sets of states in the MDP.
	The MDP already satisfies $\cM = \wminMDP{\cM}{c}$.
	We rename $\eff$ reachable from $\Cause$ to $\effcov$ and $\noeff$ to $\noeff_{\mathsf{fp}}$.
	Effect states not reachable from $\Cause$ collapse to $\effunc$.
	There are no terminal non-effect states not reachable from $c$, which would collapse to $\noeff_{\mathsf{tn}}$.
	The MEC quotient collapses MECs to states $s_{\cE_i}$ only keeping outgoing transitions.
	There is a fresh action $\tau$ in states $s_{\cE_i}$ to $\noeff_{\mathsf{tn}}$. 
	Thus, we get $\cN$ from Figure \ref{fig:establishing-assumptions-abstract-transformed}.
	\begin{figure}[t]
		\centering
		\begin{minipage}{0.45\textwidth}
			\centering
			\resizebox{\textwidth}{!}{
				\begin{tikzpicture}[scale=1,->,>=stealth',auto ,node distance=0.5cm, thick]
	\tikzstyle{round}=[thin,draw=black,circle]
	
	\node[scale=1, state] (init) {$\init$};
	\node[scale=5, ellipse, draw, below right=0.1 of init] (MDP) {\phantom{text}};
	\node[scale=1, above=-1 of MDP] (MDPtext) {MDP};
	\node[scale=0.2, circle, draw, below=0.5 of init] (a) {\phantom{a}};
	\node[scale=0.2, circle, draw, below right=0.5 of init] (b) {\phantom{a}};
	\node[scale=0.2, below right=2 of b] (c) {\phantom{a}};
	\node[scale=1, ellipse, draw, dotted, below=0.1 of c] (e) {$\mathsf{MEC}$};
	\node[scale=1, ellipse, draw, dotted, right=1 of c] (d) {$\mathsf{MEC}$};
	\node[scale=1, ellipse, draw, dotted, below=1 of b] (cause) {$\Cause$};
	\node[scale=1, ellipse, draw, below=1 of cause] (eff1) {$\eff$};
	\node[scale=1, ellipse, draw, left=0.2 of eff1] (noeff) {$\noeff$};
	\node[scale=1, ellipse, draw, right=1 of eff1] (eff2) {$\eff$};
	
	\draw[<-] (init) --++(-0.55,0.55);
	\draw[color=black ,->] (init) edge  (a) ;
	\draw[color=black ,->] (init) edge  (b) ;
	
	\draw[->] (a) --++(0.35,-0.35);
	\draw[->] (a) --++(0.5,0);
	
	\draw[->] (b) --++(0.35,-0.35);
	\draw[->] (b) --++(0.35,0.35);
	
	\draw[<-] (cause) --++(0,0.9);
	\draw[->] (cause) edge (eff1);
	\draw[->] (cause) edge (noeff);
	
	\draw[<-] (e) --++(0,0.9);
	\draw[->] (e) edge node[pos=0.7] {$\alpha$} (d);
	\draw[->] (e) edge node[pos=0.3, below left] {$\beta$} (eff2);

	\draw[<-] (d) --++(-0.65,0.65);
	\draw[->] (d) --++(0.65,0.65);
	
\end{tikzpicture}
			}
			\caption{MDP $\cM$ not satisfying \\ assumptions (A1)-(A3)}
			\label{fig:establishing-assumptions-abstract-raw}
		\end{minipage}
		\hfill
		\begin{minipage}{0.53\textwidth}
			\centering
			\resizebox{\textwidth}{!}{
				\begin{tikzpicture}[scale=1,->,>=stealth',auto ,node distance=0.5cm, thick]
	\tikzstyle{round}=[thin,draw=black,circle]
	
	\node[scale=1, state] (init) {$\init$};
	\node[scale=5, ellipse, draw, below right=0.1 of init] (MDP) {\phantom{text}};
	\node[scale=1, above=-1 of MDP] (MDPtext) {MDP};
	\node[scale=0.2, circle, draw, below=0.5 of init] (a) {\phantom{a}};
	\node[scale=0.2, circle, draw, below right=0.5 of init] (b) {\phantom{a}};
	\node[scale=0.2, below right=2 of b] (c) {\phantom{a}};
	\node[scale=1, state, below=0.01 of c] (e) {$s_{\cE_0}$};
	\node[scale=1, state, right=1 of c] (d) {$s_{\cE_1}$};
	\node[scale=1, ellipse, draw, dotted, below=1 of b] (cause) {$\Cause$};
	\node[scale=1, ellipse, draw, below=1 of cause] (eff1) {$\effcov$};
	\node[scale=1, ellipse, draw, left=0.2 of eff1] (noeff) {$\noeff_{\mathsf{fp}}$};
	\node[scale=1, ellipse, draw, right=1 of eff1] (eff2) {$\effunc$};
	\node[scale=1, ellipse, draw, right=0.2 of eff2] (noeff1) {$\noeff_{\mathsf{tn}}$};
	
	\draw[<-] (init) --++(-0.55,0.55);
	\draw[color=black ,->] (init) edge  (a) ;
	\draw[color=black ,->] (init) edge  (b) ;
	
	\draw[->] (a) --++(0.35,-0.35);
	\draw[->] (a) --++(0.5,0);
	
	\draw[->] (b) --++(0.35,-0.35);
	\draw[->] (b) --++(0.35,0.35);
	
	\draw[<-] (cause) --++(0,0.9);
	\draw[->] (cause) edge (eff1);
	\draw[->] (cause) edge (noeff);
	
	\draw[<-] (e) --++(0,0.9);
	\draw[->] (e) edge node[pos=0.7] {$\alpha$} (d);
	\draw[->] (e) edge node[pos=0.3, below left] {$\beta$} (eff2);
	\draw[->] (e) edge node[pos=0.2] {$\tau$} (noeff1);
	
	\draw[->] (d) edge node[pos=0.2] {$\tau$} (noeff1);
	\draw[<-] (d) --++(-0.65,0.65);
	\draw[->] (d) --++(0.65,0.65);
	
\end{tikzpicture}
			}
			\caption{Transformed MDP $\cN$ satisfying \\ assumptions (A1)-(A3)}
			\label{fig:establishing-assumptions-abstract-transformed}
		\end{minipage}
	\end{figure}

The following Lemma \ref{lem:probabilities_MEC-quotient} and Corollary \ref{cor:GPR_MEC} prove the soundness of the model transformation.

\begin{lem}\label{lem:probabilities_MEC-quotient}
	For each scheduler $\sched$ for $\wminMDP{\cM}{\Cause}$, there is a scheduler $\tsched$ for $\cN$, and vice versa, such that 
	\begin{itemize}
		\item
		$\Pr^{\sched}_{\wminMDP{\cM}{\Cause}}(\Diamond \Effect) = \Pr^{\tsched}_{\cN}(\Diamond \Effect)$,
		\item
		$\Pr^{\sched}_{\wminMDP{\cM}{\Cause}}(\Diamond \Cause) = \Pr^{\tsched}_{\cN}(\Diamond \Cause)$, and
		\item 
		$\Pr^{\sched}_{\wminMDP{\cM}{\Cause}}(\Diamond \Cause \land \Diamond \Effect) = \Pr^{\tsched}_{\cN} (\Diamond \effcov)$.
	\end{itemize}
\end{lem}
\begin{proof}
	By Lemma \ref{lem:MEC-contraction-terminal-states-final}, there is a scheduler $\tsched$ for $\cN$ for each scheduler $\sched$ for $\wminMDP{\cM}{\Cause}$ such that each terminal state is reached with the same probability under $\tsched$ in $\cN$ and under $\sched$ in $\wminMDP{\cM}{\Cause}$.
	The state $\effcov$ is present in $\wminMDP{\cM}{\Cause}$ under the name $\eff$.
	The state $\eff$ is furthermore reached in $\wminMDP{\cM}{\Cause}$ if and only if $\Diamond \Cause \land \Diamond \Effect$ is satisfied along a run.
	The set of terminal states in $\Effect$ is obtained from the set $\Effect$ in $\wminMDP{\cM}{\Cause}$ by collapsing states. As a scheduler $\sched$ can be viewed as a scheduler for both MDPs and these MDPs agree except for the terminal states, the first equality follows as well. 
	As the probability to reach $\Cause$ is the sum of the probabilities to reach the terminal states $\effcov$ and $\noeffc$ in $\cN$ and as these states are only renamed in $\cN$ in comparison to $\wminMDP{\cM}{\Cause}$, the claim follows. 
\end{proof}

From Lemma \ref{lem:probabilities_MEC-quotient} and Lemma \ref{lemma:wmin-criterion-PR-causes}, we conclude the following corollary that justifies working under assumptions (A1)-(A3).

\begin{cor}\label{cor:GPR_MEC}
	The set $\Cause$ is a GPR cause for $\Effect$ in $\cM$ if and only if $\Cause$ is a GPR cause for $\Effect$ in $\cN$.
\end{cor}

\begin{proof}
	By Lemma \ref{lem:probabilities_MEC-quotient}, for each scheduler $\sched$ for $\wminMDP{\cM}{\Cause}$, there is a scheduler $\tsched$ for $\cN$ such that all relevant probabilities agree, and vice versa. So, $\Cause$ is a GPR cause for $\Effect$ in $\wminMDP{\cM}{\Cause}$ if and only if it is a GPR cause in $\cN$.
	By Lemma \ref{lemma:wmin-criterion-PR-causes}, 
	$\Cause$ is a GPR cause for $\Effect$ in $\wminMDP{\cM}{\Cause}$ if and only if it is a GPR cause in $\cM$.
\end{proof}

Note, however, that the transformation changes the memory-requirements of schedulers witnessing that $\Cause$ is not a GPR cause for $\Effect$. We will address the memory requirements in the original MDP later.
With assumptions (A1)-(A3), condition (G) can be reformulated as follows:

\begin{lem}
	\label{lem:GPR-poly-constraint}
	Under assumptions (A1)-(A3), $\Cause$ satisfies (G)
	if and only if for
	each scheduler $\sched$ with $\Pr_{\cM}^{\sched}(\Diamond \Cause) >0$
	the following condition holds:
	\begin{align*}
		\label{GPR-1}
		\Pr_{\cM}^{\sched}(\Diamond \Cause) \cdot \Pr^{\sched}_{\cM}( \Diamond \effuncov)
		\ < \ 
		\bigl(1{-}\Pr_{\cM}^{\sched}(\Diamond \Cause)\bigr)
		\cdot 
		\!\!\!\!\!\!\sum_{c\in \Cause} \!\!\!\!\!\!
		\Pr^{\sched}_{\cM}(\Diamond c) \cdot w_c
		\tag{GPR-1}
	\end{align*}
\end{lem}

With assumptions (A1)-(A3),  a terminal state of $\cM$ is reached almost surely under any scheduler after finitely many steps in expectation.
Given a scheduler $\sched$ for $\cM$ recall the definition of expected frequencies of 
state action-pairs $(s, \alpha)$, states $s\in S$ and state-sets $T \subseteq S$
under $\sched$:
\begin{align*}
	\freq{\sched}{s,\alpha} & \ \ \eqdef \ \
	\mathrm{E}^{\sched}_{\cM}(\text{number of visits to $s$ in which $\alpha$ is taken})\\
	\freq{\sched}{s} &\ \ \eqdef
	\sum\nolimits_{\alpha \in \Act(s)} \freq{\sched}{s,\alpha},
	\qquad	\freq{\sched}{T} \eqdef \sum\nolimits_{s \in T} \freq{\sched}{s}.
\end{align*}
Let $T$ be one of the sets $\{\effcov\}$, $\{\effunc\}$, $\Cause$, or a singleton $\{c\}$ with $c\in \Cause$.
As $T$ is visited at most once during each run of $\cM$ (assumptions (A1) and (A2)), we have 
$\Pr^{\sched}_{\cN}(\Diamond T) = \freq{\sched}{T}$
for each scheduler $\sched$.
This allows us to express the violation of (G) in terms of a quadratic constraint system over variables for the expected frequencies of state-action pairs.
Let $\SA$ denote the set of state-action pairs in $\cM$. We consider the following constraint system over the variables $x_{s,\alpha}$ for each $(s,\alpha)\in \SA$ where we use the short form notation
$x_s = \sum_{\alpha \in \Act(s)}x_{s,\alpha}$:
\begin{align*}
	x_{s,\alpha} & \ \geqslant \ 0 \qquad \text{for all $(s,\alpha) \in \SA$}
	\tag{S1}
	\\[1ex]
	x_{\init} & \ = \
	1+ \!\!\!\!\! \sum_{(t,\alpha) \in \SA} \!\!\!\!\!
	x_{t,\alpha}\cdot P(t,\alpha,\init)
	\tag{S2}\\[1ex]
	x_{s} & \ =  \sum_{(t,\alpha) \in \SA} \!\!\!\!\! x_{t,\alpha}\cdot P(t,\alpha,s)
	\qquad \text{for all $s\in S\setminus\{\init\}$}
	\tag{S3}
\end{align*}
Using well-known results for MDPs without ECs (see, e.g., \cite[Theorem 9.16]{Kallenberg20}), given a vector $x\in \Real^{\SA}$, then $x$ is a solution to (S1) and the balance equations (S2) and (S3)
if and only if there is a (possibly history-dependent) scheduler $\sched$ for $\cM$ with $x_{s,\alpha}=\freq{\sched}{s,\alpha}$ for all $(s,\alpha)\in \SA$
if and only if there is an MR-scheduler $\sched$ for $\cM$ with $x_{s,\alpha}=\freq{\sched}{s,\alpha}$ for all $(s,\alpha)\in \SA$.

The violation of \eqref{GPR-1} in Lemma \ref{lem:GPR-poly-constraint}
and the condition $\Pr^{\sched}_{\cM}(\Diamond \Cause)>0$
can be reformulated in terms of the frequency-variables as follows where $x_{\Cause}$ is an abbreviation for $\sum_{c\in \Cause} x_c$:
\begin{align*}
	\label{non-GPR}
	0
	\ \ \geqslant \ \
	\bigl(1-  x_{\Cause} \bigl) \cdot \!\!\!\!
	&\sum_{c \in \Cause} \!\!\!\!\!\! x_{c} \cdot w_c \ - \ x_{\Cause} \cdot x_{\effuncov}
	\tag{S4} \\[1ex]
	&x_\Cause > 0 \tag{S5}
\end{align*}

\begin{lem}
	\label{prop:quadradic-system}
	Under assumptions (A1)-(A3), the set $\Cause$ is not a GPR cause for $\Effect$ in $\cM$ iff the constructed quadratic system of inequalities (S1)-(S5) has a solution.
\end{lem}

We can now prove our first main Theorem \ref{thm:checking-GPR-in-poly-space} of this subsection, stating that deciding the GPR condition can be done in $\coNP$.

\begin{proof}[Proof of Theorem \ref{thm:checking-GPR-in-poly-space}]
	The quadratic system of inequalities can be constructed from $\cM$, $\Cause$, and $\Effect$ in polynomial time.
	Except for the strict inequality constraint in (S5), it has the form of a quadratic program, for which the threshold problem can be decided in $\NP$ by \cite{Vavasis1990}.
	We will prove that also with this strict inequality, it can be checked in $\NP$ whether the system (S1)-(S5) has a solution.
	As the system of inequalities is expressing the violation of \eqref{GPR}, deciding whether a set $\Cause$ is a GPR cause can then be done in $\coNP$.
	
	To show that satisfiability of the system (S1)-(S5) is in $\NP$, we will provide a non-deterministic algorithm that runs in polynomial time and finds a solution if one exists. 
	Some of the arguments are similar to the arguments used in \cite{Vavasis1990}. Additionally, we will rely on the implicit function theorem.
	
	We begin by proving what a solution to (S1)-(S5) can be assumed to look like.
	Thus assume that a solution to (S1)-(S5) exists. 
	There are two possible cases:
	\begin{description}
		\item[Case 1]
		All solutions to (S1)-(S3) and (S5) satisfy (S4). Then, in particular, the frequency values of an MD-scheduler maximizing the probability to reach Cause are a solution to (S1)-(S3) and (S5) and hence to (S4) in this case. 
		\item[Case 2]
		There are solutions to (S1)-(S3) and (S5) that violate (S4). The space of feasible points for conditions (S1)-(S3) and (S5) is connected.
		The right hand side of (S4)
		\[
		\bigl(1-  x_{\Cause} \bigl) \cdot \!\!\!\!
		\sum_{c \in \Cause} \!\!\!\!\!\! x_{c} \cdot w_c \ - \ x_{\Cause} \cdot x_{\effuncov}
		\]
		is continuous. Hence, as there are also solutions to (S1)-(S3) and (S5) that  satisfy (S4) by assumption, there is a solution to (S1)-(S3) and (S5) that satisfies
		\[
		\bigl(1-  x_{\Cause} \bigl) \cdot \!\!\!\!
		\sum_{c \in \Cause} \!\!\!\!\!\! x_{c} \cdot w_c \ - \ x_{\Cause} \cdot x_{\effuncov} = 0. \tag{S4'}
		\]
	\end{description}
	Now, let us take a closer look at Case 2:
	First of all, we add the equation
	\[
	x_{\Cause} = \sum_{c\in \Cause} x_c \tag{S6}
	\]
	to our system.
	Thus, the variables are $x_{\Cause}$ and  $x_{s,\alpha}$ for each $(s,\alpha)\in \SA$. 
	Obviously, this does not influence the satisfiability.
	Equation (S4') now contains the new variable $x_{\Cause}$, which is not an abbreviation anymore.
	We write $x$ for the vector of variables $x_{s,\alpha}$ with $(s,\alpha)\in \SA$.
	
	In Case 2,	there is a solution $(x^\ast,x_{\Cause}^\ast)$  such that the maximal possible number of variables is $0$ 	and such that $x_{\Cause}^\ast$ is maximal among all such solutions.
	Let $X^\prime$ be the set of variables that are $0$ in $(x^\ast,x_{\Cause}^\ast)$.
	We remove all variables from $X^\prime$ from all constraints by setting them to $0$
	and call the resulting system (T1)-(T6) where (T4) is obtained from (S4'), while all other equations (T$i$) are obtained from (S$i$), by removing the chosen variables.
	We then collect the remaining variables in the vector $v=(y,y_{\Cause})$.
	Let $(y^\ast,y_{\Cause}^\ast)$ be the solution $(x^\ast,x_{\Cause}^\ast)$ after the variables in $X^\prime$ have been removed.
	Thus, all values in this vector are positive.
	
	Define the function $f$ as the right hand side of (T4):
	\[
	f(y,y_{\Cause})=\bigl(1-  y_{\Cause} \bigl) \cdot \!\!\!\!
	\sum_{c \in \Cause} \!\!\!\!\!\! y_{c} \cdot w_c \ - \ y_{\Cause} \cdot y_{\effuncov},\]
	where the variables $y$ are as the original variables $x$ after the  variables in $X^\prime$ have been removed.
	
	Now, we apply the implicit function theorem:
	Observe that
	\[
	\frac{\partial f(y,y_{\Cause})}{\partial y_{\Cause}} = -\sum_{c \in \Cause}  y_{c} \cdot w_c - y_{\effuncov}.
	\]
	Evaluated at $(y^\ast,y_{\Cause}^\ast)$, this value is non-zero as all summands are negative and there are at least some of the variables in the abbreviation $y_{c} $ with $c\in \Cause$ left, i.e., not removed because they were not $0$ due to the original constraint (S5).
	So, we can apply the implicit function theorem, which guarantees us the existence of a function $g(y)$, such that $g(y^\ast)=y_{\Cause}^\ast$ and, for all $y^\prime$ in an open ball $B_1$ around $y^{\ast}$,  we have
	\[
	f(y^\prime,g(y^\prime))=0.
	\] 
	By the implicit function theorem, we can explicitly compute the gradient 
	\[
	\nabla g = \left( \frac{\partial  g(y)}{\partial {y_1}},\dots,  \frac{\partial g(y)}{\partial {y_k}} \right) = - \left(\frac{\partial f(y,y_{\Cause})}{\partial y_{\Cause}}\right)^{-1}\cdot 
	\left( \frac{\partial  f(y,y_{\Cause})}{\partial {y_1}},\dots,  \frac{\partial f(y,y_{\Cause})}{\partial {y_k}} \right)
	\]
	of the derivatives on $B_1$
	for the appropriate $k$ from the derivatives of $f$. 
	Note that on $B_1$, the gradient  $\nabla g$ is $0$ iff 
	\[
	H(y,y_{\Cause})\eqdef \left( \frac{\partial  f(y,y_{\Cause})}{\partial {y_1}},\dots,  \frac{\partial f(y,y_{\Cause})}{\partial {y_k}} \right)
	\]
	is $0$. Furthermore, all entries of $H(y,y_{\Cause})$ are linear in the variables $v$ as the function $f$ is quadratic.
	As the function $g$ has a local maximum in $y^\ast$, we know that $\nabla g$ evaluated at $y^\ast$ is $0$.

	Equations (T2), (T3), and (T6) are linear equations in the remaining variables $v$.
	We can rewrite these three equations with a matrix $M$ and a vector $b$ whose entries can easily be expressed in terms of the coefficients of the original system  (again, after the set of variables $X^\prime$ has been removed) as
	\[
	Mv=b.
	\]
	The solutions to this equation form an $r$-dimensional affine space $W$. It can be written as
	\[
	W=\{\underbrace{c_0+c_1\cdot z_1+\dots+c_r\cdot z_r}_{\eqdef h(z_1,\dots,z_r)} \mid (z_1,\dots,z_r)\in \mathbb{R}^r\}
	\]
	for some vectors $c_0,c_1, \dots, c_r$ which can be computed from $M$ and $b$ in polynomial time.
	
	Let $B_2$ be an open ball in $\mathbb{R}^r$ such that $h(B_2)\subseteq B_1$ and such that $h(B_2)$ contains $(y^\ast,y_{\Cause}^\ast)$.
	We claim that $g\circ h\colon B_2 \to \mathbb{R}$ has an isolated local maximum at $z^\ast \eqdef h^{-1}(y^\ast,y_{\Cause}^\ast)$.
	It is clear that $g \circ h$ has a local maximum since $g$ has a local maximum at $(y^\ast,y_{\Cause}^\ast)$.
	Suppose now, that $g\circ h$ does not have an isolated local maximum at $h^{-1}(y^\ast,y_{\Cause}^\ast)$.
	As $h$ is an  affine map and the graph of $g$ is the solution to a quadratic equation, this is only possible if there is a direction $d\in \mathbb{R}^r\setminus \{0\}$ such that
	\begin{align}
		\label{eq:GPR-check-local-maximum}
		g\circ h (z^\ast) = g\circ h (z^\ast + t \cdot d)
	\end{align}
	for all $t\in \mathbb{R}$. Due to the boundedness of the polyhedron described by conditions (T1)-(T3), (T5) and (T6) and since $z^\ast$ lies in the interior of this polyhedron,  this means that there must be a value $q\in \mathbb{R}$ such that $(h (z^\ast + q \cdot d), g \circ h (z^\ast + q \cdot d)) $ provides a solution $v$ to equations (T1)-(T3), (T5) and (T6) with an additional $0$.
	By the definition of $g$, this solution furthermore satisfies (T4) and, by equation \eqref{eq:GPR-check-local-maximum}, it still satisfies (T5). This contradicts the choice of the original 
	solution $(x^\ast,x_{\Cause}^\ast)$.
	
	So, $g\circ h\colon B_2 \to \mathbb{R}$ has an isolated local maximum at $z^\ast$. This implies that on an open ball around $z^\ast$, the point $z^\ast$ is the only solution to 
	\[
	\nabla g (h(z))=0
	\]
	and consequently, the only solution to 
	\[
	H(h(z))=0.
	\]
	Since $H(h(z))$ is a vector of linear expressions in $z$, this implies that $z^\ast$ is the only solution on $\mathbb{R}^r$ to $H(h(z))=0$. 
	This is the key result that we need to provide a non-deterministic polynomial-time algorithm to check the satisfiability of the original constraint system.
	
	Let us now describe the algorithm:
	The algorithm begins by computing the frequency values of an MD-scheduler as in Case 1 in polynomial time and checks whether the resulting vector of frequency values satisfies (S1)-(S5).
	If this is the case, the algorithm returns that the system is satisfiable.
	
	If this is not the case, the algorithm tries to compute a solution to (S1)-(S3), (S5), and (S4') as in Case 2.
	The algorithm non-deterministically guesses a subset  of the variables and removes them 
	from all constraints by replacing them with $0$. 
	
	Suppose we guess the set $X^{\prime}$ as above.
	We show that we then compute a solution.
	After the variables from $X^\prime$ have been removed, $H(y,y_{\Cause})$ can be computed in polynomial time as all the derivatives of $f$ are linear expressions in the variables which require basic arithmetic and can be computed in polynomial time.
	Likewise, $M$ and $b$ can be computed in polynomial time from the original constraints after the guessed variables have been removed.
	The vectors $c_0,c_1, \dots, c_r$ describing the solution space to $Mv=b$ can then also be computed in polynomial time.
	
	Thus, also the vector $H(h(z))$ of linear expressions in the variables $z$ can be computed in polynomial time.
	The equation system $H(h(z))=0$ has a unique solution if the guessed variables were indeed $X^\prime$.
	In this case, the solution $z^\ast$ can be computed in polynomial time as well. 
	If the guess of variables was not $X^\prime$, then either there is no unique solution to this equation system which can be detected in polynomial time, or the solution, which is computed in the sequel in polynomial time, might not satisfy the original constraints, which is checked in the end. 
	
	From $z^\ast$, we can compute  $y^\ast=h(z^\ast)$ using the vectors $c_0,c_1,\dots ,c_r$. The solution
	$x^\ast$ is then  obtained by plugging in $0$s for the removed variables. Checking whether the resulting vector satisfies all constraints can also be done in polynomial time in the end.
	If $X^\prime$ was guessed correctly, this vector $x^\ast$ indeed forms a solution to the original constraints as we have seen.
	
	In summary, the algorithm needs to guess the set $X^\prime$ of variables which are $0$ in a solution to the original constraints with the maximal number of zeroes.
	All other steps are deterministic polynomial-time computations.
	Thus, satisfiability of (S1)-(S5) can be checked in NP.
\end{proof}


\noindent{\bf Memory requirements of schedulers in the original MDP (Proof of Theorem \ref{thm:MR-sufficient-GPR}). }
Every solution to the
linear system of inequalities (S1), (S2), and (S3) corresponds to the
expected frequencies of state-action pairs of an MR-scheduler in the
transformed model satisfying (A1)-(A3). Hence:

\begin{cor}\label{cor:MR-scheduler}
	Under assumptions (A1)-(A3),
	$\Cause$ is no GPR cause for
	$\Effect$ iff there exists an MR-scheduler $\tsched$
	with $\Pr^{\tsched}_{\cM}(\Diamond \Cause) >0$
	violating \eqref{GPR}.
\end{cor}

The model transformation we used for assumptions (A1)-(A3), however, does affect the memory requirements of a violating scheduler.
We may further restrict the MR-schedulers necessary to witness non-causality under assumptions (A1)-(A3).
For the following lemma, recall that $\tau$ is the action of the MEC quotient used for the extra transition from states representing MECs to a new trap state (see also assumption (A3)).

\begin{lem}
	\label{lem:MR-sufficient-global-cause}
	Assume (A1)-(A3).
	Given an MR-scheduler $\usched$ with $\Pr^{\usched}_{\cM}(\Diamond \Cause) >0$ that violates \eqref{GPR},
	an MR-scheduler $\tsched$ with $\tsched(s)(\tau)\in \{0,1\}$ for each state $s$ with $\tau \in \Act(s)$
	that satifies $\Pr^{\tsched}_{\cM}(\Diamond \Cause) >0$ and violates \eqref{GPR}
	is computable in polynomial time.
\end{lem}

\begin{proof}
	Let $\usched$ be a scheduler with
	$\Pr^{\usched}_{\cM}(\Diamond \Cause) > 0$ violating \eqref{GPR-1}, i.e.:
	\[
	\Pr_{\cM}^{\usched}(\Diamond \Cause) \cdot
	\Pr^{\usched}_{\cM}( \Diamond \effuncov)
	\ < \ 
	\bigl(1{-}\Pr_{\cM}^{\usched}(\Diamond \Cause)\bigr)
	\cdot 
	\!\!\!\!\!\sum_{c\in \Cause} \!\!\!\!\!
	\Pr^{\usched}_{\cM}(\Diamond c) \cdot w_c.
	\]	
	We will show  how to transform $\usched$ into an
	MR-scheduler $\tsched$ that schedules the $\tau$-transitions to
	$\noeffbot$ with probability 0 or 1.
	We regard the set $U$ of states
	$u$ which have a $\tau$-transition to $\noeffbot$
	(recall that then $P(u,\tau,\noeffbot)=1$)
	and where $0 < \usched(u)(\tau) < 1$.
	We now process the $U$-states in an arbitrary order,
	say $u_1,\ldots,u_k$,
	and generate a sequence $\tsched_0=\usched,\tsched_1,\ldots,\tsched_k$
	of MR-schedulers such that for $i \in \{1,\ldots,k\}$:
	\begin{itemize}
		\item
		$\tsched_i$ refutes \eqref{GPR}
		(or equivalently condition \eqref{GPR-1} from
		Lemma \ref{lem:GPR-poly-constraint})
		\item 
		$\tsched_i$
		agrees with $\tsched_{i-1}$ for all states but $u_i$,
		\item
		$\tsched_i(u_i)(\tau)\in \{0,1\}$.
	\end{itemize}
	Thus, the final scheduler $\tsched_k$ satisfies the
	desired properties.
	
	To explain how to derive $\tsched_i$ from $\tsched_{i-1}$,
	let $i\in \{1,\ldots,k\}$, $\vsched=\tsched_{i-1}$, $u=u_i$
	and $y=1{-}\vsched(u)(\tau)$. By definition we have, $0<y<1$ (as $u\in U$ and
	by definition of $U$) and
	$y= \sum_{\alpha\in \Act(u)\setminus \{\tau\}} \vsched(u)(\alpha)$.
	For $x\in [0,1]$, let $\vsched_x$ denote the MR-scheduler that
	agrees with $\vsched$ for all states but $u$, for which
	$\vsched_x$'s decision is:
	\[
	\vsched_x(u)(\tau)=1{-}x,
	\qquad
	\vsched_x(u)(\alpha)=\vsched(u)(\alpha) \cdot \frac{x}{y}
	\quad \text{for $\alpha \in \Act(u)\setminus \{\tau\}$}
	\] 
	Obviously, $\vsched_y=\vsched$.
	We now show that at least one of the two MR-schedulers
	$\vsched_0$ or $\vsched_1$ also refutes \eqref{GPR}.
	For this, we suppose by contraction that this is not the case,
	which means that \eqref{GPR} holds for both.	
	Let $f : [0,1]\to \mathbb{Q}$ be defined by
	\[
	f(x) \ = \
	\Pr_{\cM}^{\vsched_x}(\Diamond \Cause) \cdot
	\Pr^{\vsched_x}_{\cM}( \Diamond \effuncov)
	- 
	\bigl(1{-}\Pr_{\cM}^{\vsched_x}(\Diamond \Cause)\bigr)
	\cdot 
	\!\!\!\!\!\sum_{c\in \Cause} \!\!\!\!\!\!
	\Pr^{\vsched_x}_{\cM}(\Diamond c) \cdot w_c
	\]
	As $\vsched = \vsched_y$ violates \eqref{GPR-1},
	while $\vsched_0$ and $\vsched_1$
	satisfy \eqref{GPR-1} we obtain:
	\[
	f(0), f(1) < 0  \qquad \text{and} \qquad f(y) \geqslant 0
	\]
	We now split $\Cause$ into the set $C$ of states $c\in \Cause$
	such that there is a $\vsched$-path from $\init$ to $c$ that traverses
	$u$ and $D=\Cause \setminus C$.
	Thus, $\Pr^{\vsched_x}_{\cM}(\Diamond \Cause)= p_x + p$
	where $p_x=\Pr^{\vsched_x}(\Diamond C)$ and
	$p=\Pr^{\vsched}(\Diamond D)$.
	Similarly, $\Pr^{\vsched_x}_{\cM}(\Diamond \effunc)$ has the form
	$q_x + q$ where
	$q_x=\Pr^{\sched_x}_{\cM}(\Diamond (u \wedge \Diamond \effunc))$
	and
	$q= \Pr^{\sched_x}_{\cM}((\neg u) \Until \effunc)$.
	With $p_{x,c}=\Pr^{\vsched_x}_{\cM}(\Diamond c)$
	for $c\in C$
	and $p_d = \Pr^{\vsched}_{\cM}(\Diamond d)$ for $d\in D$,
	let
	\[
	v_x  = \sum_{c\in C} p_{x,c}\cdot w_c
	\qquad \text{and} \qquad 
	v = \sum_{d\in D} p_{d}\cdot w_d
	\]
	As $y$ is fixed, the values
	$p_{y},p_{y,c},q_{y},v_{y}$ can be seen as
	constants.
	Moreover, $p_x,p_{x,c},q_x,v_x$ differ from
	$p_y,p_{y,c},q_y,v_y$
	only by the factor $\frac{x}{y}$.
	That is:
	\begin{center}
		$p_x=p_y\frac{x}{y}$, \ \
		$p_{x,c}=p_{y,c}\frac{x}{y}$, \ \   
		$q_x=q_y\frac{x}{y}$ \ \ and \ \
		$v_x=v_y\frac{x}{y}$.
	\end{center}
	Thus,
	$f(x)$ has the following form:
	\begin{eqnarray*}
		f(x) & = &
		(p_x{+}p)(q_x{+}q) - \bigl(1{-}(p_x{+}p)\bigr)(v_x{+}v)
		\\[1ex]
		& = &
		\underbrace{p_x q_x {+} p_xv_x}_{\mathfrak{a}x^2} +
		\underbrace{p_x(q+v) + q_xp - v_x}_{\mathfrak{b}x} +
		\underbrace{pq -v +pv}_{\mathfrak{c}}
		\\[1ex]
		& = & \mathfrak{a}x^2 + \mathfrak{b}x + \mathfrak{c}
	\end{eqnarray*}
	For the value $\mathfrak{a}$, we have
	$\mathfrak{a}x^2=p_x q_x {+} p_x v_x $ and hence
	$\mathfrak{a}= \frac{1}{y^2}(p_yq_y + p_yv_y)>0$.
	But then the second derivative $f''(x)=2\mathfrak{a}$ of
	$f$ is positive, which yields
	that $f$ has a global minimum at some point $x_0$ and is strictly
	decreasing for $x < x_0$ and strictly increasing for $x> x_0$.
	Since $f(0)$ and $f(1)$ are both negative, we obtain
	$f(x) <0$ for all $x$ in the interval $[0,1]$.
	But this contradicts $f(y) \geqslant 0$.
	
	This yields that at least one of the schedulers $\vsched_0$
	or $\vsched_1$ witnesses the violation of \eqref{GPR}.
	Thus, we can define $\tsched_i\in \{\vsched_0,\vsched_1\}$ accordingly.
	
	The number of states $k$ in $U$ is bounded by the number of states in $S$. In each iteration of the above construction, the function value $f(0)$ is sufficient to determine one of the schedulers $\vsched_0$
	and $\vsched_1$ witnessing the violation of \eqref{GPR}.
	So, the procedure has to compute the values in condition \eqref{GPR-1} for $k$-many MR-schedulers and update the scheduler afterwards.
	As the update can easily be carried out in polynomial time, the run-time of all $k$ iterations is polynomial as well.
\end{proof}

The condition that $\tau$ only has to be scheduled  with probability $0$ or $1$ in each state is the key to transfer the sufficiency of MR-schedulers to the MDP $\wminMDP{\cM}{\Cause}$. %
This fact is of general interest as well and stated in the following theorem where $\tau$ again is the action added to move from a state $s_{\cE}$ to the new trap state in the MEC-quotient.

\begin{thm}
	\label{thm:MR-schedulers-MEC-quotient}\label{thm:MRscheduler_lift}
	Let $\cM$ be an MDP with pairwise disjoint action sets for all states. 
	Then, for each MR-scheduler $\sched$ for the MEC-quotient of $\cM$ with
	$\sched(s_{\cE})(\tau)\in \{0,1\}$ for each MEC $\cE$ of $\cM$ there is an MR-scheduler $\tsched$ for $\cM$ such that every action  $\alpha$ of $\cM$ that does not belong to an MEC of $\cM$, has the same expected frequency under $\sched$ and $\tsched$.
\end{thm}

\begin{proof}
	Let $\sched$ be an MR-scheduler for $\MEC{\cM}$ with $\sched(s_{\cE})(\tau) \in \{0,1\}$ for each MEC $\cE$ of $\cM$.
	We consider the following extension $\cM^\prime$ of $\cM$: The state space of $\cM$ is extended by a new terminal state $\bot$ and 
	a fresh action $\tau$ is enabled in each state $s$ that belongs to a MEC of $\cM$. Action $\tau$ leads to $\bot$ with probability $1$. All remaining transition probabilities are as in $\cM$.
	So, $\cM^\prime$ is obtained from $\cM$ by allowing a transition to a new terminal state $\bot$ from each state that belongs to a MEC.
	
	Now, we first provide a finite-memory scheduler $\tsched$ for $\cM^\prime$ that leaves each MEC $\cE$ for which $\sched(s_{\cE})(\tau)=0$ via the state action pair $(s,\alpha)$ with probability $\sched(s_{\cE})(\alpha)$. Recall that we assume that each action is enabled in at most one state and that the actions enabled in the state $s_{\cE}$ in $\MEC{\cM}$ are precisely the actions that are enabled in some state of $\cE$ and that do not belong to $\cE$ (Section~\ref{app:MEC-quotient}).
	
	Let us define the scheduler $\tsched$:
	In all states that do not belong to a MEC $\cE$ of $\cM$ with $\sched(s_{\cE})(\tau)=0$, the behavior of $\tsched$ is memoryless:
	For each state $s$ of $\cM$ (and hence of $\cM^\prime$) that does not belong to a MEC, $\tsched(s)=\sched(s)$. For each state $s$ in an end component 
	$\cE$ of $\cM$ with $\sched(s_{\cE})(\tau)=1$, we define $\tsched(s)(\tau)=1$.
	If a MEC $\cE$ of $\cM$ with $\sched(s_{\cE})(\tau)=0$ is entered, $\tsched$ makes use of finitely many memory modes as follows:
	Enumerate the state action pairs $(s,\alpha)$ where $s$ belongs to $\cE$, but $\alpha$ does not belong to $\cE$, and for which $\sched(s_{\cE})(\alpha)>0$ by 
	$(s_1,\alpha_1)$, \dots, $(s_k,\alpha_k)$.
	Further, let $p_i\eqdef \sched(s_{\cE})(\alpha_i)>0$ for all $1\leq i \leq k$.
	By assumption $\sum_{1\leq i \leq k} p_i =1$.
	When entering $\cE$, the scheduler works in $k$ memory modes $1$, \dots, $k$ until an action $\alpha$ that does not belong to $\cE$ is scheduled
	starting  in memory mode $1$.
	In each memory mode $i$,
	$\tsched$ follows an MD-scheduler for $\cE$ that reaches $s_i$ with probability $1$ from all states of $\cE$.
	Once, $s_i$ is reached, $\tsched$ chooses action $\alpha_i$ with probability 
	\[
	q_i \eqdef \frac{p_i}{1-\sum_{j<i} p_j}.
	\]
	Now $\tsched$ leaves $\cE$ via $(s_k,\alpha_k)$ with probability $1$ if it reaches the last memory mode $k$.
	As $\tsched$ behaves MD in each mode, it leaves the end component $\cE$ after finitely many steps in expectation.
	Furthermore, for each $i\leq k$, it leaves $\cE$ via $(s_i,\alpha_i)$ with probability $(1-\sum_{j<i} p_j)\cdot q_i = p_i$.
	Since the behavior of $\sched$ in $\MEC{\cM}$ is mimicked by $\tsched$ in $\cM^{\prime}$, we conclude that the expected frequency of actions in $\cM$ which do not belong to an end component is the same in $\cM^\prime$ under $\tsched$ and in $\MEC{\cM}$ under $\sched$.
	
	The expected frequency of each state-action pair of $\cM^\prime$ under $\tsched$ is finite, since each MEC of $\cM^\prime$ is left after finitely many steps in expectation.
	In the terminology of \cite{Kallenberg20}, the scheduler $\tsched$ is \emph{transient}.
	By \cite[Theorem 9.16]{Kallenberg20}, this implies that there is a MR-scheduler $\usched$ for $\cM^\prime$ under which the expected frequency of state-action pairs is the same as under $\tsched$ and thus the expected frequency in $\cM^\prime$ of actions $\alpha$ of $\cM$ that do not belong to an end component is the same as under $\sched$ in $\MEC{\cM}$.
	
	Finally, we modify $\usched$ such that it becomes a scheduler for $\cM$: For each end component $\cE$ of $\cM$ with $\sched(s_{\cE})(\tau)=1$, we fix a memoryless scheduler $\usched_{\cE}$ that does not leave the end component. Now, whenever a state $s$ in such an end component is visited, the modified scheduler switches to the behavior of $\usched_{\cE}$ instead of choosing action $\tau$ with probability $1$. Clearly, this does not affect the expected frequency of actions of $\cM$ that do not belong to an end component and hence the modified scheduler is as claimed in the theorem. 
\end{proof}

\begin{rem}
	The proof of Theorem \ref{thm:MRscheduler_lift} above provides an algorithm  how to obtain the scheduler $\tsched$ from $\sched$. The number of memory modes of the intermediately constructed finite-memory scheduler is bounded by the number of state-action pairs of $\cM$. Further, in each memory mode during the traversal of a MEC, the scheduler behaves in a memoryless deterministic way. Hence, the induced Markov chain is of size polynomial in the size of the MDP $\cM$ and the representation of the scheduler $\sched$. Therefore, also the expected frequencies of all state-action pairs under the intermediate finite-memory scheduler and hence under $\tsched$ can be computed in time polynomial in the size of the MDP $\cM$ and the representation of the scheduler $\sched$. 
	So, also the scheduler $\tsched$ itself which can be derived from these expected frequencies can be computed in polynomial time from $\sched$.
	
	Together with Lemma \ref{lem:MR-sufficient-global-cause}, this means that $\tsched$ and hence the scheduler with two memory modes whose existence is stated in Theorem \ref{thm:MR-sufficient-GPR} can be computed from a solution to the constraint system (S1)-(S5) from Section~\ref{sec:check-GPR} in time polynomial in the size of the original MDP and the size of the representation of the solution to (S1)-(S5).
	\Ende
\end{rem}

With these results we can now prove the second main result of this section, Theorem \ref{thm:MR-sufficient-GPR}, stating that if \eqref{GPR} does not hold there is a finite-memory scheduler with two memory cells refuting the GPR condition.

\begin{proof}[Proof of Theorem \ref{thm:MR-sufficient-GPR}]
	The model transformation establishing assumptions (A1)-(A3) results in the MEC-quotient of $\wminMDP{\cM}{\Cause}$ up to the renaming and collapsing of terminal states. By Corollary \ref{cor:MR-scheduler} and Theorem  \ref{thm:MRscheduler_lift}, we conclude that
	$\Cause$ is not a GPR cause for $\Effect$ in $\cM$ if and only if there is a MR-scheduler $\sched$ for $\wminMDP{\cM}{\Cause}$ with $\Pr^{\sched}_{\wminMDP{\cM}{\Cause}}(\Diamond \Cause)>0$ that violates \eqref{GPR}.
	As in Remark \ref{MR-sufficient-SRP}, 
	$\sched$ can be extended to
	a finite-memory randomized scheduler $\tsched$ for $\cM$ with two memory cells.
\end{proof}


\begin{rem}[On lower bounds on GPR checking]
	Solving systems of quadratic inequalities with linear side constraints is NP-hard in general (see, e.g., \cite{GareyJ79}).
	For convex problems, in which the associated symmetric matrix occurring in the quadratic inequality has only non-negative eigenvalues, the problem is, however, solvable in polynomial time \cite{kozlov1980polynomial}.
	Unfortunately, the quadratic constraint system describing a scheduler refuting \eqref{GPR} given by (S1)-(S5) is not of this form.
	We observe that even if $\Cause$ is a singleton $\{c\}$ and the variable $x_{\effuncov}$ is forced to take a constant value $y$ by (S1)-(S3), i.e., by the structure of the MDP, the inequality (S4) takes the form:
	\begin{align}
		\label{C4 one variable}
		x_c\cdot w_c - x_c^2 \cdot (w_c+y) \leq 0
	\end{align}
	Here, the $1\times 1$-matrix $({-}w_c{-}y)$ has a negative eigenvalue.
	Although it is not ruled out that
	(S1)-(S5) belongs to another class of
	efficiently solvable constraint systems,
	the NP-hardness result in \cite{pardalos1991quadratic}
	for the solvability of
	quadratic inequalities of the form \eqref{C4 one variable}
	with linear side constraints 
	might be an indication for the computational difficulty.
	\Ende
\end{rem}


\section{Quality and optimality of causes}

\label{sec:criteria}

The goal of this section is to identify notions that measure how ``good'' causes are and to present algorithms to determine good causes according to the proposed quality measures.
We have seen so far that small (singleton) causes are easy to determine (see Section~\ref{sec:check-SPR-condition}).
Moreover, it is easy to see that the proposed existence-checking algorithm can be formulated in such a way that the algorithm returns a singleton (strict or global) probability-raising cause $\{c_0\}$ with maximal \emph{precision}, i.e., a state $c_0$ where 
$\inf_{\sched} \Pr^{\sched}_{\cM}(\Diamond \Effect |\Diamond c_0)
= \Pr^{\min}_{\cM,c_0}(\Diamond \Effect)$ is maximal.
On the other hand, 
singleton or small cause sets might have poor coverage in the sense that the probability for paths that reach an effect state without visiting a cause state before (``uncovered effects'') can be large. 
This motivates the consideration of quality notions for causes that incorporate how well 
effect scenarios are covered. 
We take inspiration of quality measures that are considered in statistical analysis (see e.g. \cite{Powers-fscore}). 
This includes the \emph{recall} as a measure for the relative coverage (proportion of covered effects among all effect scenarios), the \emph{coverage ratio} (quotient of covered and uncovered effects) as well as the \emph{f-score}. 
The f-score is a standard measure for classifiers  defined by the harmonic mean of precision and recall.
It can be seen as a compromise to achieve both good precision and good recall.

In this section, we assume as before an MDP $\cM = (S,\Act,P,\init)$ and $\Effect \subseteq S$ are given where all effect states are terminal. 
Furthermore, we suppose all states $s\in S$ are reachable from $\init$.


\subsection{Quality measures for causes}

\label{sec:acc-measures}

In statistical analysis, the precision of a classifier with binary outcomes (``positive'' or ``negative'') is defined as the ratio of all true positives among all positively classified elements, while its recall is defined as the ratio of all true positives among all actual positive elements.
Translated to our setting, we consider classifiers induced by a given cause set $\Cause$ that return ``positive'' for sample paths in case that a cause state is visited 
and ``negative'' otherwise. 
The intuitive meaning of true positives, false positives, true negatives and false negatives is as described in the confusion matrix in Figure \ref{fig: confusion matrix}.
The formal definition is
\[
\begin{array}{rclrcl}
	\mathsf{tp}^\sched & = & \Pr_\cM^\sched(\lozenge \Cause \wedge \lozenge \Eff), & \mathsf{tn}^\sched & = & \Pr_\cM^\sched(\neg \lozenge \Cause \wedge \neg \lozenge \Eff),
	\\[1ex]
	\mathsf{fp}^\sched & = & \Pr_\cM^\sched(\lozenge \Cause \wedge \neg \lozenge \Eff), &\mathsf{fn}^\sched & = & \Pr_\cM^\sched(\neg \lozenge \Cause \wedge \lozenge \Eff).
\end{array}
\]

\begin{figure}[t]
	\begin{tabular}{c|c|c|}
		Path hits& $\Eff$ & $\neg \Eff$ \\
		\hline
		$\Cause$& True positive (tp) & False positive (fp) \\
		&$\Cause$ correctly predicted $\Eff$ & $\Cause$ falsely predicted $\Eff$ \\
		\hline
		$\neg \Cause$ & False negative (fn) & True negative (tn) \\
		&$\Cause$ falsely not predicted $\Eff$ & $\Cause$ correctly not predicted $\Eff$ \\
		\hline
	\end{tabular}
	\caption{Confusion matrix for $\Cause$ as a binary classifier for $\Eff$}\label{fig: confusion matrix}
\end{figure}

With this interpretation of causes as binary classifiers in mind, the recall and precision
and coverage ratio 
of a cause set $\Cause$ \emph{under a scheduler} $\sched$ are defined as follows:
\[
\begin{array}{rclcl}
	\precision^{\sched}(\Cause) & \ = \ &
	\Pr^{\sched}_{\cM}(\ \Diamond \Effect \ | \ \Diamond \Cause \ )
	& = &
	\frac{\displaystyle \mathsf{tp}^{\sched}}
	{\displaystyle \mathsf{tp}^{\sched} + \mathsf{fp}^{\sched}}
	\\[2ex] 
	\recall^{\sched}(\Cause) & = &
	\Pr^{\sched}_{\cM}(\ \Diamond \Cause  \ | \ \Diamond \Effect \ )
	& = &
	\frac{\displaystyle \mathsf{tp}^{\sched}}{\displaystyle \mathsf{tp}^{\sched}+\mathsf{fn}^{\sched}}
	\\[2ex]
	\ratiocov^{\sched}(\Cause) & \ = \ &
	\frac{\displaystyle
		\Pr^{\sched}_{\cM}
		\bigl(\Diamond (\Cause \wedge \Diamond \Effect) \bigr)}
	{\displaystyle
		\Pr^{\sched}_{\cM}\bigl((\neg \lozenge \Cause) \wedge \lozenge \Effect \bigr)}
	& = &
	\frac{\displaystyle \mathsf{tp}^{\sched}}
	{\displaystyle \mathsf{fn}^{\sched}}.
\end{array}
\]

Note that for these quality measures we make some respective assumptions on the scheduler as we assume
\begin{itemize}
	\item $\Pr^{\sched}_{\cM}(\Diamond \Cause)>0$ for $\precision^\sched(\Cause)$, 
	\item $\Pr^{\sched}_{\cM}(\Diamond \Effect)>0$ for $\recall^\sched(\Cause)$ and
	\item $\Pr^{\sched}_{\cM}\bigl( (\neg \lozenge \Cause) \wedge \lozenge \Effect \bigr)>0$) for $\covrat^\sched(\Cause)$.
\end{itemize}
If we have $\Pr^{\sched}_{\cM}\bigl( (\neg \lozenge \Cause) \wedge \lozenge \Effect \bigr)=0$ and $\Pr^\sched_{\cM}(\Diamond \Cause)>0$ for some scheduler $\sched$, we define $\covrat^\sched(\Cause)=+\infty$.
This makes sense since we can converge to such a scheduler $\sched$ with a sequence of schedulers $\sched_0 \ldots$ for which $\Pr^{\sched_i}_{\cM}\bigl( (\neg \lozenge \Cause) \wedge \lozenge \Effect \bigr)>0$ and $\Pr^{\sched_i}_{\cM}(\Diamond \Cause)>0$ for $i \in \mathbb{N}$.
The coverage ratio of such a sequence converges to $+ \infty$.

Finally, the f-score of $\Cause$ \emph{under a scheduler} $\sched$
is defined as the harmonic mean of the precision and recall.
Here we assume $\Pr^{\sched}_{\cM}(\Diamond \Cause)>0$, which implies
$\Pr^{\sched}_{\cM}(\Diamond \Effect)>0$:
\begin{equation*}
	\label{f-score}
	\fscore^{\sched}(\Cause) \ \ \eqdef \ \
	2 \cdot
	\frac{\precision^{\sched}(\Cause)\cdot \recall^{\sched}(\Cause)}
	{\precision^{\sched}(\Cause) + \recall^{\sched}(\Cause)}
	= \frac{2 \cdot \mathsf{tp}^\sched}{2 \cdot \mathsf{tp}^\sched + \mathsf{fp}^\sched + \mathsf{fn}^\sched}
\end{equation*}
If, however, $\Pr_{\cM}^\sched(\Diamond \Eff)>0$ and $\Pr_{\cM}^{\sched}(\Diamond \Cause)=0$ for some $\sched$, define $\fscore^\sched(\Cause)=0$.
This again makes sense as for a sequence of schedulers converging to $\sched$ the f-score also converges to $0$ (also see Lemma \ref{lem:fscore=0}).

To lift the definitions of the quality measures under a scheduler to the quality measure of a cause, we consider the worst-case scheduler:
\begin{defi}[Quality measures for causes]
	Let $\Cause$ be a PR cause.
	We define
	\begin{align*}
		\relcov(\Cause)  \ = \ 
		\inf_{\sched} \relcov^{\sched}(\Cause)
		\ = \
		\Pr^{\min}_{\cM}(\ \Diamond \Cause \ | \ \Diamond \Effect \ )
	\end{align*}
	when ranging over all schedulers $\sched$
	with $\Pr^{\sched}_{\cM}(\Diamond \Effect)>0$.
	Likewise, the coverage ratio and f-score of $\Cause$
	are defined by the worst-case coverage ratio resp. f-score
	-- ranging over schedulers for which $\covratio^{\sched}(\Cause)$ 
	resp. $\fscore^{\sched}(\Cause)$ is defined:
	\begin{align*}
		\covratio(\Cause)  \ = \ 
		\inf_{\sched} \covratio^{\sched}(\Cause),
		\quad
		\fscore(\Cause)  \ = \ 
		\inf_{\sched}  \fscore^{\sched}(\Cause).
	\end{align*}
\end{defi}

Besides the quality measures defined so far, which we will address in detail, there is a vast landscape of further quality measures for binary classifiers in the literature (for an overview, see, e.g., \cite{Powers-fscore}).
One prominent example which has been claimed to be superior to the f-score recently \cite{chicco2020advantages} is \emph{Matthews correlation coefficient} (MCC). In terms of the entries of a confusion matrix (as in Figure~\ref{fig: confusion matrix}), it is defined as
\begin{align*}
	\text{MCC}=\frac{tp\cdot tn - fp\cdot fn}{\sqrt{(tp+fp)\cdot(tp+fn)\cdot(tn+fp)\cdot(tn+fn)}}
\end{align*}
In contrast to the f-score (as well as recall and coverage ratio), it makes use of all four entries of the confusion matrix.
In our setting, we could assign the MCC to a $\Cause$ by again taking the infimum of the value over all sensible schedulers. 

Like the MCC, almost all (cf. \cite{Powers-fscore}) of the quality measures studied in the literature are algebraic functions (intuitively speaking, built from polynomials, fractions and root functions) in the entries of the confusion matrix. 
At the end of this section, we will comment on the computational properties of finding good causes when quality is measured by the infimum over all sensible schedulers of an algebraic function in the entries of the confusion matrix.

\subsection{Computation schemes for the quality measures for fixed cause set}

\label{sec:comp-acc-measures-fixed-cause}

For this section, we assume  a fixed PR cause $\Cause$ is given and address the problem to compute its quality values.
The first observation is, that all quality measures are preserved by the switch from $\cM$ to $\wminMDP{\cM}{\Cause}$ as well as the transformations of $\wminMDP{\cM}{\Cause}$ to an MDP that satisfies conditions (A1)-(A3) of Section~\ref{sec:check-GPR}.
In the following Lemmata \ref{lemma:accuracy-measures-M-and-Mcause} and \ref{lemma:accuracy-measures-M-and-Mcause2} we show that the quality measures $\recall, \covrat$ and $\fscore$ of a fixed $\Cause$ are compatible with the model transformations from Section 4.
These are, on one hand a transformation to $\wminMDP{\cM}{\Cause}$, which only considers the minimal probability to reach $\Eff$ starting from $\Cause$, and on the other hand a transformation to an MDP $\cN$ satisfying (A1)-(A3), which has no end components and has exactly four terminal states $\effcov, \effunc, \noeff_{\mathsf{fp}}, \noeff_{\mathsf{tn}}$.
These four terminal states exactly correspond to the four entries of the confusion matrix (Figure \ref{fig: confusion matrix}).

\begin{lem}
	\label{lemma:accuracy-measures-M-and-Mcause} 
	If $\Cause$ is an SPR or a GPR cause then:
	\begin{center}
		\begin{tabular}{rcl}
			$\recall_{\cM}(\Cause)$ & \ = \ & $\recall_{\wminMDP{\cM}{\Cause}}(\Cause)$
			\\
			$\covratio_{\cM}(\Cause)$ & = & $\covratio_{\wminMDP{\cM}{\Cause}}(\Cause)$
			\\
			$\fscore_{\cM}(\Cause)$ & = & $\fscore_{\wminMDP{\cM}{\Cause}}(\Cause)$
		\end{tabular}  
	\end{center}  
\end{lem}

\begin{proof}
	``$\leqslant$'':  
	A scheduler for $\wminMDP{\cM}{\Cause}$ can be seen as
	a scheduler $\sched$ for $\cM$ behaving as an MD-scheduler minimizing the reachability probability of $\Eff$ from every state in $\Cause$ and we have:
	\begin{center}
		\begin{tabular}{rcl}
			$\recall_{\cM}^{\sched}(\Cause)$ & \ = \ &
			$\recall^{\sched}_{\wminMDP{\cM}{\Cause}}(\Cause)$ \\
			$\covratio_{\cM}^{\sched}(\Cause)$ & = &
			$\covratio^{\sched}_{\wminMDP{\cM}{\Cause}}(\Cause)$ \\
			$\precision_{\cM}^{\sched}(\Cause)$ & = &
			$\precision^{\sched}_{\wminMDP{\cM}{\Cause}}(\Cause)$ \\
		\end{tabular}
	\end{center}
	and therefore:
	\begin{center}
		\begin{tabular}{rcl}
			$\fscore_{\cM}^{\sched}(\Cause)$ & \ = \ &
			$\fscore^{\sched}_{\wminMDP{\cM}{\Cause}}(\Cause)$
		\end{tabular}  
	\end{center}
	We obtain the claimed inequalities, e.g.
	$\recall_{\cM}(\Cause)\leqslant \recall_{\wminMDP{\cM}{\Cause}}(\Cause)$.
	
	``$\geqslant$'':  
	Let $\sched$ be a scheduler of
	$\cM$. Let $\tsched$ be
	the scheduler of $\cM$ that behaves as $\sched$ until the first
	visit to a state in $\Cause$. As soon as $\tsched$ has reached $\Cause$, it
	behaves as an MD-scheduler minimizing the probability to reach $\Effect$.
	Recall and coverage
	under $\tsched$ and $\sched$ have the form:
	\[
	\begin{array}{rclcrcl}
		\recall_{\cM}^{\sched}(\Cause)
		& \ = \ & \frac{x}{x + q}
		& \ \ \ \ &
		\covratio_{\cM}^{\sched}(\Cause) & \ = \ & \frac{x}{q}
		\\
		\recall_{\cM}^{\tsched}(\Cause)
		& = & \frac{y}{y + q}
		& & 
		\covratio_{\cM}^{\tsched}(\Cause) & = & \frac{y}{q}
	\end{array}  
	\]
	where  $x \geqslant y$ (and $q=\mathsf{fn}^{\sched}$).
	Considering $\tsched$ as a scheduler of $\cM$ and of
	$\wminMDP{\cM}{\Cause}$, we get:
	\[
	\begin{array}{rcccl}
		\recall_{\cM}^{\sched}(\Cause)
		& \geqslant &
		\recall_{\cM}^{\tsched}(\Cause)
		& = &
		\recall_{\wminMDP{\cM}{\Cause}}^{\tsched}(\Cause)
		\\
		\covratio_{\cM}^{\sched}(\Cause)
		& \geqslant &
		\covratio_{\cM}^{\tsched}(\Cause)
		& = &
		\covratio_{\wminMDP{\cM}{\Cause}}^{\tsched}(\Cause)
	\end{array}
	\]
	This implies:
	\begin{align*}
		\recall_{\cM}^{\sched}(\Cause)  \ \geqslant \ 
		\recall_{\wminMDP{\cM}{\Cause}}(\Cause),
		&&\text{and}&&
		\covratio_{\cM}(\Cause) \ \geqslant \
		\covratio_{\wminMDP{\cM}{\Cause}}(\Cause)
	\end{align*}
	With similar arguments we get:
	\[
	\begin{array}{rcccl}
		\precision_{\cM}^{\sched}(\Cause)
		& \geqslant & 
		\precision_{\cM}^{\tsched}(\Cause)
		& = &
		\precision_{\wminMDP{\cM}{\Cause}}^{\tsched}(\Cause)
	\end{array}
	\]
	As the harmonic mean viewed as a function
	$f: \Real_{>0}^2 \to \Real$, $f(x,y) = 2 \frac{xy}{x{+}y}$
	is monotonically increasing in both arguments
	(note that $\frac{df}{dx} = \frac{y^2}{(x{+}y)^2}>0$ and
	$\frac{df}{dy} = \frac{x^2}{(x{+}y)^2}>0$),
	we obtain:
	\[
	\fscore_{\cM}^{\sched}(\Cause)
	\ \geqslant \ 
	\fscore_{\cM}^{\tsched}(\Cause)
	\ = \
	\fscore_{\wminMDP{\cM}{\Cause}}^{\tsched}(\Cause)
	\]
	This yields
	$\fscore_{\cM}(\Cause) \geqslant \fscore_{\wminMDP{\cM}{\Cause}}(\Cause)$.
\end{proof}

\begin{lem}
	\label{lemma:accuracy-measures-M-and-Mcause2} 
	Let $\cN$ be the MEC-quotient of $\wminMDP{\cM}{\Cause}$ for some MDP $\cM$ with a set of terminal states $\Effect$ and 
	an SPR or a GPR cause $\Cause$. Then:
	\begin{center}
		\begin{tabular}{rcl}
			$\recall_{\wminMDP{\cM}{\Cause}}(\Cause)$ &  =  & $\recall_{\cN}(\Cause)$
			\\
			$\covratio_{\wminMDP{\cM}{\Cause}}(\Cause)$ & = & $\covratio_{\cN}(\Cause)$
			\\
			$\fscore_{\wminMDP{\cM}{\Cause}}(\Cause)$ & = & $\fscore_{\cN}(\Cause)$
		\end{tabular}  
	\end{center}  
\end{lem}

\begin{proof}
	Analogously to the proof of Lemma \ref{lem:probabilities_MEC-quotient}. 
\end{proof}

This now allows us to work under assumptions (A1)-(A3) when addressing problems concerning the quality measures for a fixed cause set.


As efficient computation methods for $\recall(\Cause)$
are known from literature (see \cite{TACAS14-condprob,Maercker-PhD20}
for poly-time algorithms to compute
conditional reachability probabilities), we can use the same methods to compute the coverage ratio.

\begin{cor}
	\label{cor:covratio-in-PTIME}
	The value $\covratio(\Cause)$ and corresponding worst-case schedulers are computable in polynomial time.
\end{cor}

\begin{proof}
	For a given scheduler $\sched$ we have 
	\begin{align*}
		\covrat^\sched(\Cause)=\frac{\mathsf{tp}^\sched} {\mathsf{fn}^\sched} \ \ \text{and} \ \
		\recall^\sched(\Cause)=\frac{\mathsf{tp}^\sched}{\mathsf{tp}^\sched + \mathsf{fn}^\sched}
	\end{align*}
	We thus get the following
	\begin{align*}
		\frac{1}{\frac{\mathsf{tp}^\sched}{\mathsf{tp}^\sched + \mathsf{fn}^\sched}} &=\frac{\mathsf{tp}^\sched + \mathsf{fn}^\sched}{\mathsf{tp}^\sched} = \frac{\mathsf{tp}^\sched}{\mathsf{tp}^\sched} + \frac{\mathsf{fn}^\sched}{\mathsf{tp}^\sched} = 1 + \frac{\mathsf{fn}^\sched}{\mathsf{tp}^\sched} = 1+\frac{1}{\frac{\mathsf{tp}^\sched}{\mathsf{fn}^\sched}}.
	\end{align*}
	This implies
	\begin{align*}
		\frac{1}{\recall^\sched(\Cause)} = 1+ \frac{1}{\covrat^\sched(\Cause)} \ \text{thus} \
		\covrat^\sched(\Cause) =1 / \left(\frac{1}{\recall^\sched(\Cause)} - 1\right).
	\end{align*}
	Computing $\covrat(\Cause)$ now implores us to take the infimum of all sensible schedulers over $1 / \left(\frac{1}{\recall^\sched(\Cause)} - 1\right)$ which is the same as taking the infimum of all sensible schedulers over $\recall^\sched(\Cause)$.
	This amounts to computing
	\begin{align*}
		\inf_{\sched}\recall^\sched(\Cause) = \Pr_{\cM}^{\min}( \ \lozenge\Cause \mid \lozenge \Eff \ ),
	\end{align*}
	which can be computed in polynomial time by \cite{TACAS14-condprob, Maercker-PhD20}.
\end{proof}

In contrast to these results, we are not aware of
known concepts that are applicable 
for computing the f-score.
Indeed, this quality measure is efficiently computable:

\begin{thm}
	\label{fscore-in-PTIME}
	The value $\fscore(\Cause)$ and corresponding worst-case schedulers are computable in polynomial time.
\end{thm}


The remainder of this subsection is devoted to the proof of Theorem \ref{fscore-in-PTIME}.
We can express $\fscore(\Cause)$ in terms of the supremum of a quotient of reachability probabilities for disjoint sets of terminal states.
More precisely, under assumptions (A1)-(A3) and assuming $\fscore(\Cause)>0$,
we have:
\begin{center}
	$
	\fscore(\Cause) = \frac{2}{X+2}
	\quad \text{where} \quad
	X \ = \ %
	\sup_\sched \, \frac{\Pr^{\sched}_{\cM}(\Diamond \noefffp)
		+ \Pr^{\sched}_{\cM}(\Diamond \effuncov)}
	{\Pr^{\sched}_{\cM}(\Diamond \effcov)}
	$
\end{center}
where $\sched$ ranges over all schedulers with
$\Pr_{\cM}^\sched (\Diamond \effcov)>0$.
Moreover, we can show that we can handle the corner case of $\fscore(\Cause) = 0$.
\begin{lem}\label{lem:fscore=0}
	Let $\Cause$ be an SPR or a GPR cause. Then, the following three statements
	are equivalent:
	\begin{enumerate}
		\item [(a)] $\recall(\Cause)=0$
		\item [(b)] $\fscore(\Cause)=0$
		\item [(c)] There is a scheduler $\sched$ such that
		$\Pr^{\sched}_{\cM}(\Diamond \Effect)>0$ and
		$\Pr^{\sched}_{\cM}(\Diamond \Cause)=0$.
	\end{enumerate}
\end{lem}

\begin{proof}
	Let $C=\Cause$.
	Using results of \cite{TACAS14-condprob,Maercker-PhD20}, there exist schedulers $\tsched$ and $\usched$
	with
	\begin{itemize}
		\item
		$\Pr^{\tsched}_{\cM}(\Diamond \Effect)>0$ 
		and 
		$\Pr^{\tsched}_{\cM}(\ \Diamond C \ | \Diamond \Effect \ )
		= 
		\inf_{\sched} \Pr^{\sched}_{\cM}(\ \Diamond C \ | \Diamond \Effect \ )$
		where $\sched$ ranges over all schedulers with positive
		effect probability,
		\item
		$\Pr^{\usched}_{\cM}(\Diamond C)>0$ 
		and 
		$\Pr^{\usched}_{\cM}(\ \Diamond \Effect \ | \Diamond C \ )
		= 
		\inf_{\sched} \Pr^{\sched}_{\cM}(\ \Diamond \Effect \ | \Diamond C \ )$
		where $\sched$ ranges over all schedulers with
		$\Pr^{\sched}_{\cM}(\Diamond C)>0$.
	\end{itemize}
	In particular,
	$\recall(C)=\Pr^{\tsched}_{\cM}(\ \Diamond C \ | \Diamond \Effect \ )$
	and
	$\precision(C)=\Pr^{\usched}_{\cM}(\ \Diamond \Effect \ | \Diamond C \ )$.
	By \eqref{GPR} applied to $\usched$ and $\tsched$
	(recall that each SPR cause is a GPR cause too, 
	see Lemma \ref{lemma:strict-implies-global}),
	we obtain the following statements (i) and (ii):
	\begin{description}
		\item [\text{\rm (i)}]
		$p \ \eqdef \ \precision(C) \ > \ 0$
		\item [\text{\rm (ii)}]
		If $\Pr^{\tsched}_{\cM}(\Diamond C)>0$
		then $\Pr^{\tsched}_{\cM}(\Diamond C \wedge \Diamond \Effect)>0$
		and therefore
		$\recall(C) >0$.
	\end{description}
	Obviously,
	if there is no scheduler $\sched$ as in statement (c) then
	$\Pr^{\tsched}_{\cM}(\Diamond C)>0$.
	Thus, from (ii) we get:
	\begin{description}
		\item [\text{\rm (iii)}]
		If there is no scheduler $\sched$ as in statement (c)
		then 
		$\recall(C) >0$.
	\end{description}
	
	``(a) $\Longrightarrow$ (b)'':
	We prove $\fscore(C)>0$ implies $\recall(C)>0$.
	If $\fscore(C)>0$ then, by definition of the f-score, there is no scheduler
	$\sched$ as in statement (c).
	But then $\recall(C)>0$ by statement (iii).
	
	``(b) $\Longrightarrow$ (c)'':
	Let $\fscore(C)=0$.
	Suppose by contradiction that there is no scheduler as in (c).
	Again by (iii) we obtain $\recall(C) >0$.
	But then, for each scheduler $\sched$ with $\Pr^{\sched}_{\cM}(\Diamond C)>0$:
	\[
	\precision^{\sched}(C) \ \geqslant \ p \ \stackrel{\text{\tiny (i)}}{>} \ 0
	\]
	and, with $r \eqdef \recall(C)$:
	\[
	\recall^{\sched}(C) \ \geqslant \ r \ > \ 0
	\]
	The harmonic mean as a function 
	$]0,1]^2\to \Real$, $(x,y) \mapsto 2 \frac{xy}{x+y}$ is
	monotonically increasing in both arguments.
	But then:
	\[
	\fscore^{\sched}(C) \ \geqslant \ 
	2 \frac{p \cdot r}{p{+}r} \ > \ 0
	\]
	Hence, $\fscore(C) = \inf_{\sched} \fscore^{\sched}(C) \geqslant 2 \frac{p \cdot r}{p{+}r}  >  0$.
	Contradiction.
	
	``(c) $\Longrightarrow$ (a)'':
	Let $\sched$ be a scheduler as in statement (c).
	Then, 
	\[
	\Pr^{\sched}_{\cM}(\ \Diamond C \ | \Diamond \Effect \ )  \ = \ 0.
	\]
	Hence:
	$\recall(C) \ = \ 
	\Pr^{\min}_{\cM}(\ \Diamond C \ | \Diamond \Effect \ ) \ = \ 0$.
\end{proof}

The remaining task to prove Theorem \ref{fscore-in-PTIME} is a generally applicable technique for computing extremal ratios of reachability probabilities in MDPs without ECs.

\paragraph*{\bf Max/min ratios of reachability probabilities for disjoint sets of terminal states.}
\label{sec:comp-quotient}

Suppose we are given an MDP $\cM= (S,\Act,P,\init)$ without ECs and disjoint subsets $U,V\subseteq S$  of terminal states. 
Given a scheduler $\sched$ with $\Pr_\cM^\sched(\Diamond V)>0$ we define:
\begin{align*}
	\ratio{\sched}{\cM}{U,V} \ = \
	\frac{\Pr_{\cM}^{\sched}(\Diamond U)}{\Pr_{\cM}^{\sched} (\Diamond V)}
\end{align*}
The goal is an algorithm for computing the extremal values:
\begin{align*}
	\ratio{\min}{\cM}{U,V} = \inf_{\sched} \ratio{\sched}{\cM}{U,V} \quad \text{and} \quad \ratio{\max}{\cM}{U,V} = \sup_{\sched} \ratio{\sched}{\cM}{U,V}
\end{align*}
where $\sched$ ranges over all schedulers with $\Pr_\cM^\sched(\Diamond V)>0$. 

To compute these, we rely on a polynomial reduction to the classical \emph{stochastic shortest path problem} \cite{BT91}. 
For this, consider the MDP $\cN$ arising from $\cM$ by adding reset transitions from all terminal states $t \in S \backslash V$ to $\init$.
Thus, exactly the $V$-states are terminal in $\cN$.
$\cN$ might contain ECs, which, however, do not intersect with $V$.
We equip $\cN$ with the weight function that assigns $1$ to all states in $U$ and $0$ to all other states. 
For a scheduler $\tsched$ with $\Pr^{\tsched}_{\cN}(\Diamond V)=1$, let $\mathrm{E}^{\tsched}_{\cN}(\boxplus V)$ be the expected accumulated weight until reaching $V$ under $\tsched$.
Let $\mathrm{E}^{\min}_{\cN}(\boxplus V) =  \inf_{\tsched} \mathrm{E}^{\tsched}_{\cN}(\boxplus V)$ and  $\mathrm{E}^{\max}_{\cN}(\boxplus V) =  \sup_{\tsched} \mathrm{E}^{\tsched}_{\cN}(\boxplus V)$, 
where $\tsched$ ranges over all schedulers 
with
\mbox{$\Pr^{\tsched}_{\cN}(\Diamond V)=1$.}
We can rely on known results \cite{BT91,Alfaro-CONCUR99,LICS18-SSP} to obtain
that both
$\mathrm{E}^{\min}_{\cN}(\boxplus V)$ and $\mathrm{E}^{\max}_{\cN}(\boxplus V)$
are computable in polynomial time.
As $\cN$ has only non-negative weights,
$\mathrm{E}^{\min}_{\cN}(\boxplus V)$ is finite
and a corresponding MD-scheduler with minimal expectation exists.
If $\cN$ has an EC containing at least one $U$-state,
which is the case iff $\cM$ has a scheduler $\sched$
with $\Pr^{\sched}_{\cM}(\Diamond U)>0$ and $\Pr^{\sched}_{\cM}(\Diamond V)=0$,
then
$\mathrm{E}^{\max}_{\cN}(\boxplus V) = +\infty$.
Otherwise, $\mathrm{E}^{\max}_{\cN}(\boxplus V)$
is finite and the maximum is achieved by an MD-scheduler as well.

\begin{thm}\label{thm:comp-Q}
	Let $\cM$ be an MDP without ECs and
	$U,V$ disjoint sets of terminal states in $\cM$, and let $\cN$ be as before.
	Then,
	$\ratio{\min}{\cM}{U,V}=\mathrm{E}^{\min}_{\cN}(\boxplus V)$ and
	$\ratio{\max}{\cM}{U,V}=\mathrm{E}^{\max}_{\cN}(\boxplus V)$.
	Thus, both values are computable in polynomial time,
	and there is an MD-scheduler minimizing $\ratio{\sched}{\cM}{U,V}$,
	and an MD-scheduler maximizing $\ratio{\sched}{\cM}{U,V}$ if 
	$\ratio{\max}{\cM}{U,V}$ is finite.
\end{thm}

\begin{figure}[t]
	\centering
	\begin{minipage}{0.41\textwidth}
		\centering
		\resizebox{\textwidth}{!}{
			\begin{tikzpicture}[scale=1,->,>=stealth',auto ,node distance=0.5cm, thick]
	\tikzstyle{round}=[thin,draw=black,circle]
	
	\node[scale=1, state] (init) {$\init$};
	\node[scale=1, state, below=1.25 of init] (pre) {$\effuncov$};
	\node[scale=1, state, right=3.75 of init] (c) {$c$};
	\node[scale=1, state, right=1.25 of pre] (t) {$\noeff$};
	\node[scale=1, state, below=1.25 of c] (post) {$\effcov$};
	
	\draw[<-] (init) --++(-0.55,0.55);
	\draw[color=black,->] (init) edge  node [pos=0.2, anchor=center] (n5) {} node [pos=0.5,above] {$1/2$} (c) ;
	\draw[color=black,->] (init)  edge  node [pos=0.5, anchor=center] (n0) {} node [pos=0.5, left] {$1/2$} (pre) ;
	\draw[color=black,->] (c) edge  node [near start, anchor=center] (m5) {} node [pos=0.5,right] {$1/2$} (post) ;
	\draw[color=black,->] (init) edge[out=325, in=125] node [near start, anchor=center] (n6) {} node [pos=0.5,right] {$1/2$} (t) ;
	\draw[color=black,->] (init) edge[out=305, in=145] node [near start, anchor=center] (a6) {} node [pos=0.7,left] {$1/2$} (t) ;
	\draw[color=black, very thick, -] (n6.center) edge [bend right=45] node [pos=0.3] {$\alpha$} (n5.center);
	\draw[color=black, very thick, -] (n0.center) edge [bend right=45] node [pos=0.3,above] {$\beta$} (a6.center);
	
	\draw[color=black,->] (c)  edge  node [near start, anchor=center] (m6) {} node [pos=0.5,left] {$1/2$} (t) ;
	
\end{tikzpicture}
		}
		\caption{MDP $\cM$ from Example \ref{ex:quotient}}
		\label{fig:quotient1}
	\end{minipage}\hspace{25pt}
	\begin{minipage}{0.41\textwidth}
		\centering
		\resizebox{\textwidth}{!}{
			\begin{tikzpicture}[scale=1,->,>=stealth',auto ,node distance=0.5cm, thick]
	\tikzstyle{round}=[thin,draw=black,circle]
	
	\node[scale=1, state] (init) {$\init$};
	\node[scale=1, state, below=1.25 of init] (pre) {$\effuncov$};
	\node[scale=1, state, right=3.75 of init] (c) {$c$};
	\node[scale=1, state, right=1.25 of pre] (t) {$\noeff$};
	\node[scale=1, state, below=1.25 of c] (post) {$\effcov$};
	
	\draw[<-] (init) --++(-0.55,0.55);
	\draw[color=black,->] (init) edge  node [pos=0.2, anchor=center] (n5) {} node [pos=0.5,above] {$1/2$} (c) ;
	\draw[color=black,->] (init)  edge  node [pos=0.5, anchor=center] (n0) {} node [pos=0.5, left] {$1/2$} (pre) ;
	\draw[color=black,->] (c) edge  node [near start, anchor=center] (m5) {} node [pos=0.5,right] {$1/2$} (post) ;
	\draw[color=black,->] (init) edge[out=325, in=125] node [near start, anchor=center] (n6) {} node [pos=0.5,right] {$1/2$} (t) ;
	\draw[color=black,->] (init) edge[out=305, in=145] node [near start, anchor=center] (a6) {} node [pos=0.7,left] {$1/2$} (t) ;
	\draw[color=black, very thick, -] (n6.center) edge [bend right=45] node [pos=0.3] {$\alpha$} (n5.center);
	\draw[color=black, very thick, -] (n0.center) edge [bend right=45] node [pos=0.3,above] {$\beta$} (a6.center);
	
	\draw[color=black,->] (c)  edge  node [near start, anchor=center] (m6) {} node [pos=0.5,left] {$1/2$} (t) ;
	\node[scale=0.02, above = 0.15 of c] (help) {};
	\draw[color=red,->] (post) to[in=345, out=45] (help) to[out=165, in=30] (init);
	
	\node[scale=0.02, left=0.15 of pre] (help1) {};
	\draw[color=red,->] (t) to[in=270, out=225] (help1) to[out=90, in=225] (init);
	
\end{tikzpicture}
		}
		\caption{MDP $\cN$ with reset transitions for $\ratio{\min}{\cM}{\effcov, \effunc}$}
		\label{fig:quotient2}
	\end{minipage}
\end{figure}

\begin{exa}
	\label{ex:quotient}
	Consider the MDP $\cM$ from Figure \ref{fig:quotient1} with $\Eff =\{\effunc, \effcov\}$ and suppose the task is to compute $\covratio(c) = \ratio{\min}{\cM}{\effcov, \effunc}$.
	The construction for the algorithm is depicted in Figure \ref{fig:quotient2} resulting in $\cN$, where reset transitions for $\noeff$ and $\effcov$ have been added (red edges) and $\effunc$ is the only terminal state.
	The weight function now assigns $1$ to $\effcov$ and $0$ to all others.
	By Theorem \ref{thm:comp-Q} we have $\ratio{\min}{\cM}{\effcov, \effunc} = \mathrm{E}^{\min}_\cN(\boxplus \effunc).$
	\Ende
\end{exa}

\begin{proof}[Proof of Theorem \ref{thm:comp-Q}]
	$\cM$ has a scheduler $\sched$
	with $\Pr^{\sched}_{\cM}(\Diamond U)>0$ and $\Pr^{\sched}_{\cM}(\Diamond V)=0$
	if and only if
	the transformed MDP $\cN$ in Section~\ref{sec:comp-quotient}
	(Max/min ratios of reachability probabilities for disjoint sets of terminal states)
	has an EC containing at least one $U$-state.
	Therefore, we then have 
	$\mathrm{E}^{\max}_{\cN}(\boxplus V) = +\infty.$
	Otherwise, $\mathrm{E}^{\max}_{\cN}(\boxplus V)$ is finite.
	
	For the following we only consider $\ratio{\min}{\cM}{U,V} = \mathrm{E}^{\min}_\cN(\boxplus V)$ since the arguments for the maximum are similar.
	First, we show
	$\ratio{\min}{\cM}{U,V}  \geqslant \mathrm{E}^{\min}_{\cN}(\boxplus V)$.
	For this, we consider an arbitrary scheduler $\sched$ for $\cM$.
	Let
	\begin{align*}
		x  &= \Pr^{\sched}_{\cM}(\Diamond U)& 
		p & = \Pr^{\sched}_{\cM}(\Diamond V) &
		q & = 1 - x - p 
	\end{align*}
	For $p>0$ we have
	\[\frac{\Pr^{\sched}_{\cM}(\Diamond U)}
	{\Pr^{\sched}_{\cM}(\Diamond V)}    
	\ \ = \ \
	\frac{x}{p}\]
	Let $\tsched$ be the scheduler that behaves as $\sched$
	in the first round and after each reset.
	Then:
	\begin{align}
		\label{big-sum} 
		\mathrm{E}^{\tsched}_{\cN}(\boxplus V) \ \ = \ \
		\sum_{n=0}^{\infty}
		\sum_{k=0}^{\infty}
		n \cdot x^n \cdot
		\left( \!\!\!\begin{array}{c}
			n{+}k \\
			k
		\end{array} \!\!\!\right)
		q^k
		\cdot p
		\ \ \stackrel{\text{\eqref{second basic}}}{=} \ \
		\frac{x}{p}
	\end{align}  
	where \eqref{second basic} relies on some basic calculations
	(see Lemma \ref{second-basic-fact}).
	This yields: 
	\[
	\ratio{\sched}{\cM}{U,V} \ = \ \frac{x}{p} \ = \ 
	\mathrm{E}^{\tsched}_{\cN}(\boxplus V)
	\ \geqslant \ \mathrm{E}^{\min}_{\cN}(\boxplus V)
	\]
	Hence,
	$\ratio{\min}{\cM}{U,V}\geqslant \mathrm{E}^{\min}_{\cN}(\boxplus V)$.

	For the other direction $\mathrm{E}^{\min}_{\cN}(\boxplus V) \geqslant \ratio{\min}{\cM}{U,V}$, we use the fact that there is an MD-scheduler
	$\tsched$ for $\cN$ such that
	$\mathrm{E}^{\tsched}_{\cN}(\boxplus V)
	= \mathrm{E}^{\min}_{\cN}(\boxplus V)$.
	$\tsched$ can be viewed as an MD-scheduler for the original MDP $\cM$.
	Again we can rely on \eqref{big-sum} to obtain that:
	\[
	\mathrm{E}^{\tsched}_{\cN}(\boxplus V) \ \ = \ \
	\frac{\Pr^{\tsched}_{\cM}\bigl(\Diamond U \bigr)}
	{\Pr^{\tsched}_{\cM}\bigl(\Diamond V \bigr)}
	\ \ = \ \ \ratio{\tsched}{\cM}{U,V}
	\ \ \geqslant \ \ \ratio{\min}{\cM}{U,V}
	\]
	But this yields
	$\mathrm{E}^{\min}_{\cN}(\boxplus V) 
	\ \geqslant  \ \ratio{\min}{\cM}{U,V}$.
	For $\mathrm{E}^{\max}_{\cN}(\boxplus V) \ = \ \ratio{\max}{\cM}{U,V}$ we use similar arguments.
	We can now rely on known results \cite{BT91,Alfaro-CONCUR99,LICS18-SSP} to compute $\mathrm{E}^{\min}_{\cN}(\boxplus V)$ and $\mathrm{E}^{\max}_{\cN}(\boxplus V)$ in polynomial time.
\end{proof} 

\begin{lem}
	\label{second-basic-fact}
	Let
	$x,p,q \in [0,1]$ such that $x{+}q{+}p=1$.
	Then:
	\begin{align}
		\label{second basic}
		\sum_{n=0}^{\infty}
		\sum_{k=0}^{\infty}
		n \cdot x^n \cdot
		\left( \!\!\!\begin{array}{c}
			n{+}k \\
			k
		\end{array} \!\!\!\right)
		q^k
		\cdot p
		\ \ = \ \
		\frac{x}{p}
	\end{align}
	
\end{lem}

\begin{proof}
	We first show for $0 < q <1$, $n\in \Nat$ and
	\begin{eqnarray*}
		a_n & \eqdef &
		\sum_{k=0}^{\infty}
		\left(\!\!\!\begin{array}{c}
			n{+}k \\
			k
		\end{array} \!\!\!\right)
		q^k,
	\end{eqnarray*}
	we have
	\[
	a_n = \frac{1}{(1{-}q)^{n+1}}
	\]
	This is done by induction on $n$. The claim is clear for $n{=}0$.
	For the step of induction we use:
	\[
	\left( \!\!\!\begin{array}{c}
		n{+}1{+}k \\
		k
	\end{array} \!\!\!\right)
	\ \ = \ \
	\left( \!\!\!\begin{array}{c}
		n{+}k \\
		k
	\end{array} \!\!\!\right)
	\ + \
	\left( \!\!\!\begin{array}{c}
		n{+}k \\
		k{-}1
	\end{array} \!\!\!\right)
	\ \ = \ \
	\left( \!\!\!\begin{array}{c}
		n{+}k \\
		k
	\end{array} \!\!\!\right)
	\ + \
	\left( \!\!\!\begin{array}{c}
		(n{+}1)+(k{-}1) \\
		k{-}1
	\end{array} \!\!\!\right)    
	\]
	But this yields $a_{n+1}= a_n + q\cdot a_{n+1}$.
	Hence:
	\[
	a_{n+1} \ = \ \frac{a_n}{1{-}q}
	\]
	The claim then follows directly from the induction hypothesis.
	
	The statement of Lemma \ref{second-basic-fact} now follows 
	by some calculations and
	the preliminary induction.
	\begin{eqnarray*}
		\sum_{n=0}^{\infty}
		\sum_{k=0}^{\infty}
		n \cdot x^n \cdot
		\left( \!\!\!\begin{array}{c}
			n{+}k \\
			k
		\end{array} \!\!\!\right)
		q^k
		\cdot p
		& = &
		\sum_{n=0}^{\infty}
		n \cdot x^n \cdot \frac{1}{(1{-}q)^{n+1}} \cdot p
		\\
		\\[0.5ex]
		& = &
		\frac{p}{1{-}q} \cdot
		\sum_{n=0}^{\infty} n \cdot \left( \frac{x}{1{-}q} \right)^n
		\\
		\\[0.5ex]
		& = &
		\frac{p}{1{-}q} \cdot
		\frac{\displaystyle  \frac{x}{1{-}q} }
		{\displaystyle \ \Bigl(1-\frac{x}{1{-}q}\Bigr)^2 \ }
		\\
		\\[0.5ex]
		& = &
		\frac{px}{(1{-}q{-}x)^2}  
		\ \ \ = \ \ \
		\frac{px}{p^2} \ \ \ = \ \ \ \frac{x}{p}  
	\end{eqnarray*}
	where we use $p=1{-}q{-}x$.
\end{proof}  

Applying this framework for $\ratio{\max}{\cM}{U,V}$ to the f-score we now prove Theorem \ref{fscore-in-PTIME}.

\begin{proof}[Proof of Theorem \ref{fscore-in-PTIME}]
	We use the simplifying assumptions (A1)-(A3) that can be made due to Lemmas  \ref{lemma:accuracy-measures-M-and-Mcause} and \ref{lemma:accuracy-measures-M-and-Mcause2}.
	For $\fscore(\Cause)$ we have
	after some straight-forward transformations
	\begin{align*}
		\fscore^{\sched}(\Cause) & =  \frac{2 \mathsf{tp}^\sched}{2\mathsf{tp}^\sched + \mathsf{fn}^\sched + \mathsf{fp}^\sched}.
	\end{align*}
	Using this we get
	\begin{align*}
		\frac{2}{\fscore^{\sched}(\Cause)} -2 & = \frac{\mathsf{fp}^\sched + \mathsf{fn}^\sched}{\mathsf{tp}^\sched} =
		\frac{\Pr^{\sched}_{\cM}(\Diamond \noeff_{\mathsf{fp}})	+ \Pr^{\sched}_{\cM}(\Diamond \effuncov)}
		{\Pr^{\sched}_{\cM}(\Diamond \effcov)}
	\end{align*}
	Thus, the task is to compute
	\[X = \sup_\sched \frac{2}{\fscore^\sched(\Cause)}-2 = \sup_\sched \frac{\Pr^{\sched}_{\cM}(\Diamond \noeff_{\mathsf{fp}})	+ \Pr^{\sched}_{\cM}(\Diamond \effuncov)}
	{\Pr^{\sched}_{\cM}(\Diamond \effcov)},\]
	where $\sched$ ranges over all schedulers with $\Pr_{\cM}^\sched (\Diamond \effcov)>0$.
	We have \[\fscore(\Cause) = \frac{2}{X+2}.\]
	But $X$ can be expressed as a supremum in the form of Theorem \ref{thm:comp-Q}. 
	This yields the claim that the optimal value is computable in polynomial time. 
	
	In case $\fscore(\Cause) = 0$, we do not obtain an optimal scheduler via Theorem \ref{thm:comp-Q}.
	Lemma \ref{lem:fscore=0}, however, shows that there is a scheduler $\sched$ with $\Pr^{\sched}_{\cM}(\Diamond \Effect ) > 0$ and $\Pr^{\sched}_{\cM}(\Diamond \Cause ) = 0$. 
	Such a scheduler can be computed in polynomial time as any (memoryless) scheduler in the largest sub-MDP of $\cM$ that does not contain states in $\Cause$.
	(This sub-MDP can be constructed by successively removing states and state-action pairs.)
\end{proof}


\subsection{Quality-optimal probability-raising causes}	
\label{sec:opt-PR-causes}

For the computation there is no difference between GPR and SPR causes as only the quality properties of the set are in question.
However, when finding optimal causes the distinction makes a difference.
Here, we say an SPR cause $\Cause$ is \emph{recall-optimal} if $\relcov(\Cause) = \max_C \relcov(C)$ where $C$ ranges over all SPR causes.
Likewise, \emph{ratio-optimality} resp. \emph{f-score-optimality} of $\Cause$ means maximality of $\ratiocov(\Cause)$ resp. $\fscore(\Cause)$ among all SPR causes.
Recall-, ratio- and f-score-optimality for GPR causes are defined accordingly.


\begin{lem}
	\label{lemma:recall-opt=ratio-opt}
	Let $\Cause$ be an SPR or a GPR cause.
	Then, $\Cause$ is recall-optimal if and only if $\Cause$ is ratio-optimal.
\end{lem}

\begin{proof}
	Essentially the proof uses the same connection between $\recall$ and $\covrat$ as Corollary \ref{cor:covratio-in-PTIME}.
	Here we do not assume (A1)-(A3). 
	However, for each scheduler $\sched$ and each set $C$ of states we have:
	\[
	\Pr^{\sched}_{\cM}(\Diamond \Effect) \ = \ \mathsf{fn}^{\sched}_C+\mathsf{tp}^{\sched}_C
	\]
	where  
	$\mathsf{fn}^{\sched}_C=\Pr^{\sched}_{\cM}\bigl((\neg \lozenge C) \wedge \lozenge \Effect \bigr)$
	and
	$\mathsf{tp}^{\sched}_C =
	\Pr^{\sched}_{\cM}\bigl(\Diamond (C \wedge \Diamond \Effect) \bigr)$.
	If $C$ is a cause where $\mathsf{fn}^{\sched}_C$ is positive then $
	\ratiocov^{\sched}(C) \ = \ \frac{\mathsf{tp}^{\sched}_C}{\mathsf{fn}^{\sched}_C}$
	and
	$\relcov^{\sched}(C) \ = \ \frac{\mathsf{tp}^{\sched}_C}{\mathsf{fn}^{\sched}_C+\mathsf{tp}^{\sched}_C}$.
	
	For all non-negative reals $p,q,p',q'$ where $q,q'> 0$ we have:
	\[
	\frac{p}{q} < \frac{p'}{q'}
	\qquad \text{iff} \qquad
	\frac{p}{p+q} < \frac{p'}{p'+q'}.
	\]
	Hence, if $C$ is fixed and  $\sched$ ranges over all schedulers with
	$\mathsf{tp}_C^{\sched}>0$:
	\begin{align*}
		\frac{\mathsf{tp}^{\sched}_C}{\mathsf{fn}^{\sched}_C} \ \text{ is minimal iff } \ \frac{\mathsf{tp}^{\sched}_C}{\mathsf{fn}^{\sched}_C+\mathsf{tp}^{\sched}_C} \ \text{ is minimal}
	\end{align*}  
	Thus, if $C$ is fixed and $\sched=\sched_C$ is a scheduler achieving the
	worst-case (i.e., minimal) coverage ratio for $C$ then
	$\sched$ achieves the minimal recall for $C$, and vice versa.
	
	Let now $\mathsf{fn}_C=\mathsf{fn}_C^{\sched_C}$, $\mathsf{tp}_C=\mathsf{tp}_C^{\sched_C}$ where
	$\sched_C$ is a scheduler that minimizes the coverage ratio and
	minimizes the recall for cause set $C$.
	Then:
	\begin{center}
		$\ratiocov(C)= \frac{\mathsf{tp}_C}{\mathsf{fn}_C}$ is maximal \ iff \
		$\frac{\mathsf{tp}_C}{\mathsf{fn}_C+\mathsf{tp}_C}$ is maximal \ iff \
		$\relcov(C)$ is maximal
	\end{center}  
	where the extrema range over all SPR resp. GPR causes $C$.
	This yields the claim.
\end{proof}


\paragraph*{\bf Recall-  and ratio-optimal SPR causes.}
\label{sec:opt-SPR-causes}
The techniques of Section~\ref{sec:check-SPR-condition}
yield an algorithm for generating a canonical SPR
cause with optimal recall and coverage ratio.
To see this, let $\cC$ denote the set of all states which constitute a singleton SPR cause.
The canonical cause $\CanCause$ is defined as the set of states $c\in \cC$ such that there is a scheduler $\sched$ with $\Pr_{\cM}^\sched((\neg \cC) \Until c)>0$.
So to speak $\CanCause$ is the ``front'' of $\cC$.
Obviously, $\cC$ and $\CanCau$ are computable in polynomial time.

\begin{thm}%
	\label{thm:optimality-of-canonical-SPR-cause}
	If $\cC\not= \varnothing$ then $\CanCause$ is a ratio- and recall-optimal SPR cause.
\end{thm}

\begin{proof}
	By definition of SPR causes any subset $C \subseteq \cC$ satisfying $\Pr^{\max}(\neg C \until c)$ for each $c \in C$ constitutes an SPR cause and thus $\CanCause$ is also an SPR cause.
	Optimality is a consequence as $\CanCause$ even yields
	path-wise optimal coverage in the following sense.
	If $C$ is any SPR cause then $C \subseteq \cC$ by definition and for each path $\pi$ in $\cM$:
	\begin{align*}
		\pi \models (\neg \lozenge \CanCause) \wedge \lozenge \Effect &\implies \pi \models (\neg \lozenge C) \wedge \lozenge \Effect\text{ and} \\
		\pi \models \Diamond (C \wedge \Diamond \Effect) &\implies \pi \models \Diamond (\CanCause \wedge \Diamond \Effect).
	\end{align*}
	But then
	\begin{align*}
		\Pr^{\sched}_{\cM}(\Diamond (C \wedge \Diamond \Effect)) &\leqslant	\Pr^{\sched}_{\cM}(\Diamond (\CanCause \wedge \Diamond \Effect)),\\
		\Pr^{\sched}_{\cM}((\neg \lozenge C) \wedge \lozenge \Effect)) &\geqslant	\Pr^{\sched}_{\cM}((\neg \lozenge \CanCause) \wedge \Effect)
	\end{align*}
	for every scheduler $\sched$,which yields the claim.
\end{proof}


\begin{figure}[t]
	\resizebox{0.45\textwidth}{!}{

\begin{tikzpicture}[->,>=stealth',shorten >=1pt,auto ,node distance=0.5cm, thick]
	\node[scale=1, state] (s0) {$\init$};
	\node[scale=0.01, below = 1 of s0] (help) {};
	\node[scale=1, state] (eff) [right = 2.5 of help] {$\eff$};
	\node[scale=1, state] (noeff) [left =2.5 of help] {$\noeff$};
	\node[scale=1, state] (s1) [left=0.5 of help] {$s_1$};
	\node[scale=1, state] (s2) [right = 0.5 of help] {$s_2$};
	
	\draw[<-] (s0) --++(-0.55,0.55);
	\draw (s0) -- (eff) node[above, pos=0.5,scale=1] {$1/4$};
	\draw (s0) -- (noeff) node[above , pos=0.5,scale=1] {$1/4$};
	\draw (s0) -- (s1) node[right, pos=0.5,scale=1] {$1/2$};
	\draw (s1) -- (noeff) node[pos=0.5,scale=1] {$1/4$};
	\draw (s1) -- (s2) node[below, pos=0.5,scale=1] {$3/4$};
	\draw (s2) -- (eff) node[below, pos=0.5,scale=1] {$1$};
\end{tikzpicture}}
	\caption{Markov chain from Remark \ref{rem:canonical-not-fscore-optimal}}\label{fig:canonical-not-fscore-optimal}
\end{figure}
\begin{rem}\label{rem:canonical-not-fscore-optimal}		
	It is not true that the canonical SPR cause $\CanCause$ is f-score-optimal.
	To see this, Consider the Markov chain from Figure \ref{fig:canonical-not-fscore-optimal}.
	There we have $\CanCau = \{s_1\}$, which has $\precision(\CanCau) = \frac{3}{4}$ and $\recall(\CanCau) = \frac{3}{8}/(\frac{1}{4}+\frac{3}{8}) = \frac{3}{5}.$
	But the SPR cause $\{s_2\}$ has better f-score as its precision is $1$ and it has the same recall as $\CanCau$.
	\Ende
\end{rem}


\paragraph*{\bf F-score-optimal SPR cause.}
From Section~\ref{sec:comp-acc-measures-fixed-cause}, we see that f-score-optimal SPR causes in MDPs can be computed 
in polynomial space by computing the f-score for all potential SPR causes one by one in polynomial time (Theorem \ref{fscore-in-PTIME}).
As the space can be reused after each computation, this results in polynomial space.
For Markov chains, we can do better
and compute an f-score-optimal SPR cause in polynomial time
via a polynomial reduction to the stochastic shortest path problem:

\begin{thm}
	\label{thm:fscore-opt-MC}
	In Markov chains that have SPR causes, an f-score-optimal SPR cause can be computed in polynomial time.
\end{thm}

\begin{proof}
	We regard the given Markov chain 
	$\cM$ as an MDP with a singleton action set $\Act=\{\alpha\}$.
	As $\cM$ has SPR causes, the set $\cC$
	of states that constitute a singleton SPR cause
	is nonempty.
	We may assume that $\cM$ has no non-trivial (i.e., cyclic)
	bottom strongly connected components as we may collapse them.
	Let $w_c$ $=$ $\Pr_{\cM,c}(\Diamond \Effect)$. 
	
	We switch from $\cM$ to a new MDP $\cK$ with state space
	$S_{\cK}=S \cup \{\effcov,\noefffp\}$ with 
	fresh states $\noefffp$ and $\effcov$
	and the action set
	$\Act_{\cK}=\{\alpha,\gamma \}$.
	The MDP $\cK$ arises from $\cM$
	by adding 
	\begin{description}
		\item[(i)] for each SPR state $c\in \cC$ a fresh state-action pair $(c,\gamma)$ such that \mbox{$P_{\cK}(c,\gamma,\effcov)=w_c$} and $P_{\cK}(c,\gamma,\noefffp)=1{-}w_c$ and
		\item[(ii)] reset transitions to $\init$ with action label $\alpha$
		from the new state $\noefffp$ and
		all terminal states of $\cM$,
		i.e.,
		$P_{\cK}(\noefffp,\alpha,\init)=1$ and $P_{\cK}(s,\alpha,\init)=1$ 
		for $s \in \Effect$ or if $s$ is a terminal
		non-effect state of $\cM$.
	\end{description}
	So, exactly $\effcov$ is terminal in $\cK$, 
	and $\Act_{\cK}(c)=\{\alpha,\gamma\}$ for $c\in \cC$, while
	$\Act_{\cK}(s)=\{\alpha\}$ 
	for all other states $s$.
	Intuitively, taking action $\gamma$ in state $c \in \cC$ selects
	$c$ to be a cause state.
	The states in $\Effect$ represent uncovered effects in $\cK$,
	while $\effcov$ stands for covered effects.
	
	We assign weight $1$ to all states 
	in $U= \Effect \cup \{\noefffp\}$ 
	and weight $0$ to all other states of~$\cK$. 
	Let $V=\{\effcov\}$.
	Then, $f= \mathrm{E}^{\min}_{\cK}(\boxplus V)$
	and  an MD-scheduler $\sched$ for $\cK$ such that
	$\mathrm{E}^{\sched}_{\cK}(\boxplus V)=f$ are computable in
	polynomial time.
	Let $\cC_{\gamma}$ 
	denote the set of states $c\in \cC$ where
	$\sched(c)=\gamma$ 
	and let $\Cause$ be the set of states $c\in \cC_{\gamma}$
	where $\cM$ has a path satisfying
	$(\neg \cC_{\gamma}) \Until c$.
	Then, $\Cause$ is an SPR cause of $\cM$.
	With arguments as in Section~\ref{sec:comp-acc-measures-fixed-cause}
	we obtain
	$\fscore(\Cause)=2/(f{+}2)$. 
	
	It remains to show that $\Cause$ is f-score-optimal.
	Let $C$ be an arbitrary SPR cause.
	Then, $C \subseteq \cC$. Let $\tsched$ be the MD-scheduler
	for $\cK$ that schedules $\gamma$ in $C$ and
	$\alpha$ for all other states of $\cK$.
	Then, $\fscore(C)=2/(f^{\tsched}{+}2)$
	where $f^{\tsched}=\mathrm{E}^{\tsched}_{\cK}(\boxplus V)$.
	Hence, $f \leqslant f^{\tsched}$, which yields 
	$\fscore(\Cause) \geqslant \fscore(C)$.
\end{proof}

The na\"ive adaption of 
the construction presented in the proof of Theorem \ref{thm:fscore-opt-MC}
for MDPs would yield a stochastic game structure where the objective of
one player is to minimize the expected accumulated weight until reaching
a target state.
Although algorithms for \emph{stochastic shortest path (SSP) games} are known
\cite{patek1999stochastic}, 
they rely on assumptions on the game structure which would not
be satisfied here.
However, for the threshold problem \emph{SPR-f-score} where inputs are an MDP $\cM$, $\Effect$ and $\vartheta \in \Rat_{\geq 0}$ and the task is to decide the existence of an SPR cause whose f-score exceeds $\vartheta$, we can establish a polynomial reduction to SSP games, which yields an $\NP \cap \coNP$ upper bound:

\begin{thm}
	\label{fscore-threshold-poblem-via-stochMPgames}
	The decision problem SPR-f-score is in $\NP\cap \coNP$.
\end{thm}

Recall that for a given SPR cause $C$ and scheduler $\sched$ we have 
\begin{center}
	$\fscore^\sched(C) > \vartheta$ iff $\frac{\displaystyle 2\mathsf{tp}^\sched}{\displaystyle 2\mathsf{tp}^\sched + \mathsf{fp}^\sched + \mathsf{fn}^\sched} > \vartheta$.
\end{center}
In order to proof the upper bound of SPR-f-score we reformulate the condition of SPR-f-score.

\begin{lem}\label{lem:reformulating_fscore}
	Let $\cM=(S,\Act,P,\init)$ be an MDP with a set of terminal states $\Effect$, let $C$ be an SPR cause for $\Effect$ in $\cM$, and let $\vartheta$ be a rational.
	Then, $\fscore(C)>\vartheta$ iff
	\begin{align}
		\label{reformulate-fscore}
		2(1{-}\vartheta)\mathsf{tp}^\sched - \vartheta \mathsf{fp}^\sched
		-\vartheta \mathsf{fn}^\sched  >  0
	\end{align}
	for all schedulers $\sched$ for $\cM$ with $\Pr^{\sched}_{\cM}(\Diamond \Effect)>0$.
\end{lem}

\begin{proof}
	Assume that $\fscore(C)>\vartheta$ and let $\sched$ be a scheduler with $\Pr^{\sched}_{\cM}(\Diamond \Effect) > 0$.
	If $\Pr^{\sched}_{\cM}(\Diamond C)=0$, then $\fscore(C)$ would be $0$. So, $\Pr^{\sched}_{\cM}(\Diamond C) > 0$.
	Then,
	\begin{center}
		$\fscore^{\sched}(C) = 
		\frac{\displaystyle 2\mathsf{tp}^\sched}{\displaystyle 2\mathsf{tp}^\sched + \mathsf{fp}^\sched + \mathsf{fn}^\sched} > \vartheta, \ \ \
		$ which implies $ \ \ \ 2(1{-}\vartheta)\mathsf{tp}^\sched - \vartheta \mathsf{fp}^\sched
		-\vartheta \mathsf{fn}^\sched  >  0$
	\end{center}
	
	Now, suppose that \eqref{reformulate-fscore} holds for a scheduler $\sched$ with $\Pr^{\sched}_{\cM}(\Diamond \Effect) > 0$. Let $\sched$ be a scheduler that minimizes 
	$\fscore^{\sched}(C)$. Such a scheduler exists by Theorem \ref{fscore-in-PTIME}.
	From \eqref{reformulate-fscore}, we conclude 
	\begin{center}
		$2(1{-}\vartheta)\mathsf{tp}^\sched - \vartheta \mathsf{fp}^\sched
		-\vartheta \mathsf{fn}^\sched  >  0$
	\end{center}
	and hence that $\fscore^\sched(C)>\vartheta$ as above.
\end{proof}

\begin{proof}[Proof of Theorem \ref{fscore-threshold-poblem-via-stochMPgames}]
	Let $\cM=(S,\Act,P,\init)$ be an MDP, $\Effect\subseteq S$ a set of terminal states, and $\vartheta$ a rational.
	Consider $\cC$, the set of states $c\in S \setminus \Effect$ where $\{c\}$ is an SPR cause.
	If $\cC$ is empty then the threshold problem is trivial as there is no SPR cause at all.
	Thus, we suppose that $\cC$ is nonempty.
	
	Note that $\Pr^{\min}_{\cM,c}(\Diamond \Effect)>0$ for all $c\in \cC$. 
	As the states in $\Effect$ are not part of any end component of $\cM$,  no state $c\in \cC$ is contained in an end component of $\cM$ either.
	Let $\cN=(S_{\cN},\Act_{\cN},P_{\cN},\init_{\cN})$ be the MEC-quotient of $\cM$ with the new additional terminal state $\bot$. The MEC-quotient $\cN$ contains the states from $\Effect$ and $\cC$.
	
	\paragraph*{Claim 1: }
	There is an SPR cause $C$ for $\Effect$ in $\cM$ with $\fscore(C)>\vartheta$ if and only if there is an SPR cause $C^\prime$ for $\Effect$ in $\cN$ with $\fscore(C^\prime)>\vartheta$.
	
	\noindent
	{\it Proof of Claim 1.}
	We first observe that all reachability probabilities involved in the claim do not depend on the behavior during the traversal of MECs. Furthermore, staying inside a MEC in $\cM$ can be mimicked in $\cN$ by moving to $\bot$, and vice versa. 
	More precisely, let $C\subseteq \cC$. Then,
	analogously to Lemma \ref{lem:probabilities_MEC-quotient},  for each scheduler $\sched$ for $\cM$, there is a scheduler $\tsched$ for $\cN$, and vice versa,
	such that 
	\begin{itemize}
		\item
		$\Pr^{\sched}_{\cM}(\Diamond \Effect \mid (\neg C) \Until c) = \Pr^{\tsched}_{\cN}(\Diamond \Effect \mid (\neg C) \Until c)$ for all $c\in C$ where the values are defined, 
		\item
		$\Pr^{\sched}_{\cM}(\Diamond \Effect) = \Pr^{\tsched}_{\cN}(\Diamond \Effect)$,
		\item
		$\Pr^{\sched}_{\cM}(\Diamond \Effect\mid \Diamond C) = \Pr^{\tsched}_{\cN}(\Diamond \Effect \mid \Diamond C)$ if the values are defined, and 
		\item
		$\Pr^{\sched}_{\cM}(\Diamond C\mid \Diamond \Effect) = \Pr^{\tsched}_{\cN}(\Diamond C \mid \Diamond \Effect)$ if the values are defined.
	\end{itemize}
	Hence, $C$ is an SPR cause for $\Effect$ in $\cM$ if and only if it is in $\cN$ and furthermore, if it is an SPR cause, the f-score of $C$ in $\cM$ and in $\cN$ agree. This finishes the proof of Claim~1.

	\paragraph*{\it Model transformation 
		for ensuring positive effect probabilities.}
	Recall that the f-score is only defined for  schedulers reaching $\Effect$ with positive probability.
	Now, we will provide a further model transformation that will ensure that $\Effect$ is reached with positive probability under all schedulers. 
	If this is already the case, there is nothing to do. So, we assume now that $\Pr^{\min}_{\cN,\init_{\cN}}(\Diamond  \Effect)=0$.
	
	We define the subset of states from which $\Eff$ can be avoided as $D\subseteq S_{\cN}$ by
	\[
	D\eqdef \{s\in S_{\cN} \mid \Pr^{\min}_{\cN,s}(\Diamond \Effect) =0\}.
	\]
	Note that $\init_{\cN} \in D$. 
	For each $s\in D$, we further define the set of actions minimizing the reachability probability of $\Eff$ from $s$ by
	\[
	\Act^{\min}(s)=\{\alpha \in \Act_{\cN}(s) \mid P_{\cN}(s,\alpha,D)=1\}.
	\]
	Finally, let $E\subseteq D$ be the set of states that are reachable from $\init_{\cN}$ when only choosing actions from $\Act^{\min}(\cdot)$. Note that $E$ does not contain any states from $\cC$, meaning no state in $E$ constitutes a singleton SPR cause.
	
	All schedulers that reach $\Effect$ with positive probability in $\cN$ have to leave the sub-MDP consisting of $E$ and the actions in $\Act^{\min}(\cdot)$ at some point. Let us call this sub-MDP $\cN^{\min}_E$.
	We define the set of state-action pairs $\Pi$ that leave the sub-MDP $\cN^{\min}_E$:
	\[
	\Pi\eqdef \{ (s,\alpha) \mid s\in E\text{ and } \alpha\in \Act_{\cN}(s) \setminus \Act^{\min}(s)\}.
	\]
	We now construct a further MDP $\cK$. The idea is that $\cK$ behaves like the end-component free MDP $\cN$ after initially a scheduler is forced to  choose  a probability distribution over 
	state-action pairs from $\Pi$. In this way, $\Effect$ is reached with positive probability under all schedulers.
	Given an SPR cause, we will observe that for the f-score of this cause under a scheduler, it is only important how large the probabilities with which state action pairs from $\Pi$ are chosen are relative to each other while the absolute values are not important. Due to this observation, for each SPR cause $C$ and for each scheduler $\sched$ for $\cN$ that reaches $\Effect$ with positive probability, we can then construct a scheduler for $\cK$ that leads to the same recall and precision of $C$.
	
	Formally, $\cK$ is defined as follows: The state space is $S_{\cN}\cup \{\init_{\cK}\}$ where $\init_{\cK}$ is a fresh initial state.
	For all states in $S_{\cN}$, the same actions as in $\cN$ are available with the same transition probabilities. I.e., for all $s,t\in S_{\cN}$,
	\[
	\Act_{\cK}(s)\eqdef \Act_{\cN}(s) \quad \text{and} \quad P_{\cK}(s,\alpha,t)\eqdef P_{\cN}(s,\alpha,t) \text{ for all }\alpha \in \Act_{\cK}(s).
	\]
	For each state-action pair $(s,\alpha)$ from $\Pi$, we now add a new action $\beta_{(s,\alpha)}$ that is enabled only in $\init_{\cK}$. 
	These are all actions enabled in $\init_{\cK}$, i.e., 
	\[
	\Act_{\cK}(\init_{\cK}) \eqdef \{\beta_{(s,\alpha)}\mid (s,\alpha)\in \Pi\}.
	\]
	For each state $t\in S_{\cN}$, we define the transition probabilities under $\beta_{(s,\alpha)}$ by
	\[
	P_{\cK}(\init_{\cK},\beta_{(s,\alpha)}, t) \eqdef P_{\cN}(s,\alpha,t).
	\]
	
	\paragraph*{Claim 2: }
	A subset $C\subseteq \cC$ such that for all $c \in C : \Pr^{\max}_\cN(\neg C \Until c) > 0$ is an SPR cause  for $\Effect$ in $\cN$ with $\fscore(C)>\vartheta$
	if and only if 
	for all schedulers $\tsched$ for $\cK$, we have
	\begin{align}
		\label{decision-fscore-claim2}
		2(1{-}\vartheta)\mathsf{tp}^\tsched_\cK - \vartheta \mathsf{fp}^\tsched_\cK
		-\vartheta \mathsf{fn}^\tsched_\cK  >  0.
	\end{align}
	
	\noindent
	{\it Proof of Claim 2.}
	We first prove the direction ``$\Rightarrow$''. So, let $C$ be an SPR cause for $\Effect$ in end-component free MDP $\cN$ with $\fscore(C)>\vartheta$.
	As first observation we have that in order to prove \eqref{decision-fscore-claim2} for all schedulers $\tsched$ for $\cK$, it suffices to consider schedulers $\tsched$ that start with a deterministic choice for state $\init_{\cK}$ and then behave in an arbitrary way. 
	\begin{itemize}
		\item []
		To see this, we consider the MDP $\cK_C$ which consists of two copies of $\cK$: ``before $C$'' and ``after $C$''.
		When $\cK_C$ enters a $C$-state in the first copy (``before $C$''), it switches to the second copy (``after $C$'') and stays there forever.
		Let us write $(s,1)$ for state $s$ in the first copy and $(s,2)$ for the copy of state $s$ in the second copy.
		Thus, in $\cK_C$ the event $\Diamond C \wedge \Diamond \Effect$ is equivalent to reaching a state $(\eff,2)$ where $\eff\in \Effect$, while $\Diamond C \wedge \neg \Diamond \Effect$ is equivalent to reaching a non-terminal state in the second copy, while $\neg \Diamond C \wedge \Diamond \Effect$ corresponds to the event reaching an effect state in the first copy.
		
		Obviously, there is a one-to-one-correspondence of the schedulers of $\cK$ and $\cK_C$. 
		As $\cK$ has no end components so does $\cK_C$.
		Therefore, a terminal state will be reached almost surely under every scheduler.
		Furthermore, we equip $\cK_C$ with a weight function on states which assigns 
		\begin{itemize}
			\item
			weight $2(1{-}\vartheta)$ to the states $(\eff,2)$ where $\eff \in \Effect$,
			\item
			weight $-\vartheta$ to the states $(\eff,1)$ where $\eff \in \Effect$ and
			to the states $(s,2)$ where $s$ is a terminal non-effect state in $\cK$ 
			(and $\cK_C$), and
			\item
			weight 0 to all other states.
		\end{itemize}
		Let $V$ denote the set of all terminal states in $\cK_C$.
		Then, the expression on the left hand side of \eqref{decision-fscore-claim2} equals $\mathrm{E}^{\tsched}_{\cK_C}(\boxplus V)$, the expected accumulated weight until reaching a terminal state under scheduler $\tsched$. 
		Hence,  \eqref{decision-fscore-claim2} holds for all schedulers $\tsched$ in $\cK$ if and only if $\mathrm{E}^{\min}_{\cK_C}(\boxplus V) >0$. 
		
		It is well-known that the minimal expected accumulated weight in EC-free MDPs is achieved by an MD-scheduler \cite{BaierK2008}. 
		Thus, there is an MD-scheduler $\tsched$ of $\cK_C$ with $\mathrm{E}^{\min}_{\cK_C}(\boxplus V) = \mathrm{E}^{\tsched}_{\cK_C}(\boxplus V)$. 
		When viewed as a scheduler of $\cK$, $\tsched$ behaves memoryless deterministic before reaching $C$.
		In particular, $\tsched$'s initial choice in $\init_{\cK}$ is deterministic.
	\end{itemize}
	Recall the set $\Act^{\min}(s)$ of actions minimizing the reachability probability of $\Eff$ from $s$.
	Consider a scheduler $\tsched$ for $\cK$ with a deterministic choice $\tsched(\init_{\cK})(\beta_{(s,\alpha)})=1$ where $(s,\alpha)\in \Pi$. 
	To construct an analogous scheduler $\sched$ of $\cN$, we pick an MD-scheduler $\usched$ of the sub-MDP $\cN^{\min}_E$ of $\cN$ induced by the state-action pairs $(u,\beta)$ where $u\in E$ and $\beta \in \Act^{\min}(u)$ such that there is a $\usched$-path from $\init_{\cN}$ to state $s$.
	
	Scheduler $\sched$ of $\cN$ operates with the mode $\mathfrak{m}_1$ and the modes $\mathfrak{m}_{2,t}$ for $t\in S_{\cN}$.
	In its initial mode $\mathfrak{m}_1$, scheduler $\sched$ behaves as $\usched$ as long as state $s$ has not been visited. 
	When having reached state $s$ in mode $\mathfrak{m}_1$, then $\sched$ schedules the action $\alpha$ with probability 1. 
	Let $t \in S_{\cN}$ be the state that $\sched$ reaches via the $\alpha$-transition from $s$. 
	Then, $\sched$ switches to mode $\mathfrak{m}_{2,t}$ and behaves from then on as the residual scheduler $\residual{\tsched}{\varpi}$ of $\tsched$ for the $\tsched$-path $\varpi = \init_{\cK} \, \beta_{(s,\alpha)} \, t$ in $\cK$. 
	That is, after having scheduled the action $\beta_{(s,\alpha)}$, scheduler $\sched$ behaves exactly as $\tsched$.
	
	Let $\lambda$ denote $\sched$'s probability to leave mode $\mathfrak{m}_1$, which equals $\usched$'s probability to reach $s$ from $\init_{\cN}$. 
	Thus, $\lambda = \Pr_{\cN}^{\usched}(\Diamond s)$ when $\usched$ is viewed as a scheduler of $\cN$. 
	As $E$ is disjoint from $C$ and $\Effect$, scheduler $\sched$ stays forever in mode $\mathfrak{m}_1$ and never reaches a state in $C \cup \Effect$ with probability $1{-}\lambda$. 
	
	$\sched$ and $\tsched$ behave identically after choosing the state-action pair $(s,\alpha) \in \Pi$ or the corresponding action $\beta_{(s,\alpha)}$, respectively, which implies that 
	\begin{itemize}
		\item 
		$\Pr^{\sched}_{\cN}(\Diamond C \land \Diamond \Effect) = \lambda \cdot \Pr^{\tsched}_{\cK}(\Diamond C \land \Diamond \Effect)$,
		\item
		$\Pr^{\sched}_{\cN}(\Diamond \Effect) = \lambda \cdot \Pr^{\tsched}_{\cK}(\Diamond \Effect)$ and 
		\item
		$\Pr^{\sched}_{\cN}(\Diamond C \land \neg \Diamond \Effect)  = \lambda \cdot \Pr^{\tsched}_{\cK}(\Diamond C \land \neg \Diamond \Effect)$.
	\end{itemize}
	
	Recall the sub-MDP $\cN^{\min}_E$ consisting of $E$ and the actions in $\Act^{\min}(\cdot)$.
	As $\sched$ leaves the sub-MDP $\cN^{\min}_E$ with probability $\lambda >0$, we have $\Pr^{\sched}_{\cN}(\Diamond \Effect)>0$.
	By Lemma \ref{lem:reformulating_fscore}, we can conclude that 
	\[
	2(1{-}\vartheta)\mathsf{tp}^\sched_\cN - \vartheta \mathsf{fp}^\sched_\cN
	-\vartheta \mathsf{fn}^\sched_\cN  >  0.
	\]
	By the equations above, this in turn implies that 
	\[
	2(1{-}\vartheta)\mathsf{tp}^\tsched_\cK - \vartheta \mathsf{fp}^\tsched_\cK
	-\vartheta \mathsf{fn}^\tsched_\cK  >  0.
	\]
		
	For the direction ``$\Leftarrow$'', first recall that any subset of $\cC$ satisfying (M) is an SPR cause for $\Effect$ in $\cN$ by definition of $\cC$. Now, let $\sched$ be a scheduler for $\cN$ with $\Pr^{\sched}_{\cN}(\Diamond \Effect)>0$.
	Let $\Gamma$ be the set of finite $\sched$-paths $\gamma$ in the sub-MDP $\cN^{\min}_E$  such that 
	$\sched$ chooses an action in $\cA = \Act_{\cN}(\mathit{last}(\gamma))\setminus \Act^{\min}(\mathit{last}(\gamma))$ with positive probability after $\gamma$ where $\mathit{last}(\gamma)$ denotes the last state of $\gamma$.
	Let
	\[
	q \eqdef \sum_{\gamma \in \Gamma}  \sum_{\alpha\in \cA} P_{\cN}(\gamma)\cdot \sched(\gamma)(\alpha).
	\]
	So, $q$ is the overall probability that a state-action pair from $\Pi$ is chosen under $\sched$.
	We now define a scheduler $\tsched$ for $\cK$: For each action $\gamma\in\Gamma$ ending in a state $s$ and each action $\alpha\in \Act_{\cN}(s)\setminus \Act^{\min}(s)$, the scheduler $\tsched$ chooses action $\beta_{(s,\alpha)}$ in $\init_{\cK}$ with probability $P_{\cN}(\gamma)\cdot \sched(\gamma)(\alpha) / q$.
	When reaching a state $t$ afterwards, $\tsched$ behaves like $\residual{\sched}{\gamma\, \alpha \, t}$ afterwards.
	Note that by definition this indeed defines a probability distribution over the actions in the initial state $\init_{\cK}$.
	
	By assumption, we know that now
	\[
	2(1{-}\vartheta)\mathsf{tp}^\tsched_\cK - \vartheta \mathsf{fp}^\tsched_\cK
	-\vartheta \mathsf{fn}^\tsched_\cK  >  0.
	\]
	As the probability with which an action $\beta_{(s,\alpha)}$ is chosen by $\tsched$ for a $(s,\alpha) \in \Pi$ is $1/q$ times the probability that $\alpha$ is chosen in $s$ to leave the sub-MDP $\cN^{\min}_E$  under $\sched$ in $\cN$ and as the residual behavior is identical, we conclude that 
	\begin{align*}
		2(1{-}\vartheta)\mathsf{tp}^\sched_\cN - \vartheta \mathsf{fp}^\sched_\cN
		-\vartheta \mathsf{fn}^\sched_\cN
		= \,\, 	q\cdot ( 2(1{-}\vartheta)\mathsf{tp}^\tsched_\cK - \vartheta \mathsf{fp}^\tsched_\cK
		-\vartheta \mathsf{fn}^\tsched_\cK)  >  0.
	\end{align*}
	By Lemma \ref{lem:reformulating_fscore}, this shows that $\fscore(C)>\vartheta$ in $\cN$ and finishes the proof of Claim 2.
	
	\paragraph*{\it Construction of a game structure.}
	Recall the set of singleton SPR causes $\cC$.
	We now construct a \emph{stochastic shortest path game} (see \cite{patek1999stochastic}) to check whether there is a subset $C\subseteq \cC$ such that \eqref{decision-fscore-claim2} holds in the EC-free MDP $\cK$ in which visiting effect states always has a non-zero probability.
	Such a game is played on an MDP-like structure with the only difference that the set of states is partitioned into two sets indicating which player controls which states.
	
	The game $\cG$ has states $(S_{\cK}\times \{\yes,\no\}) \cup \cC\times\{\choice\}$.
	On the subset $S_{\cK}\times\{\yes\}$, all actions and transition probabilities are just as in $\cK$ and this copy of $\cK$ cannot be left.
	Formally, for all $s,t\in S_{\cK}$ and $\alpha \in \Act_{\cK}(s)$, we have $\Act_{\cG}((s,\yes))=\Act_{\cK}(s)$ and $P_{\cG}((s,\yes),\alpha,(t,\yes ))=P_{\cK}(s,\alpha,t)$.
	
	In the ``$\no$''-copy, the game also behaves like $\cK$ but when a state in $\cC$ would be entered, the game moves to a state in $\cC\times \{\choice\}$ instead.
	In a state of the form $(c,\choice)$ with $c\in \cC$, two actions $\alpha$ and $\beta$ are available. Choosing $\alpha$ leads to the state $(c,\yes)$ while choosing $\beta$ leads to $(c,\no)$ with probability $1$.
	
	Formally, this means that for all state $s\in S_{\cK}$, we define $\Act_{\cG}((s,\no))=\Act_{\cK}(s)$ and for all actions $\alpha\in \Act_{\cK}(s)$:
	\begin{itemize}
		\item
		$P_{\cG}((s,\no),\alpha,(t,\no))=P_{\cK}(s,\alpha,t)$
		for all states $t\in S_{\cK}\setminus \cC$ 
		\item
		$P_{\cG}((s,\no),\alpha,(c,\choice))=P_{\cK}(s,\alpha,c)$
		for all states $c\in \cC$
	\end{itemize}
	For states $s\in S_{\cK}$, $c\in \cC$, and $\alpha\in \Act_{\cK}(s)$,
	we furthermore define: 
	\begin{center}
		$P_{\cG}((c,\choice),\alpha,(c,\yes))=P_{\cG}((c,\choice),\beta,(c,\no))=1$.
	\end{center}
	Intuitively speaking, whether a state $c\in \cC$ should belong to the cause set can be decided in the state $(c,\choice)$. 
	The ``$\yes$''-copy encodes that an effect state has been selected. 
	More concretely, the ``$\yes$-copy'' is entered as soon as $\alpha$ has been chosen in some state $(c,\choice)$ and will never be left from then on.
	The ``$\no$''-copy of $\cK$ then encodes that no state $c\in \cC$ which has been selected to become a cause state has been visited so far.  
	That is, if the current state of a play in $\cG$ belongs to the $\no$-copy then	in all previous decisions in the states $(c,\choice)$, action $\beta$ has been chosen.
	
	Finally, we equip the game with a weight structure.
	All states in $\Effect\times\{\yes\}$ get weight $2(1-\vartheta)$. All remaining terminal states in $S_{\cK}\times \{\yes\}$ get weight $-\vartheta$. 
	All states in $\Effect \times \{\no\}$ get weight $-\vartheta$. 
	All remaining states have weight $0$.

	The game is played between two players $0$ and $1$. Player $0$ controls all states in $\cC\times \{\choice\}$ while player $1$ controls the remaining states.
	The goal of player $0$ is to ensure that the expected accumulated weight is $>0$.

	\paragraph*{Claim 3: }
	Player $0$ has a winning strategy ensuring that the expected accumulated weight is $>0$ in the game $\cG$ if and only if there is a subset $C\subseteq \cC$ in $\cK$ which satisfies for all $c \in C: \Pr^{\max}_\cK(\neg C \until c) > 0$ and for all schedulers $\tsched$ for $\cK$ we have
	\begin{align}
		\label{decision-fscore-claim3}
		2(1{-}\vartheta)\mathsf{tp}^\tsched_\cK - \vartheta \mathsf{fp}^\tsched_\cK
		-\vartheta \mathsf{fn}^\tsched_\cK  >  0.
	\end{align}
	
	\noindent
	{\it Proof of Claim 3.}
	As $\cK$ has no end components, also in the game $\cG$ a terminal state is reached almost surely under any pair of strategies. Hence, we can rely on the results of \cite{patek1999stochastic} that state that both players have an optimal memoryless deterministic strategy.
	
	We start by proving direction ``$\Rightarrow$'' of Claim 3. Let $\zeta$ be a memoryless deterministic winning strategy for player $0$. I.e., $\zeta$ assigns to each state in $\cC\times \{\choice\}$ an action from $\{\alpha, \beta\}$.
	We define 
	\[
	\cC_{\alpha} \eqdef \{c\in \cC \mid \zeta((c,\choice))=\alpha \}.
	\]
	Note that $\cC_{\alpha}$ is not empty as otherwise a positive expected accumulated weight in the game is not possible. 
	(Here we use the fact that only the effect states in the $\yes$-copy have positive weight and that the $\yes$-copy can only be entered by taking $\alpha$ in one of the states $(c,\choice)$.)
	To ensure for all $c \in C: \Pr^{\max}_\cK(\neg C \until c) > 0$, we remove states that cannot be visited as the first state of this set:
	\[
	C\eqdef \{c\in \cC_{\alpha} \mid \cK,c \models \exists (\neg \cC_{\alpha})\Until c\}.
	\]
	Note that the strategies for player $0$ in $\cG$ which correspond to the sets $\cC_{\alpha}$ and $C$ lead to exactly the same plays.
	
	Let $\tsched$ be a scheduler for $\cK$. This scheduler can be used as a strategy for player $1$ in $\cG$.
	Let us denote the expected accumulated weight when player $0$ plays according to $\zeta$ and player $1$ plays according to $\tsched$ by $w(\zeta,\tsched)$.
	As $\zeta$ is winning for player 0 we have
	$
	w(\zeta,\tsched) > 0.
	$
	By the construction of the game, it follows directly that 
	\[
	w(\zeta,\tsched) = 2(1{-}\vartheta)\mathsf{tp}^\tsched_\cK - \vartheta \mathsf{fp}^\tsched_\cK
	-\vartheta \mathsf{fn}^\tsched_\cK. 
	\]
	Putting things together yields:
	\begin{align}
		\label{decision-fscore-claim3+}
		2(1{-}\vartheta)\mathsf{tp}^\tsched_\cK - \vartheta \mathsf{fp}^\tsched_\cK
		-\vartheta \mathsf{fn}^\tsched_\cK  >  0
	\end{align}
	
	For the other direction, suppose there is a set $C\subseteq \cC$ that satisfies $\Pr^{\max}_\cK(\neg C \until c) > 0$ for all $c \in C$ and \eqref{decision-fscore-claim3} for all schedulers $\tsched $ for $\cK$.
	We define the MD-strategy $\zeta$ from $C$ by letting $\zeta((c,\choice))=\alpha $ if and only if $c\in C$. For any strategy $\tsched$ for player $1$, we can again view $\tsched$ also as a scheduler for $\cK$. Equation \eqref{decision-fscore-claim3+} holds again and shows that the expected accumulated weight in $\cG$ is positive if player $0$ plays according to $\zeta$ against any strategy for player $1$. This finishes the proof of Claim 3.
	
	\paragraph*{Putting together Claims 1-3.} 
	We conclude that there is an SPR cause $C$ in the original MDP $\cM$ with $\fscore(C)>\vartheta$ if and only if player $1$ has a winning strategy in the constructed game $\cG$.
	As both players have optimal MD-strategies in $\cG$ \cite{patek1999stochastic}, the decision problem is in $\NP \cap \coNP$: We can guess the MD-strategy for player $0$ and solve the resulting stochastic shortest path problem in polynomial time \cite{BT91} to obtain an $\NP$-upper bound. Likewise, we can guess the MD-strategy for player $1$ and solve the resulting stochastic shortest path problem to obtain the $\coNP$-upper bound. 
\end{proof}


\paragraph*{\bf Optimality and threshold constraints for GPR causes.}
Computing optimal GPR causes for either quality measure can be done in polynomial space by considering all cause candidates, checking (G) in $\coNP$ and computing the corresponding quality measure in polynomial time (Section~\ref{sec:comp-acc-measures-fixed-cause}).
As the space can be reused after each computation, this results in polynomial space.
However, we show that no polynomial-time algorithms can be expected as the corresponding threshold problems are $\NP$-hard.
Let GPR-covratio (resp. GPR-recall, GPR-f-score) denote the decision
problems: 
Given $\cM,\Effect$ and $\vartheta \in \Rat$, decide
whether there exists a GPR cause with
coverage ratio (resp. recall, f-score) at least $\vartheta$.

\begin{thm}
	\label{thm:GPR-recall-NPhard-and-in-PSPACE}
	The problems GPR-covratio, GPR-recall and GPR-f-score
	are $\NP$-hard and belong to $\Sigma^P_2$.
	For Markov chains, all three problems are $\NP$-complete.
	$\NP$-hardness even holds
	for tree-like Markov chains.
\end{thm}  

\begin{proof}
	\textit{$\Sigma^P_2$-membership.}
	The algorithms for GPR-covratio, GPR-recall and GPR-f-score rely on the guess-and-check principle: 
	they start by non-deterministically guessing a set $\Cause \subseteq S$,
	then check in $\coNP$ whether
	$\Cause$ constitutes a GPR cause
	(see Section~\ref{sec:check}) and finally check $\relcov(\Cause) \leq \vartheta$ (with standard techniques),
	resp. $\ratiocov(\Cause) \leq \vartheta$,
	resp. $\fscore(\Cause) \leq \vartheta$
	(Theorem \ref{fscore-in-PTIME})
	in polynomial time.
	The alternation between the existential quantification for guessing $\Cause$ and the universal quantification for the $\coNP$ check of the GPR condition results in the complexity $\Sigma^P_2$ of the polynomial-time hierarchy.
	
	\textit{$\NP$-membership for Markov chains.}
	$\NP$-membership for all three problems within Markov chains is straightforward as we
	may non-deterministically guess a cause and check in
	polynomial time whether it constitutes a GPR cause and satisfies the threshold
	condition for the recall, coverage ratio or f-score.
	
	\textit{$\NP$-hardness of GPR-recall and GPR-covratio.}
	With arguments as in the proof of Lemma \ref{lemma:recall-opt=ratio-opt},
	the problems GPR-recall and GPR-covratio are polynomially interreducible
	for Markov chains.
	Thus, it suffices to prove NP-hardness of GPR-recall.
	For this, we provide a polynomial reduction from the knapsack problem.
	The input of the latter are sequences
	$A_1,\ldots,A_n,A$ and $B_1,\ldots,B_n,B$ of positive natural numbers
	and the task is to decide whether there exists a subset $I$ of $\{1,\ldots,n\}$ such that
	\begin{align}
		\label{knapsack1}
		\sum_{i\in I} A_i \ < \ A \qquad \text{and} \qquad
		\sum_{i\in I} B_i \ \geqslant \ B
	\end{align}
	Let $K$ be the maximum of $A,A_1,\ldots,A_n,B,B_1,\ldots,B_n$ and $N = 8(n{+}1)\cdot (K{+}1)$.
	We then define
	\begin{align*}
		a_i&=\frac{A_i}{N},&
		a&=\frac{A}{N},&
		b_i&=\frac{B_i}{N},&
		b=&\frac{B}{N}.
	\end{align*}
	Then, $a, a_1,\ldots,a_n,b,b_1,\ldots,b_n$
	are positive rational numbers strictly smaller than $\frac{1}{8(n{+}1)}$,
	and \eqref{knapsack1} can be rewritten as:
	\begin{align}
		\label{knapsack2}
		\sum_{i\in I} a_i \ < \ a \qquad \text{and} \qquad
		\sum_{i\in I} b_i \ \geqslant \ b
	\end{align}
	For $i\in \{1,\ldots,n\}$, let 
	$p_i=2(a_i+b_i)$
	and
	$w_i=\frac{b_i}{p_i}=\frac{1}{2}\cdot \frac{b_i}{a_i+b_i}$.
	Then,
	$0 < p_i < \frac{1}{2(n{+}1)}$ and $0< w_i < \frac{1}{2}$.
	Moreover,$p_i \bigl(\frac{1}{2}-w_i \bigr) = a_i$ and $p_i \cdot w_i =  b_i$.
	Hence, \eqref{knapsack2} can be rewritten as:  
	\begin{align*}
		\sum\limits_{i\in I} p_i \bigl(\frac{1}{2}-w_i\bigr) \ < \ a
		\qquad \text{and} \qquad \sum\limits_{i\in I} p_i w_i \ \geqslant \ b
	\end{align*}
	which again is equivalent to:
	\begin{align}
		\label{knapsack3}
		\frac{\sum\limits_{i\in I_0} p_iw_i}{\sum\limits_{i\in I_0} p_i} \ > \
		\frac{1}{2}
		\qquad \text{and} \qquad
		\sum\limits_{i\in I_0} p_i w_i \ \geqslant \ \ p_0+b
	\end{align}
	where $p_0 = 2a$, $w_0=1$ and $I_0=I\cup\{0\}$.
	Note that $a<\frac{1}{8(n{+}1)}$ and hence $p_0 <\frac{1}{4(n{+}1)}$.
	
	Define a tree-shape Markov chain $\cM$ with non-terminal states
	$\init$, $s_0,s_1,\ldots,s_n$, 
	and terminal states
	$\eff_0,\ldots,\eff_n$, $\effuncov$ and
	$\noeff,\noeff_1,\ldots,\noeff_n$.
	Transition probabilities are as follows:
	\begin{itemize}
		\item
		$P(\init,s_i)=p_i$ for $i=0,\ldots,n$
		\item
		$P(\init,\effuncov) \ = \
		\frac{1}{2}-\sum\limits_{i=0}^n p_iw_i$
		
		\item   
		$P(\init,\noeff)=1-\sum\limits_{i=0}^n p_i - P(\init,\effuncov)$,
		\item
		$P(s_i,\eff_i)=w_i$, $P(s_i,\noeff_i)=1{-}w_i$ for $i=1,\ldots,n$
		\item
		$P(s_0,\eff_0)=1=w_0$.
	\end{itemize}
	
	Note that $p_0+p_1+\ldots+p_n  < \frac{1}{2}$ as all $p_i$'s are
	strictly smaller
	than $\frac{1}{2(n{+}1)}$. As the $w_i$'s are bounded by 1,
	this yields $0 < P(\init,\effuncov) <\frac{1}{2}$ and
	$0< P(\init,\noeff) < 1$.
	
	The graph structure of $\cM$ is indeed a tree and $\cM$ can be constructed from the values
	$A,A_1,\ldots,A_n$ and $B,B_1,\ldots,B_n$ in polynomial time.
	For $\Effect = \{\effuncov\}\cup\{\eff_i : i=0,1,\ldots,n\}$
	we have:
	\[
	\Pr_{\cM}(\Diamond \Effect)
	\ \ = \ \
	\sum_{i=0}^n p_iw_i + P(\init,\effuncov) \ \ = \ \ \frac{1}{2}.
	\]
	As the values $w_1,\ldots,w_n$ are strictly smaller than $\frac{1}{2}$,
	we have
	$\Pr_{\cM}(\ \Diamond \Effect \ | \ \Diamond C \ ) < \frac{1}{2}$
	for each nonempty subset $C$ of $\{s_1,\ldots,s_n\}$.
	Thus, the only candidates for GPR causes are the sets
	$C_I =\{s_i : i\in I_0\}$ where $I \subseteq \{1,\ldots,n\}$
	where as before $I_0=I\cup\{0\}$.
	Note that for all states $s\in C_I$ there is a path satisfying
	$(\neg C_I)\Until s$. Thus, $C_I$ is a GPR cause iff $C_I$ satisfies (G).
	We have:
	\[
	\Pr_{\cM}(\ \Diamond \Effect \ | \ \Diamond C_I \ ) \ \ = \ \
	\frac{\sum\limits_{i\in I_0} p_i w_i}{\sum\limits_{i\in I_0} p_i}
	\]
	and
	\[
	\recall(C_I) \ \ = \ \
	\Pr_{\cM}
	(\ \Diamond (C_I \wedge \Diamond \Effect) \ | \ \Diamond \Effect \ )
	\ \ = \ \
	2 \cdot \sum_{i\in I_0} p_iw_i.
	\]
	Thus, $C_I$ is a GPR cause with recall at least $2(p_0+b)$
	if and only if the two conditions in \eqref{knapsack3} hold,
	which again is equivalent to the satisfaction of the conditions
	in \eqref{knapsack1}.
	But this yields that $\cM$ has a GPR cause with recall at least
	$2(p_0+b)$ if and only if the knapsack problem is solvable
	for the input $A,A_1,\ldots,A_n,B,B_1,\ldots,B_n$.
	
	\textit{$\NP$-hardness of GPR-f-score.}
	Using similar ideas, we also provide a polynomial reduction
	from the knapsack problem.
	Let $A,A_1,\ldots,A_n,B,B_1,\ldots,B_n$ be an input for the knapsack
	problem.
	We replace the $A$-sequence with $a,a_1,\ldots,a_n$ where
	$a=\frac{A}{N}$ and $a_i=\frac{A_i}{N}$
	where $N$ is as before.
	The topological structure of the Markov chain that we are going to construct
	is the same as in the NP-hardness proof for GPR-recall.
	
	Next, we will consider the polynomial-time computable values
	$p_0,p_1,\ldots,p_n \in \ ]0,1[$ (where $p_i=P(\init,s_i)$),
	$w_1,\ldots,w_n \in \ ]0,1[$ (where $w_i=P(s_i,\eff_i)$)
	and auxiliary variables $\delta \in \ ]0,1[$ and $\lambda > 1$    
	such that:
	\begin{enumerate}
		\item [(1)] $p_0+p_1 + \ldots + p_n < \frac{1}{2}$
		\item [(2)] $\lambda = \frac{p_0 + \frac{1}{2} - \delta}{p_0}$ 
		\item [(3)] for all $i\in \{1,\ldots,n\}$:
		\begin{enumerate}
			\item [(3.1)]
			$a_i \ = \ p_i \bigl(\frac{1}{2}-w_i)$
			(in particular $w_i < \frac{1}{2}$)
			\item [(3.2)]
			$B_i \ = \ \frac{1}{\delta} B p_i \bigl( \lambda w_i -1)$
			(in particular $w_i > \frac{1}{\lambda}$)
		\end{enumerate}
	\end{enumerate} 
	Assuming such values have been defined, we obtain:
	\begin{eqnarray*}
		\sum_{i\in I} B_i \ \geqslant \ B
		& \ \ \text{iff} \ \ & 
		\frac{1}{\delta} B \sum_{i \in I} p_i (\lambda  w_i -1) \ \geqslant \ B
		\\
		\\[0ex]
		& \text{iff} &
		\sum_{i \in I} p_i (\lambda  w_i -1) \ \geqslant \ \delta
		\\
		\\[0ex]
		& \text{iff} &
		\lambda \sum_{i\in I} p_i w_i
		\ \geqslant \ \delta + \sum_{i\in I} p_i
	\end{eqnarray*}
	Hence:
	\begin{eqnarray*}
		\sum_{i\in I} B_i \ \geqslant \ B
		& \ \ \text{iff} \ \ & 
		\frac{\displaystyle \sum\limits_{i\in I} p_i w_i}
		{\displaystyle \delta + \sum\limits_{i\in I} p_i}
		\ \geqslant \ 
		\frac{1}{\lambda}
	\end{eqnarray*}
	For all positive real numbers $x,y,u,v$ with $\frac{x}{y}=\frac{1}{\lambda}$
	we have:
	\[
	\frac{x+u}{y+v} \geqslant \frac{1}{\lambda}
	\ \ \ \ \ \text{iff} \ \ \ \ \
	\frac{u}{v}\geqslant \frac{1}{\lambda}
	\]
	By the constraints for $\lambda$ (see (2)), we have
	$\frac{p_0}{p_0 + \frac{1}{2} - \delta}=\frac{1}{\lambda}$.
	Therefore:
	\begin{eqnarray*}
		\frac{\displaystyle \sum\limits_{i\in I} p_i w_i}
		{\displaystyle \delta + \sum\limits_{i\in I} p_i}
		\ \geqslant \ 
		\frac{1}{\lambda}
		& \ \text{iff} \ &
		\frac{\displaystyle p_0 + \sum\limits_{i\in I} p_i w_i}
		{\displaystyle (p_0 + \frac{1}{2} - \delta)
			+ \delta + \sum\limits_{i\in I} p_i}
		\ \ = \ \   
		\frac{\displaystyle p_0 + \sum\limits_{i\in I} p_i w_i}
		{\displaystyle p_0 + \frac{1}{2} + \sum\limits_{i\in I} p_i}
		\ \geqslant \ 
		\frac{1}{\lambda}    
	\end{eqnarray*}  
	As before let $w_0=1$ and  $I_0=I\cup \{0\}$. Then,
	the above yields:
	\begin{eqnarray*}
		\sum_{i\in I} B_i \ \geqslant \ B
		& \ \ \ \text{iff} \ \ \ & 
		\frac{\displaystyle  \sum\limits_{i\in I_0} p_i w_i}
		{\displaystyle   \frac{1}{2} + \sum\limits_{i\in I_0} p_i}
		\ \geqslant \ 
		\frac{1}{\lambda}    
	\end{eqnarray*}
	As in the NP-hardness proof for GPR-recall and using (3.1):
	\[
	\Pr_{\cM}(\Diamond \Effect) \ = \ \frac{1}{2} \ > \ w_i
	\qquad \text{for $i=1,\ldots,n$}
	\]
	Thus,
	each GPR cause must have the form $C_I=\{s_i : i \in I_0\}$
	for some subset $I$ of $\{1,\ldots,n\}$.
	Moreover:
	\[
	\Pr_{\cM}(\Diamond C_I) \ = \ \sum_{i\in I_0} p_i
	\qquad \text{and} \qquad
	\Pr_{\cM}(\Diamond (C_I \wedge \Diamond \Effect))
	\ = \
	\sum_{i\in I_0} p_i w_i
	\]
	So, the f-score of $C_I$ is:
	\[
	\fscore(C_I) \ \ = \ \
	2 \cdot \frac{\Pr_{\cM}(\Diamond (C_I \wedge \Diamond \Effect))}
	{\Pr_{\cM}(\Diamond \Effect) + \Pr_{\cM}(\Diamond C_I)}
	\ \ = \ \
	2 \cdot \frac{\sum\limits_{i\in I_0} p_i w_i}
	{\frac{1}{2}+\sum\limits_{i\in I_0} p_i}
	\]
	This implies:               
	\begin{eqnarray*}
		\sum_{i\in I} B_i \ \geqslant \ B
		& \ \text{iff} \ &
		\fscore(C_I) \ \geqslant \
		\frac{2}{\lambda}    
	\end{eqnarray*}
	With $p_0=2a$ and using (3.1) and arguments as in the NP-hardness proof
	for GPR-recall, we obtain:
	\begin{eqnarray*}
		\sum_{i\in I} A_i \ < \ A
		& \ \ \text{iff} \ \ &
		\text{$C_I$ is a GPR cause}
	\end{eqnarray*}
	Thus, the constructed Markov chain has a GPR cause with f-score at least
	$\frac{2}{\lambda}$ if and only if the knapsack problem is solvable
	for the input $A,A_1,\ldots,A_n,B,B_1,\ldots,B_n$.
	
	It remains to define the values $p_1,\ldots,p_n,w_1,\ldots,w_n$ and
	$\delta$. (The value of $\lambda$ is then obtained by (2).)
	(3.1) and (3.2) can be rephrased as equations for $w_i$:
	\begin{description}
		\item [\text{\rm (3.1')}] $w_i=\frac{1}{2}-\frac{a_i}{p_i}$
		\item [\text{\rm (3.2')}]
		$w_i=\frac{1}{\lambda}\bigl( \delta \frac{B_i}{Bp_i}+1 \bigr)$
	\end{description}
	This yields an equation for $p_i$:
	\[
	\frac{1}{2}-\frac{a_i}{p_i} \ \ = \ \ 
	\frac{1}{\lambda}\Bigl( \delta \frac{B_i}{Bp_i}+1 \Bigr)
	\]
	and leads to:
	\begin{align}
		\label{equation-for-pi}
		p_i \ \ = \ \
		\frac{2\lambda}{\lambda-2} a_i \ + \
		\frac{2\delta}{\lambda-2} \cdot \frac{B_i}{B}
	\end{align}
	We now substitute $\lambda$ by (2) and arrive at
	\[
	p_i \ \ = \ \ 
	\frac{p_0}{\frac{1}{2}-\delta}a_i \ + \ a_i + \ \frac{\delta p_0}{\frac{1}{2}-\delta}\frac{B_i}{B}.
	\]
	By choice of $N$, all $a_i$'s and $a$ are smaller than $\frac{1}{8(n{+}1)}$.
	Using this together with $p_0 = 2a$, we get:
	\begin{align}
		\label{five-star}
		p_i \ < \ \frac{1}{4(n{+}1)(\frac{1}{2}-\delta)}\frac{1}{8(n{+}1)} \ + \ \frac{1}{8(n{+}1)} \ + \ \frac{\delta}{4(n{+}1)(\frac{1}{2}-\delta)}\frac{B_i}{B}
	\end{align}
	Let now $\delta =\frac{1}{8K}$ (where $K$ is as above, i.e.,
	the maximum of the values $A,A_1,\ldots,A_n$, $B$, $B_1,\ldots,B_n$).
	Then, $p_1,\ldots,p_n$  are computable
	in polynomial time, and so are the values $w_1,\ldots,w_n$ 
	(by (3.1')).  
	As $\frac{2\lambda}{\lambda-2}>2$ and using \eqref{equation-for-pi},
	we obtain 
	$p_i > 2a_i$.
	So, by (3.1') we get $0 < w_i < \frac{1}{2}$.
	
	It remains to prove (1). 
	Using $\delta=\frac{1}{8K}$, we obtain from \eqref{five-star}:
	\[
	p_i \ < \ \frac{1}{4(n{+}1)(\frac{1}{2}-\frac{1}{8K})}\frac{1}{8(n{+}1)} + \frac{1}{8(n{+}1)}+\frac{1}{32(n{+}1)(\frac{1}{2}-\frac{1}{8K})K}\frac{B_i}{B}  \ \eqdef \ x
	\]
	As $\frac{1}{2}-\frac{1}{8K}\geq \frac{1}{4}$ and $\frac{B_i}{B} < K$, this yields:
	\[
	p_i \ < \ x \ < \ \frac{1}{8(n{+}1)^2}+ \frac{1}{8(n{+}1)}+\frac{1}{8(n{+}1)} \ < \ \frac{1}{2(n{+}1)}.
	\]
	But then condition (1) holds.
\end{proof}

\paragraph*{Arbitrary quality measures}

Consider any algebraic function $f(\mathsf{tp}, \mathsf{tn}, \mathsf{fp}, \mathsf{fn})$.
That is $f$ satisfies some polynomial equation where the coefficients are polynomials in $\mathsf{tp}, \mathsf{tn}, \mathsf{fp}$ and $\mathsf{fn}$.
Almost every quality measure for binary classifiers (see \cite{Powers-fscore}) is such a function.
Taking the worst case scheduler for such a function we define
\[
f(\Cause) = \inf_{\sched} f^\sched(\mathsf{tp}_\Cause^\sched, \mathsf{tn}_\Cause^\sched, \mathsf{fp}_\Cause^\sched, \mathsf{fn}_\Cause^\sched),
\]
where $\sched$ ranges over all schedulers such that $f^\sched$ is well defined.
Given a PR cause $\Cause$ and a rational $\vartheta \in \Rat$, deciding whether $f(\Cause) \leq \vartheta$ can be done in $\PSPACE$ as a satisfiability problem in the existential theory of the reals \cite{ETH-Canny88}.

As we can decide for a given cause candidate $\Cause$ whether it is a SPR cause in $\PTIME$ or a GPR cause in $\coNP$, this also yields an algorithm for finding optimal causes for $f$.
Given an MDP $\cM$ with terminal effect set $\Eff$ and quality measure $f$ as an algebraic function we consider each cause candidate $\Cause$, check whether it is a PR cause (SPR or GPR) and consider the decision problem $f(\Cause) \leq \vartheta$.
As all of these steps have a complexity upper bound of $\PSPACE$ and we only need to save the best cause candidate so far with its value $f(\Cause)$, this results in polynomial space as well.


\section{\texorpdfstring{$\omega$}{ω}-regular effect scenarios}
\label{sec:regular}

In this section, we turn to an extension of the previous definition of PR causes.
So far, we considered both the cause and the effect as sets of states in an MDP $\cM$ with state space $S$. We will refer to this setting as the \emph{state-based} setting from now on.
In a more general approach, we  now consider the effect to be an $\omega$-regular language $\rEff \subseteq S^{\omega}$ over the state space $S$. 
Note that we  denote regular events as effects mainly by $\rEff$ to avoid confusion with effects as sets of states.

In a first step, we  still consider sets of states $\Cause \subseteq S$ as causes, which we call \emph{reachability causes} (Section~\ref{sub:states_as_causes}). 
For reachability GPR causes, the techniques from the previous section are  mostly still applicable. For reachability SPR causes, on the other hand, we observe that they take on the flavor of state-based GPR causes as well.
Afterwards, we  generalize the definition further to allow
$\omega$-regular co-safety properties over the state space $S$ as causes, which we call \emph{co-safety causes} (Section~\ref{sub:co-safety_as_causes}). While this allows us to express much more involved cause-effect relationships, we will see that attempts of checking co-safety SPR causality or of finding good causes for a given effect encounter major new difficulties.

\subsection{Sets of states as causes}
\label{sub:states_as_causes}
Throughout this section, let $\cM=(S,\Act,P,\init)$ be an MDP.
As long as we use sets of states as causes, the definition of GPR and SPR causes can easily be adapted to $\omega$-regular effects:
\begin{defi}[Reachability GPR/SPR causes]
	\label{def:regular-effect-PR-causes} 
	Let $\cM$ be as above. Let $\rEffect\subseteq S^\omega$ be an $\omega$-regular language over $S$ and $\Cause$ a nonempty subset of $S$ such that for each $c\in \Cause$, there is a scheduler $\sched$ with $\Pr^{\sched}_{\cM}( (\neg \Cause) \Until c ) >0$.
	Then, $\Cause$ is said to be a \emph{reachability GPR cause}  for
	$\rEffect$ iff the following condition (rG) holds:
	\begin{description}
		\item [(rG)]
		For each scheduler $\sched$ where $\Pr^{\sched}_{\cM}( \Diamond \Cause) >0$:
		\begin{align*}
			\label{rGPR}  
			\Pr^{\sched}_{\cM}(\rEffect \ | \ \Diamond \Cause)
			\ > \ \Pr^{\sched}_{\cM}(\rEffect).
			\tag{\text{rGPR}}
		\end{align*}
	\end{description}
	$\Cause$ is called a \emph{reachability SPR cause} for $\rEffect$ iff the following condition (rS) holds:
	\begin{description}
		\item [(rS)]
		For each state $c\in \Cause$ and each scheduler $\sched$ where $\Pr^{\sched}_{\cM}( (\neg \Cause) \Until c ) >0$:
		\begin{align*}
			\label{rSPR}  
			\Pr^{\sched}_{\cM}(\rEffect \mid	(\neg \Cause) \Until c)
			\ > \ \Pr^{\sched}_{\cM}(\rEffect).
			\tag{\text{rSPR}}
		\end{align*}   
	\end{description}
\end{defi}
There is one small caveat that we want to mention here: If the effect $\rEff$ is a reachability property $\lozenge \Eff$ for a set of states $\Eff\subseteq S$, 
then this new definition allows for GPR/SPR causes $\Cause$ not disjoint from the set of states $\Eff$.
If two sets $\Cause,\Eff\subseteq S$ are disjoint, however, then $\Cause$ is a GPR/SPR cause for $\Eff$ according to Definition \ref{def:causes} iff
$\Cause$ is a reachability GPR/SPR cause for the $\omega$-regular event $\lozenge\Eff$ according to the new definition.
As we now view the effect as the $\omega$-regular property on infinite executions, one can, nevertheless, argue that
the \emph{temporal priority} \textbf{(C2)} is captured by the new definition since the cause will be reached after finitely many steps if it is reached.
We will address problems with this interpretation and a stronger notion of temporal priority in Section~\ref{sec:finding_rcauses}.

A first simple observation that follows as in the state-based setting is that a reachability SPR cause for $\rEff$ is also a reachability GPR cause for $\rEff$

\subsubsection{Checking causality and existence of reachability PR causes}

To explore how this change of definition influences the previously established results for GPR and SPR causes, we have to clarify how effects will be represented.
We use deterministic Rabin automata (DRAs) as they are expressive enough to capture all $\omega$-regular languages and they are deterministic, which will allow us to form well-behaved products of the automata with MDPs.
Let $\cM$ be an MDP, $\rEff$ an effect given by the DRA $\cA_{\rEff}$ and $\Cause \subseteq S$ a cause candidate.
As a special case we again have Markov chains with no non-deterministic choices.
Then, the conditions (rG) and (rS) can easily be checked by computing the corresponding probabilities in polynomial time (see \cite{TACAS14-condprob} for algorithms to compute conditional probabilities in MCs for path properties).
We now consider the case where non-deterministic choices exist.
We will provide a model transformation of $\cM$ using the DRA such that the resulting MDP has no end components and the effect is a reachability property again.

\begin{nota}
	[Removing end components]
	\label{notation:MDP-mit-min-prob-ab-cause-candidates-regular-effect}  
	Let $\cM$ and $\cA_{\rEff}$ be as above.
	Consider the product MDP $\cN \eqdef \cM \otimes \cA_{\rEff}$.
	This product is an MDP equipped with a Rabin acceptance condition found in the second component of each state in the product.
	
	We now take two copies of $\cN$ representing a mode before $\Cause$ has been reached and one mode after $\Cause$ has been reached.
	So, each state $\mathfrak{s}$ is equipped with one extra bit $0$ or $1$.
	Initially, the MDP starts in the copy labeled with $0$ and behaves like $\cN$ until a state with its first component in $\Cause$ is reached.
	From there, the process moves to the corresponding successor states in the second copy labeled with $1$.
	We call the resulting MDP $\cN^\prime$ and denote the set of states with their first component in $\Cause$ in the first copy that are reachable in $\cN^\prime$ by $\Cause_{\cA_{\rEff}}$, in particular, to express that these states are enriched with states of the automaton $\cA_{\rEff}$.
	
	Next, we consider the MECs $\cE_1,\dots, \cE_k$ of $\cN^\prime$.
	Note that the states in $\Cause_{\cA_{\rEff}}$ are not contained in any MEC.
	Furthermore, all MECs consist either only of states from the first copy labeled $0$, or only of states from the second copy labeled with $1$.
	For each MEC $\cE_i$, we determine whether there is a scheduler for $\cE_i$ that ensures the event $\mathit{Acc}(\cA_{\rEff})$ that the acceptance condition of $\cA_{\rEff}$ is met with probability $1$ and whether there is a scheduler that ensures this event with probability $0$.
	With the techniques of \cite{deAlfaro1997} and \cite{ATVA2009} this can be done in polynomial time.
	We then add four new terminal states $\effcov$, $\noeffc$, $\effunc$, and $\noeffbot$ and construct the MEC-quotient of $\cN^\prime$ while, for each $i\leq k$, enabling a new action  in the state $s_{\cE_i}$ obtained from $\cE_i$ leading to 
	\begin{itemize}
		\item
		$\effunc$ if $\mathit{Acc}(\cA_{\rEff})$ can be ensured with probability $1$ in $\cE_i$ and $\cE_i$ is contained in copy $0$, and
		\item
		another new action leading to $\noeffbot$ if $\mathit{Acc}(\cA_{\rEff})$ can be ensured with probability $0$ in $\cE_i$ and $\cE_i$ is contained in copy $0$;
		\item
		$\effcov$ if $\mathit{Acc}(\cA_{\rEff})$ can be ensured with probability $1$ in $\cE_i$ and $\cE_i$ is contained in copy $1$, and
		\item
		another new action leading to $\noeffc$ if $\mathit{Acc}(\cA_{\rEff})$ can be ensured with probability $0$ in $\cE_i$ and $\cE_i$ is contained in copy $1$.
	\end{itemize}
	Finally, we remove all states which are not reachable (from the initial state).
	We call the resulting MDP $\wminMDP{\cM}{\rEff,\Cause}$ and emphasize that this MDP contains all states in $\Cause_{\cA_{\rEff}}$, has no end components, and has the four terminal states $\effcov$, $\noeffc$, $\effunc$, and $\noeffbot$. Furthermore, for each $c\in \Cause$, we denote the subset of states in $\Cause_{\cA_{\rEff}}$ with $c$ in their first component by $c_{\cA_{\rEff}}$.
	\Ende
\end{nota}

\begin{lem}
	\label{lem:regular-effect-wmin}
	Let $\cM$, $\Cause$, and $\cA_{\rEff}$ be as above and let $\wminMDP{\cM}{\rEff,\Cause}$ be the constructed MDP in Notation \ref{notation:MDP-mit-min-prob-ab-cause-candidates-regular-effect} which contains the set $\Cause_{\cA_{\rEff}}$.
	Then, $\Cause$ is a reachability GPR cause for $\rEff$ in $\cM$ if and only if $\Cause_{\cA_{\rEff}}$ is a GPR cause for $\{\effcov,\effunc\}$ in $\wminMDP{\cM}{\rEff,\Cause}$.
	Furthermore, $\Cause$ is a reachability SPR cause for $\rEff$ in $\cM$ if and only if, for each $c\in \Cause$, the set  $c_{\cA_{\rEff}}$ of states in $\Cause_{\cA_{\rEff}}$ with $c$ in their first component is a GPR cause for $\{\effcov,\effunc\}$ in $\wminMDP{\cM}{\rEff,\Cause}$.
\end{lem}

\begin{proof}
	The set $\Cause_{\cA_{\rEff}}$ in $\wminMDP{\cM}{\rEff,\rCause}$ satisfies $\Pr^{\max}_{\wminMDP{\cM}{\rEff,\rCause}}(\neg \Cause_{\cA_{\rEff}} \Until d)>0$ for each state $d\in \Cause_{\cA_{\rEff}}$ by construction, since all states in $\Cause_{\cA_{\rEff}}$ are reachable and a run cannot reach two different states in $\Cause_{\cA_{\rEff}}$.
	So, this minimality condition is satisfied in any case and, of course, $\{\effcov,\effunc\}$  and $\Cause_{\cA_{\rEff}}$ are disjoint.
	
	Now, let $\tsched$ be a scheduler for $\wminMDP{\cM}{\rEff,\Cause}$.
	The scheduler $\tsched$ can be mimicked by a scheduler $\sched$ for $\cM$: 
	As long as $\tsched$ moves through the MEC-quotient of $\cN^\prime$, the scheduler $\sched$ follows this behavior by leaving MECs through the corresponding actions in $\cM$.
	Whenever $\tsched$ moves to one of the states $\effcov$ and $\effunc$, the last step begins in a state $s_{\cE}$ obtained from a MEC $\cE$ of $\cM$.
	In this case, $\sched$ will stay in $\cE$ while ensuring with probability $1$ that the resulting run is accepted by $\cA_{\rEff}$.
	Similarly, if $\tsched$ moves to $\noeffbot$ or $\noeffc$, the scheduler $\sched$ for $\cM$ stays in the corresponding MEC and realizes the acceptance condition of $\cA_{\rEff}$ with probability $0$.
	Vice versa, a scheduler $\sched$ for $\cM$ can be mimicked by a scheduler $\tsched$ for $\wminMDP{\cM}{\rEff,\Cause}$ analogously.
	So in an end component $\cE$, if $\sched$ stays in $\cE$ and ensures that the resulting run is accepted by $\cA_{\rEff}$ with probability $p_1$, 
	stays in $\cE$ and ensures that the resulting run is not accepted by $\cA_{\rEff}$ with probability $p_2$,
	and leaves $\cE$ with probability $p_3$,
	$\tsched$ will move to the corresponding state $\effcov$ or $\effunc$  with probability $p_1$, to $\noeffbot$ or $\noeffc$ with probability $p_2$, and take actions leaving $s_{\cE}$ to other states  with the same probability distribution with which $\sched$ takes the leaving actions of $\cE$.
	
	For such a pair of schedulers $\sched$ and $\tsched$, we observe that 
	\begin{align}
		\label{eqn:schedulers1}
		\Pr^{\sched}_{\cM}(\lozenge \Cause)& = \Pr^{\tsched}_{\wminMDP{\cM}{\rEff,\Cause}}(\lozenge \Cause_{\cA_{\rEff}}), \\
		\label{eqn:schedulers2}
		\Pr^{\sched}_{\cM}(\rEff)& = \Pr^{\tsched}_{\wminMDP{\cM}{\rEff,\Cause}}(\lozenge E), \text{ and} \\
		\label{eqn:schedulers3}
		\Pr^{\sched}_{\cM}(\lozenge \Cause \wedge \rEff)& = \Pr^{\tsched}_{\wminMDP{\cM}{\rEff,\Cause}}(\lozenge \effcov) = \Pr^{\tsched}_{\wminMDP{\cM}{\rEff,\Cause}}(\lozenge \Cause_{\cA_{\rEff}} \wedge \lozenge E),
	\end{align}
	where $E = \{\effcov, \effunc\}$.
	Hence, $\Cause_{\cA_{\rEff}}$ is a GPR cause for $\{\effcov, \effunc\}$ in $\wminMDP{\cM}{\rEff,\Cause}$ if and only if $\Cause$ is a reachability GPR cause for $\rEff$ in $\cM$.
	
	Now consider each element $c \in \Cause$ individually.
	We can use the same argumentation to see that a scheduler $\tsched$ in $\wminMDP{\cM}{\rEff,\Cause}$  can be mimicked by a scheduler $\sched$ for $\cM$, and vice versa,
	such that
	\begin{align*}
		\Pr^\sched_\cM( \neg \Cause \until c)&= \Pr^\tsched_{\wminMDP{\cM}{\rEff,\Cause}}(\lozenge c_{\cA_{\rEff}}), \text{ and} \\
		\Pr^\sched_\cM(\rEff \mid \neg \Cause \until c)& = \Pr^\tsched_{\wminMDP{\cM}{\rEff,\Cause}}(\lozenge \effcov \mid \lozenge c_{\cA_{\rEff}}).
	\end{align*}
	Thus,   $\Cause$ is a reachability SPR cause for $\rEff$ in $\cM$ if and only if $c_{\cA_{\rEff}}$ is a GPR cause for $\{\effcov,\effunc\}$ in $\wminMDP{\cM}{\rEff,\Cause}$ for all $c\in \Cause$.
\end{proof}
This Lemma \ref{lem:regular-effect-wmin} shows that reachability SPR causality shares similarities with GPR causality. 
Our algorithmic results for reachability PR causes stem from the reduction provided by Notation \ref{notation:MDP-mit-min-prob-ab-cause-candidates-regular-effect} and Lemma \ref{lem:regular-effect-wmin}.
As an immediate consequence we can check conditions (rG) and (rS) in $\wminMDP{\cM}{\rEff, \Cause}$ by using the already established algorithms for GPR causes.
This results in the following complexity upper bounds.

\begin{cor}
	\label{cor:complexities-check-regular}
	Let $\cM$ be an MDP and $\rEff \subseteq S^\omega$ an $\omega$-regular language given as DRA $\cA_{\rEff}$.
	Given a set $\Cause \subseteq S$ we can decide
	whether $\Cause$ is a reachability SPR/GPR cause for $\rEff$ in $\coNP$.
\end{cor}
\begin{proof}
	By Lemma \ref{lem:regular-effect-wmin}, we can use the construction of $\wminMDP{\cM}{\rEff, \Cause}$, which takes polynomial time.
	Then, Theorem \ref{thm:checking-GPR-in-poly-space} allows us to check whether $\Cause$ is a reachability GPR cause in $\coNP$ directly, while we can apply the GPR check to each set $c_{\cA_{\rEff}}$ and the effect $\{\effcov,\effunc\}$ in $\wminMDP{\cM}{\rEff, \Cause}$ for $c\in \Cause$ in order to determine whether $\Cause$ is a reachability SPR $\Cause $ for $\rEff$.
\end{proof}

As in the state-based setting, we can argue that there is a reachability GPR cause iff there is a reachability SPR cause iff there is a singleton reachability SPR cause.
Consequently, the existence of a reachability GPR/SPR cause can be checked by checking for each state $c$ of the MDP whether $\{c\}$ constitutes a reachability SPR cause. 
We conclude:

\begin{cor}
	The existence of a reachability GPR/SPR cause can be decided in $\coNP$.
\end{cor}

\subsubsection{Computing quality measures of reachability PR causes}
\label{sec:quality_reachability_causes}

As in Section~\ref{sec:acc-measures}, we can view reachability PR causes as binary classifiers.
This leads to the confusion matrix as before (Figure \ref{fig: confusion matrix}) with the difference that this time the path does not ``hit'' the effect set, but rather $\rEff$ holds on the path.
Hence, we define the following entries of the confusion matrix:
Given $\rEff$, $\Cause$ and a scheduler $\sched$, we let
\[
\begin{array}{rclrcl}
	\mathsf{tp}^\sched & \eqdef & \Pr_\cM^\sched(\lozenge \Cause \wedge \rEff), & \mathsf{tn}^\sched & \eqdef & \Pr_\cM^\sched(\neg \lozenge \Cause \wedge \neg \rEff),
	\\[1ex]
	\mathsf{fp}^\sched & \eqdef & \Pr_\cM^\sched(\lozenge \Cause \wedge \neg \rEff), &\mathsf{fn}^\sched & \eqdef & \Pr_\cM^\sched(\neg \lozenge \Cause \wedge \rEff).
\end{array}
\]

\begin{lem}
	\label{lem:preserve_confusion}
	Let $\cM$ be as above, $\rEff$ an $\omega$-regular effect given by the DRA $\cA_{\rEff}$ and let $\Cause\subseteq S$ be a reachability GPR/SPR cause. Further, let 
	$\wminMDP{\cM}{\rEff, \Cause}$ be as constructed above with the set of cause states $\Cause_{\cA_\rEff}$ and the set of effect states $\{\effcov,\effunc\}$.
	Then, for each scheduler $\sched$ for $\cM$, there is a scheduler $\tsched$ for $\wminMDP{\cM}{\rEff, \Cause}$, and vice versa, such that 
	\[
	\mathsf{tp}^\sched = \mathsf{tp}^\tsched, \qquad \mathsf{tn}^\sched = \mathsf{tn}^\tsched, \qquad
	\mathsf{fp}^\sched = \mathsf{fp}^\tsched, \qquad \mathsf{fn}^\sched = \mathsf{fn}^\tsched,
	\]
	where the values for $\tsched$ are defined in $\wminMDP{\cM}{\rEff, \Cause}$ as for the state-based setting (cf. Section~\ref{sec:acc-measures}).
\end{lem}

\begin{proof}
	In the proof of Lemma \ref{lem:regular-effect-wmin}, we have seen that for each scheduler $\sched$ for $\cM$, we can find a scheduler $\tsched$ for $\wminMDP{\cM}{\rEff, \Cause}$ and vice versa
	such that Equations (\ref{eqn:schedulers1}) - (\ref{eqn:schedulers3}) hold. This implies the equalities claimed here.
\end{proof}

Analogously to Section~\ref{sec:acc-measures}, we can now define $\relcov$, $\covratio$, and $\fscore$ of a reachability PR cause as the infimum of these values in terms of 
$\mathsf{tp}^\sched$, $\mathsf{fp}^\sched$, $\mathsf{tn}^\sched$, and $\mathsf{fn}^\sched$ over all schedulers $\sched$ for which the respective quality measures are defined.
The computation of these values can then be done with the methods from the state-based setting:

\begin{cor}
	\label{cor:reachability-comp-quality}
	Let $\cM$ be an MDP and $\rEff \subseteq S^\omega$ an $\omega$-regular language given as DRA $\cA_{\rEff}$.
	Given a reachability SPR/GPR cause $\Cause \subseteq S$ we can compute $\recall(\Cause), \covrat(\Cause)$ and $\fscore(\Cause)$ in polynomial time.
\end{cor}
\begin{proof}
	Use Lemma \ref{lem:preserve_confusion}, Corollary \ref{cor:covratio-in-PTIME} and Theorem \ref{fscore-in-PTIME}.
\end{proof}

\subsubsection{Finding quality optimal reachability PR causes}
\label{sec:finding_rcauses}

When trying to find good causes for an $\omega$-regular effect $\rEff$, we cannot say that effects and causes should be disjoint as in the setting where effects and causes were sets of states. 
This leads to the possibility that causes might exist that do not capture the intuition behind temporal priority: E.g., if the effect $\rEff$ is a reachability property $\lozenge E$ for a set of states $E$, the set of states $E$ itself will be a reachability PR cause if for each state $c\in E$, $\Pr^{\max}_{\cM}(\neg E \Until c)>0$ holds. Furthermore, there might be causes $C$ that can only be reached after $E$ has already been reached.

In order to account for the temporal priority of causes, we will include the following condition when trying to find good causes: 
We require for a cause $\Cause\subseteq S$ that
\begin{align}
	\Pr^{\min}_{\cM}(\rEff\mid \neg \Cause \Until c)<1 \text{ for all $c \in  \Cause$.} \label{eq:TempPrio}\tag{TempPrio}
\end{align}
Intuitively, this states that it is never already certain that the execution will belong to $\rEff$ when the cause is reached.
A variation of this criterium has been proposed also in \cite{ICALP21}.

\begin{rem}
	\label{rem:tempprio}
	The condition \eqref{eq:TempPrio} could be added to the definition of reachability PR causes. After the product construction in Notation \ref{notation:MDP-mit-min-prob-ab-cause-candidates-regular-effect}, the condition can easily be checked for a given cause candidate $\Cause\subseteq S$: For each $c\in \Cause$, there must be at least one state with $c$ in the first component in $\wminMDP{\cM}{\rEff, \Cause}$ such that $\noeffc$ is reachable from this state.
	
	Furthermore, this condition is stronger than the requirement that causes and effects are disjoint in the state-based setting. In the state-based setting, however, the analogue of condition \eqref{eq:TempPrio} could also be included easily. Instead of having to be disjoint from a set of effect states $\Eff$, a cause $\Cause$ would then simply not be allowed to contain any state $s$ with $\Pr^{\min}_{\cM,s}(\Diamond \Eff)=1$. \Ende
\end{rem}

Now, we want to find recall- and coverage ratio-optimal reachability SPR causes.
As in the state-based setting, we define the set $\cC$ of all possible singleton reachability SPR causes for $\rEff$
that in addition satisfy \eqref{eq:TempPrio} as explained in Remark \ref{rem:tempprio}.
By Corollary \ref{cor:complexities-check-regular}, we can check whether a state $c\in S$ is a singleton reachability SPR cause in coNP; whether there is at least one state with $c$ in the first component in $\wminMDP{\cM}{\rEff, \Cause}$ such that $\noeffc$ is reachable from this state is checkable in polynomial time.
Thus, we can again define the set of singleton reachability SPR causes 
\[\cC = \{s \in S \mid s \text{ is reachability SPR cause satisfying } \Pr^{\min}_{\cM}(\rEff\mid \neg \Cause \Until c)<1 \}.\]
As before, the canonical cause $\CanCau$ is now the set of states $c \in \cC$ for which there is a scheduler $\sched$ with $\Pr^\sched(\neg \cC \until c)$.

For the complexity of the computation of the recall- and coverage ratio-optimal canonical cause and its values, the observations above lead us to the complexity class
$\PFNP$ as defined in \cite{Selman1994}. It consists of all functions that can be computed in polynomial time with access to an $\NP$-oracle, or equivalently a $\coNP$-oracle.

\begin{prop}
	If $\cC \neq \emptyset$ then $\CanCau$ is a ratio- and recall optimal reachability SPR cause for $\rEff$.
	The threshold problem for the coverage ratio and the recall can be decided in $\coNP$.
	The optimal values $\recall(\CanCau)$ and $\covrat(\CanCau)$ can be computed in $\PFNP$.
\end{prop}
\begin{proof}
	The optimality of $\CanCau$ follows by the arguments used for Theorem \ref{thm:optimality-of-canonical-SPR-cause}.
	In order to compute $\CanCau$ we check for each state $s$ whether (rS) does not hold in $\NP$ and then take the remaining states $c \in S$ to define $\cC$ by checking \eqref{eq:TempPrio} for each $c$.
	If \eqref{eq:TempPrio} holds then $c \in \cC$ and $\CanCau = \{c \in \cC \mid \Pr^{\max}_\cM(\neg \cC \until c) > 0\}$.
	This allows to compute the recall- and coverage ratio-optimal cause $\CanCau$ in $\PFNP$.
	
	For the threshold problem whether there is a reachability SPR cause with recall at least a given $\vartheta \in \Rat$, the $\coNP$ upper bound can be shown as follows:
	For each state $c$ that does not belong to $\cC$, i.e., that is not a singleton reachability SPR cause, there is a polynomial size certificate for this, as it can be checked in $\coNP$.
	The collection of all states that do not belong to $\cC$ together with these certificates (they do not belong to $\cC$) can now serve as a certificate that the optimal recall is less than $\vartheta$ in this the case.
	Given such a collection of certificates, we can check in polynomial time that indeed all provided states do not belong to $\cC$.
	The complement of the provided states forms a super set $\cD$ of  $\CanCau$.
	Computing the recall of the set $\cD$ can then be done in polynomial time as in Corollary \ref{cor:reachability-comp-quality}. This value is an upper bound for the recall of $\CanCau$. Note that if all states that do not belong to  $\cC$ are given in the certificate, the value even equals the recall of $\CanCau$. 
	
	So, if the optimal value is less than $\vartheta$, this is witnessed by the described certificate containing all states not in $\cC$.
	Vice versa, if a certificate is given resulting in a super set $\cD$ of $\CanCau$ such that the recall of $\cD$ is less than $\vartheta$, then there is no reachability SPR cause with a recall of at least $\vartheta$. So, the threshold problem lies in $\coNP$. For the coverage ratio, the analogous argument works.
\end{proof}

For the threshold problems for f-score-optimal reachability SPR causes and   reachability GPR causes optimal with respect to recall, coverage ration or f-score that satisfy \eqref{eq:TempPrio},
we rely on the guess-and-check approach used for optimal GPR causes in the state-based setting.
We guess a subset $\Cause$ of states of the MDP, check whether we found a reachability SPR/GPR cause in $\coNP$, and compute the quality measure under consideration in polynomial time as explained in the previous section.
For the computation of an optimal cause, we obtain a polynomial-space algorithm.
\begin{cor}
	\label{cor:complexities-regular-quality-GPR-threshold}
	Let $\cM$ be an MDP and $\rEff$ be an $\omega$-regular language given by $\cA_{\rEff}$.
	Given $\vartheta \in \Rat$, deciding whether there exists a reachability GPR cause $\Cause$ with $\recall(\Cause) \geq \vartheta$ (resp. $\covrat(\Cause) \geq \vartheta, \fscore(\Cause) \geq \vartheta$) can be done in $\Sigma^P_2$ and is $\NP$-hard.
	$\NP$-hardness even holds for Markov chains.
	
	Deciding whether there is a reachability SPR cause  $\Cause$ with $\fscore(\Cause) \geq \vartheta$ can be done in $\Sigma^P_2$.
	A recall-, covratio-, or f-score-optimal reachability GPR cause as well as an f-score-optimal reachability SPR cause can be computed in polynomial space.
\end{cor}
\begin{proof}
	Obviously the lower bounds extend to this setting as we can interpret GPR causes for $\Eff$ as reachability GPR causes for $\lozenge \Eff$.
	The upper bounds extend to this setting by using the construction from Notation \ref{notation:MDP-mit-min-prob-ab-cause-candidates-regular-effect}.
	Since for Theorem \ref{thm:GPR-recall-NPhard-and-in-PSPACE} we relied on guess-and-check algorithms to establish the upper bounds for the threshold problems, we can use analogous algorithms in setting of $\omega$-regular effects.
	We guess a set $\Cause \subseteq S$, check the reachability GPR causality in $\coNP$ (Corollary \ref{cor:complexities-check-regular}) and compute the value of the quality measure in polynomial time (Corollary \ref{cor:reachability-comp-quality}).
	Again, the alternation between the existential quantification for guessing $\Cause$ and the universal quantification for the $\coNP$ check results in the complexity $\Sigma^P_2$ of the polynomial-time hierarchy.
	
	In order to show that the decision problem for $\fscore(\Cause) \geq \vartheta$ for SPR causes $\Cause$ is in $\Sigma^P_2$ we resort to the constructed MDP $\wminMDP{\cM}{\rEff, \Cause}$.
	By Lemma \ref{lem:regular-effect-wmin} a reachability SPR cause $\Cause$ in $\cM$ corresponds to a set of GPR causes $\{c_{\cA_{\rEff}} \mid c \in \Cause\}$ in $\wminMDP{\cM}{\rEff, \Cause}$, which can be interpreted as GPR cause $C = \bigcup_{c \in \Cause} c_{\cA_{\rEff}}$.
	This way we can encode the property $\fscore(\Cause) \geq \vartheta$ in $\cM$ by $\fscore(C) \geq \vartheta$ in $\wminMDP{\cM}{\rEff, \Cause}$.
	This results in a decision problem GPR-f-score for GPR causes which is in $\Sigma^P_2$ by Theorem \ref{thm:GPR-recall-NPhard-and-in-PSPACE}.
	
	For computing optimal GPR causes as well as f-score-optimal SPR causes we can try all cause candidates by computing the related value ($\recall, \covrat$ or $\fscore$) and always store the best one so far.
	As the space for the cause can be reused, this results in a polynomial space algorithm.
\end{proof}

\subsection{\texorpdfstring{$\omega$}{ω}-regular co-safety properties as causes}
\label{sub:co-safety_as_causes}

We now want to discuss an extension of the previous framework when we also consider causes to be regular sets of executions.
However, in order to account for the \emph{temporal priority} of causes, i.e., the fact that causes should occur before their effects, it makes sense to restrict causes to $\omega$-regular co-safety properties. 
The reason is that an $\omega$-regular co-safety property $\mathcal{L}$ is uniquely determined by the regular set of minimal \emph{good} prefixes of words in $\mathcal{L}$.
Recall that a good prefix $\pi$ for $\mathcal{L}$ is a finite word such that all infinite extensions of $\pi$ belong to $\mathcal{L}$ and that  all infinite words in the co-safety language $\mathcal{L}$ have a good prefix.
Hence, we can say that a cause $\rCause$ occurred as soon as a {good} prefix for $\rCause$ has been produced.
For this subsection we will denote regular effects and causes mainly by $\rEff$ and $\rCause$ to avoid confusion with effects and causes as sets of states.
In the following formal definition, we use finite words $\sigma\in S^\ast$ to denote the event $\sigma S^\omega$.
\begin{defi}[co-safety GPR/SPR causes]
	\label{def: regular PR causes}
	Let $\cM$ be an MDP with state space $S$ and let $\rEff \subseteq S^\omega$ be an $\omega$-regular language.
	An $\omega$-regular co-safety language $\rCause \subseteq S^\omega$ is a \emph{co-safety GPR cause} for $\rEff$ if the following condition (coG) holds:
	\begin{description}
		\item [(coG)]
		For each scheduler $\sched$ where
		$\Pr^{\sched}_{\cM}(\rCause) >0$:
		\begin{align*}
			\label{cosafeGPR}  
			\Pr^{\sched}_{\cM}(\rEffect \mid \rCause)
			\ > \ \Pr^{\sched}_{\cM}(\rEffect).
			\tag{\text{cosafeGPR}}
		\end{align*}
	\end{description}
	The event $\rCause$ is called a \emph{co-safety SPR cause} for
	$\rEffect$ if the following condition (coS) hold:
	\begin{description}
		\item [(coS)]
		For each minimal good prefix $\sigma$ for $\rCause$ and each scheduler $\sched$	where $\Pr^{\sched}_{\cM}(  \sigma) > 0$:
		\begin{align*}
			\label{cosafeSPR}  
			\Pr^{\sched}_{\cM}(\rEffect \mid	 \sigma)
			\ > \ \Pr^{\sched}_{\cM}(\rEffect).
			\tag{\text{cosafeSPR}}
		\end{align*}   
	\end{description}
\end{defi}

As in the state-based setting it follows that co-safety SPR cause are also co-safety GPR causes.

\subsubsection{Checking co-safety causality}

We will represent co-safety PR causes as DFAs which accept good prefixes of the represented $\omega$-regular event.
Note that, for any $\omega$-regular co-safety property, there is a DFA accepting exactly the minimal good prefixes.
So, we will restrict to such DFAs that accept the minimal good prefixes of an $\omega$-regular co-safety property.
Such a DFA can never accept a word $w$ as well as a proper prefix $v$ of $w$.

Let now $\cM$ be an MDP, $\rEff$ an effect given by the DRA $\cA_{\rEff}$ and $\rCause$ a cause candidate given by a DFA $\cA_{\rCause}$ as above. 
So, in particular, $\cA_{\rCause}$ accepts exactly the minimal good prefixes for $\rCause$.
We now want to check, whether $\rCause$ is a co-safety SPR cause (resp. co-safety GPR cause). 
For the special case of Markov chains the check can be done in polynomial time analogously to reachability PR causes by computing the corresponding conditional probabilities.
We can provide a model transformation of $\cM$ using both automata such that the resulting MDP has no end components and the effect is a reachability property again similar to Notation \ref{notation:MDP-mit-min-prob-ab-cause-candidates-regular-effect}.

For this consider the product $\cN \eqdef \cM\otimes \cA_{\rEff} \otimes \cA_{\rCause}$.
This product is an MDP equipped with two kinds of acceptance conditions.
The Rabin acceptance of $\cA_{\rEff}$ in the second component of each state and the acceptance condition of $\cA_{\rCause}$ in the third component.
Now let $\rCause_{\cA_{\rEff}}$ be the set of all states of $\cN$ whose third component is accepting in $\cA_{\rCause}$ and which are reachable from the initial state.

As in Notation \ref{notation:MDP-mit-min-prob-ab-cause-candidates-regular-effect}, we construct an MDP $\cN^\prime$ by introducing a mode before $\rCause_{\cA_{\rEff}}$ and a mode after $\rCause_{\cA_{\rEff}}$. We then take the MEC-quotient with the four terminal states $\effcov$, $\noeffc$, $\effunc$, and $\noeffbot$, which are reachable from states $s_{\cE}$ that result from collapsing the MEC $\cE$ depending on whether $\cE$ is contained in the before- or after-$\rCause_{\cA_{\rEff}}$ mode and whether the acceptance condition of $\cA_{\rEff}$ can be realized with probability $0$ and $1$, respectively in $\cE$, analogously to Notation \ref{notation:MDP-mit-min-prob-ab-cause-candidates-regular-effect}.
We call the resulting MDP $\wminMDP{\cM}{\rEff,\rCause}$ and emphasize that this MDP still contains all states in $\rCause_{\cA_{\rEff}}$ as they are not contained in any end component.

We start with the observation, that for co-safety GPR causes this reduction characterizes the condition (coG) completely.
\begin{lem}
	\label{lem:regular-wmin-GPR}
	Let $\cM$ be and MDP, $\cA_{\rEff}$ an DRA, and $\cA_{\rCause}$ a DFA be as above and let $\wminMDP{\cM}{\rEff,\rCause}$ be the constructed MDP that contains the set $\rCause_{\cA_{\rEff}}$  of reachable states that have an accepting $\cA_{\rCause}$-component.
	Then, $\rCause$ is a co-safety GPR cause for $\rEff$ in $\cM$ if and only if the set of states $\rCause_{\cA_{\rEff}}$  is a GPR cause for $\{\effcov,\effunc\}$ in $\wminMDP{\cM}{\rEff,\rCause}$.
\end{lem}

\begin{proof}
	The set $\rCause_{\cA_{\rEff}}$ in $\wminMDP{\cM}{\rEff,\rCause}$ satisfies $\Pr^{\max}_{\wminMDP{\cM}{\rEff,\rCause}}(\neg \rCause_{\cA_{\rEff}} \Until c)>0$ for each $c\in \rCause_{\cA_{\rEff}}$ by construction, since all states in $\rCause_{\cA_{\rEff}}$  are reachable and a run cannot reach two different states in $\rCause_{\cA_{\rEff}}$.
	Thus, the minimality condition is satisfied in any case.
	
	Now, let $\sched$ be a scheduler for $\wminMDP{\cM}{\rEff,\rCause}$.
	The scheduler $\sched$ can be mimicked by a scheduler $\tsched$ for $\cM$: 
	As long as $\sched$ moves through the MEC-quotient of $\cN^\prime$, the scheduler $\tsched$ follows this behavior by leaving MECs through the corresponding actions in $\cM$.
	Whenever $\sched$ moves to one of the states $\effcov$ and $\effunc$, the last step begins in a state $s_{\cE}$ obtained from a MEC $\cE$ of $\cM$. In this case, $\tsched$ will stay in $\cE$ while ensuring with probability $1$ that the resulting run is accepted by $\cA_{\rEff}$.
	Similarly, if $\sched$ moves to $\noeffbot$ or $\noeffc$, the scheduler $\tsched$ for $\cM$ stays in the corresponding MEC and realizes the acceptance condition of $\cA_{\rEff}$ with probability $0$. 
	Vice versa, a scheduler $\tsched$ for $\wminMDP{\cM}{\rEff,\rCause}$ can be mimicked by a scheduler $\sched$ for $\cM$ analogously (see also the proof of Lemma \ref{lem:regular-effect-wmin}).
	
	For such a pair of schedulers $\sched$ and $\tsched$, we observe that 
	\begin{align}
		\label{eqn:GPRschedulers1}
		\Pr^{\sched}_{\cM}(\rCause)&=\Pr^{\tsched}_{\wminMDP{\cM}{\rEff,\rCause}}(\lozenge \rCause_{\cA_{\rEff}}), \\
		\label{eqn:GPRschedulers2}
		\Pr^{\sched}_{\cM}(\rEff)&=\Pr^{\tsched}_{\wminMDP{\cM}{\rEff,\rCause}}(\lozenge E), \text{ and} \\
		\label{eqn:GPRschedulers3}
		\Pr^{\sched}_{\cM}(\rCause\land \rEff) &=\Pr^{\tsched}_{\wminMDP{\cM}{\rEff,\rCause}}(\lozenge \effcov)=\Pr^{\tsched}_{\wminMDP{\cM}{\rEff,\rCause}}(\lozenge \rCause_{\cA_{\rEff}} \land \lozenge E),
	\end{align}
	where $E = \{\effcov, \effunc\}$.
	Hence, $\rCause_{\cA_{\rEff}}$ is a GPR cause for $E$ in $\wminMDP{\cM}{\rEff,\rCause}$ if and only if $\rCause$ is a regular GPR cause for $\rEff$ in $\cM$.
\end{proof}

For co-safety SPR causes this is not a full characterization but only holds in one direction:
\begin{lem}
	\label{lem:regular-wmin-SPR}
	Let $\cM$, $\cA_{\rEff}$, and $\cA_{\rCause}$ be as above and let $\wminMDP{\cM}{\rEff,\rCause}$ be the constructed MDP that contains the set $\rCause_{\cA_{\rEff}}$ of reachable states that have an accepting $\cA_{\rCause}$-component.
	If $\rCause$ is a co-safety SPR cause for $\rEff$ in $\cM$, then the set of states $\rCause_{\cA_{\rEff}}$ is an SPR cause for $\{\effcov,\effunc\}$ in $\wminMDP{\cM}{\rEff,\rCause}$.
\end{lem}

\begin{proof}
	Suppose $\rCause$ is a regular SPR cause for $\rEff$ in $\cM$.
	Let $c\in \rCause_{\cA_{\rEff}}$ and let $\tsched$ be a scheduler for $\wminMDP{\cM}{\rEff,\rCause}$ such that $\Pr^{\tsched}_{\wminMDP{\cM}{\rEff,\rCause}}(\neg \rCause_{\cA_{\rEff}} \Until c)>0$.
	Let
	\[
	\Pi=\{\pi  \text{ a  path in $\wminMDP{\cM}{\rEff,\rCause}$} \mid \pi \vDash\neg \rCause_{\cA_{\rEff}} \Until c \text{ and  $\Pr^{\tsched}_{\wminMDP{\cM}{\rEff,\rCause}}(\pi)>0$}\}.
	\]
	Let $\sched$ be a scheduler for $\cM$ mimicking $\tsched$ as described in the proof of Lemma \ref{lem:regular-wmin-GPR}.
	Let $\Sigma_\pi$ be the set of $\tsched$-paths in $\cM$ that correspond to the path $\pi$ in $\Pi$.
	I.e., a path $\sigma$ belongs to $\Sigma_\pi$ if it moves through the MECs of $\cM$ in the same way as the path $\pi$ moves through the MEC-quotient until $\sigma$ and $\pi$ reach the state $c$ and if furthermore, $\sigma$ has positive probability under $\sched$.
	Now,
	\[
	\Pr^{\tsched}_{\wminMDP{\cM}{\rEff,\rCause}}(\lozenge E\mid \neg \rCause_{\cA_{\rEff}} \Until c) = \frac{\sum_{\pi\in \Pi}\sum_{\sigma \in \Sigma_\pi} \Pr^{\sched}_{\cM}(\rEff\mid \sigma)\cdot \Pr^{\sched}_{\cM}(\sigma)}{\sum_{\pi\in \Pi}\sum_{\sigma \in \Sigma_\pi} \Pr^{\sched}_{\cM}(\sigma)},
	\]
	where $E = \{\effcov, \effunc\}$.
	All the terms $\Pr^{\sched}_{\cM}(\rEff\mid \sigma)$ are greater than the value  of the term $\Pr^{\sched}_{\cM}(\rEff)=\Pr^{\tsched}_{\wminMDP{\cM}{\rEff,\rCause}}(\lozenge E)$.
	Hence, we conclude that 
	\[
	\Pr^{\tsched}_{\wminMDP{\cM}{\rEff,\rCause}}(\lozenge E\mid \neg \rCause_{\cA_{\rEff}} \Until c) > \Pr^{\tsched}_{\wminMDP{\cM}{\rEff,\rCause}}(\lozenge E). \qedhere
	\]
\end{proof}

For checking co-safety SPR causality this is not sufficient.
The underlying problem, in which co-safety SPR causes and SPR causes differ, can be seen in the following example:

\begin{figure}[t]
		\centering
		\resizebox{.3\textwidth}{!}{
			\begin{tikzpicture}
	[scale=1,->,>=stealth',auto ,node distance=0.5cm, thick]
	
	\node[scale=1, state] (init) {$\init$};
	\node[scale=1, state, below=1 of init] (b) {$b$};
	\node[scale=1, state, right=1 of init] (d) {$d$};
	\node[scale=1, state, right=1 of b] (c) {$c$};
	\node[scale=1, state, left=1 of b] (t) {$t$};
	\node[scale=1, state, below=2 of t] (noeff) {$\noeff$};
	\node[scale=1, state, accepting, below=2 of c] (eff) {$\eff$};
	
	\draw[<-] (init) --++(-0.55,0.55);
	\draw[color=black,->] (init) edge node[pos=0.5, left] {$1/2$} (t);
	\draw[color=black,->] (init) edge node [pos=0.5] {$1/4$} (b);
	\draw[color=black,->] (init) edge node [pos=0.5] {$1/4$} (d);
	\draw[color=black,->] (b) edge (c);
	\draw[color=black,->] (d) edge (c);
	\draw[color=black,->] (t) edge node [pos=0.3] {$3/4$} (noeff);
	\draw[color=black,->] (t) edge node [pos=0.3, above] {$1/4$} (eff);
	\draw[color=black,->] (c) edge node [pos=0.5, right] {$\beta$} (eff);
	\draw[color=black,->] (c) edge[out=250, in=110] node[pos=0.3, anchor=center] (b1) {} node [pos=0.5, left] {$1/2$} (eff);
	\draw[color=black,->] (c) edge node[pos=0.2, anchor=center] (b2) {} node [pos=0.4, above] {$1/2$} (noeff);
	
	\draw[color=black, very thick, -] (b1.center) edge [bend left=45] node [pos=0.4, above ] {$\alpha$} (b2.center);
\end{tikzpicture}
		}
		\caption{An MDP $\cM$ for which the reduction from co-safety to state-based PR causes fails}
		\label{fig:no-co-safety-SPR-check}
	\end{figure}
\begin{exa}
	\label{ex:no-co-safety-SPR-check}
	Consider the MDP $\cM$ from Figure \ref{fig:no-co-safety-SPR-check} with $\rEff = \lozenge \eff$.
	For every scheduler $\sched$ with $\Pr^\sched_\cM(\lozenge c) > 0$ we have
	\[
	\Pr^\sched_\cM(\lozenge \eff \mid \lozenge c) > \Pr^\sched_\cM(\lozenge \eff)
	\]
	and thus $c$ is a state-based SPR cause for $\eff$.
	On the other hand for the scheduler $\tau$, which chooses $\alpha$ after the path $\pi = \init \; b \; c$ and $\beta$ otherwise ,we have 
	\[
	\Pr^\tau_\cM(\lozenge \eff \mid \pi) = \frac{1}{2} =  \Pr^\tau_\cM(\lozenge \eff).
	\]
	Therefore, the desired reduction as in Lemma \ref{lem:regular-wmin-GPR} does not work for $\cM$.
	Note that the violation of the condition (coS) is only possible in this example if the scheduler behaves differently depending on how state $c$ is reached. This different behavior, however, does not have anything to do with the effect, and potentially different residual properties that have to be satisfied to achieve the effect; in the example, the effect is just a reachability property.
	Furthermore, we want to emphasize that the concrete probabilities of the individual  paths leading to $c$ are important for the violation. In general, this imposes a major challenge for checking the condition (coS), for which we do not know a solution.
	Similar problems arise when trying to check the existence of a co-safety SPR cause. A witness might be just one individual path, potentially only together with a scheduler that realizes this path with very low probability.
	\Ende
\end{exa}

\subsubsection{Computation of quality measures of co-safety causes}
Analogously to Section~\ref{sec:quality_reachability_causes}, we can define recall, coverage ratio, and f-score of co-safety PR causes.
With the construction of $\wminMDP{\cM}{\rEff,\rCause}$ and the correspondence between schedulers $\sched$ for $\cM$ and $\tsched$ for $\wminMDP{\cM}{\rEff,\rCause}$
satisfying Equations (\ref{eqn:GPRschedulers1})-(\ref{eqn:GPRschedulers3}) established in the proof of Lemma  \ref{lem:regular-wmin-GPR}, we obtain analogously to the setting of reachability PR causes:
\begin{cor}
	\label{cor:co-safety-comp-quality}
	Let $\cM$ be an MDP and $\rEff \subseteq S^\omega$ an $\omega$-regular language given as DRA $\cA_{\rEff}$.
	Given a co-safety SPR/GPR cause $\rCause \subseteq S^\omega$ by a DFA $\cA_{\rCause}$, we can compute $\recall(\Cause)$, $\covrat(\Cause)$ and $\fscore(\Cause)$ in polynomial time.
\end{cor}

\subsubsection{Finding optimal co-safety PR causes}
Already for reachability PR causes, we have seen that without further restrictions on the causes we allow, causes might be trivial and intuitively violate the idea of temporal priority (cf. Section~\ref{sec:finding_rcauses}).
Hence, also here, we impose an additional condition, a variation of the condition \eqref{eq:TempPrio} used above:
In line with the difference between the definitions of reachability PR causes and co-safety PR causes, we require that after any good prefix $\sigma$ of a co-safety cause $\rCause$, the probability that effect $\rEff$ will occur is not guaranteed to be $1$, i.e.,  we require that 
\begin{align}
	\label{eq:TempPrio2}
	\Pr^{\min}_{\cM}(\rEff\mid \sigma)<1 \text{ for all good prefixes $\sigma$ of $\rCause$. } \tag{TempPrio2}
\end{align}
Unfortunately, we will observe that there are some obstacles in the way when trying to find optimal co-safety PR causes. 

Following the observations from Theorem \ref{thm:optimality-of-canonical-SPR-cause} we can define a canonical co-safety PR cause which is an optimal co-safety SPR cause for both recall and coverage ratio.
In this fully path-based setting this canonical cause consists of all minimal paths which are singleton co-safety SPR causes.
However, as we are not aware of a feasible way to check (coS), the computation of this cause is unclear.

For co-safety GPR causes, the following example illustrates that there might be no recall-optimal causes that respect \eqref{eq:TempPrio2}.
Intuitively, the reason is that causes can be pushed arbitrarily close towards a violation of the probability raising condition while increasing the recall:

\begin{figure}[t]
	\centering
	\centering
	\resizebox{0.45\textwidth}{!}{

\begin{tikzpicture}[->,>=stealth',shorten >=1pt,auto ,node distance=0.5cm, thick]
	\node[scale=1, state] (s0) {$\init$};
	\node[scale=1, state] (a) [right = 1 of s0, yshift=-2cm] {$a$};
	\node[scale=1, state] (b) [right =1 of s0, yshift=2cm] {$b$};
	\node[scale=1, state] (c) [right =1 of a] {$c$};
	\node[scale=1, state] (e) [right = 1 of c, yshift=2cm] {$e$};
	\node[scale=1, state] (ne) [right = 1 of c] {$f$};
			\node[scale=1, state] (ne2) [above = 1 of e] {$f$};
	
	\draw[<-] (s0) --++(-0.55,-0.55);
	\draw (s0) -- (a) node[left, pos=0.5,scale=1] {$1/3$};
	\draw (s0) -- (b) node[left, pos=0.5,scale=1] {$1/3$};
	\draw (s0) -- (e) node[below=2pt, pos=0.2,scale=1] {$1/3$};
	\draw (a) -- (c) node[below, pos=0.5,scale=1] {$1/2$};
	\draw (a) edge [loop left] node[left, pos=0.5,scale=1] {$1/2$} (a);
	\draw (b) -- (e) node[below=2pt,pos=0.3,scale=1] {$3/4$};
	\draw (b) -- (ne2) node[above=2pt,pos=0.3,scale=1] {$1/4$};
	\draw (c) -- (e) node[above=2pt, pos=0.2,scale=1] {$1/4$};
	\draw (c) -- (ne) node[below=2pt,pos=0.3, scale=1] {$3/4$};
	\draw (e) edge [loop right]  (e);
	\draw (ne) edge [loop right]  (ne);
	\draw (ne2) edge [loop right]  (ne2);

\end{tikzpicture}
	}
	\caption{Markov chain $\cM$ where the regular effect $\lozenge e$ has no recall-optimal regular PR-cause.}
	\label{fig:no-recall-optimal-cause}
\end{figure}

\begin{exa}
	This example will show that there is a Markov chain $\cM$ with a state $e$ such that the effect $\rEff=\lozenge e$ has regular GPR causes that respect the condition \eqref{eq:TempPrio2}, but no recall-optimal co-safety GPR cause that respects \eqref{eq:TempPrio2}.
	
	Consider the Markov chain $\cM$ depicted in Figure \ref{fig:no-recall-optimal-cause} with states $S$ and the effect $\rEff=\lozenge e$. 
	First of all, we have that $\Pr_{\cM}(\rEff)=2/3$. Furthermore, clearly the cause $\init\, b\, S^\omega$ with the unique minimal good prefix $\init \, b$ is a regular GPR cause for example
	as $\Pr_{\cM}(\rEff \mid \init \, b)=3/4$.
	
	Next, note that there cannot be a regular GPR cause $\rCause^\prime$ that does not have $\init \, b$ as a minimal good prefix. By (TempPrio2) the minimal good prefixes are not allowed to end in $e$.
	Furthermore $\init$ is clearly also no candidate for a minimal good prefix of $\rCause^\prime$ as that would imply that $\rCause$ consists of all paths of $\cM$ and cannot satisfy the condition \eqref{rGPR}.
	If $\init \, b$ is  not a minimal good prefix, hence all minimal good prefixes have to end in $a$, $c$, or $f$. Afterwards, the probability to reach $e$ is at most $1/4$ and hence also the probability $\Pr_{\cM}(\rEff\mid \rCause^\prime)\leq 1/4$ because it is a weighted average of the probabilities $\Pr_{\cM}(\rEff\mid \sigma)\leq 1/4$ of the minimal good prefixes $\sigma$ of $\rCause^\prime$.
	
	So, all regular GPR causes have the minimal good prefix $\init \, b$ together with potentially further minimal good prefixes. 
	Let now 
	\[
	\rCause_p \eqdef \init \, b \, S^\omega \cup A_p
	\]
	where $A_p$ is a regular subset of the paths $\{\init \, a^k \, c \mid k\geq 1\}$ such that all the paths in $A_p$ together have probability mass $\frac{1}{3}p$.
	Note that we can find such a set $A_p$ for a dense set of values $p\in [0,1]$.
	We compute
	\[
	\Pr_{\cM}(\rCause_p)= \frac{1}{3}(1+p)
	\]
	and 
	\[
	\Pr_{\cM}(\rCause_p \land \rEff)=\frac{1}{4}+\frac{1}{12}p.
	\]
	So, we obtain that $\rCause_p$ is a regular GPR cause if
	\[
	\Pr_{\cM}(\rEff\mid \rCause_p) = \frac{\frac{1}{4}+\frac{1}{12}p}{\frac{1}{3}(1+p)}>\frac{2}{3} =\Pr_{\cM}(\rEff).
	\]
	After multiplying the inequality with $12(1+p)$, we see that this holds iff $9+3p>8+8p$ iff $p<\frac{1}{5}$.
	The recall of $\rCause_p$ is now
	\[
	\Pr_{\cM}(\rCause_p \mid \rEff) = \frac{\frac{1}{4}+\frac{1}{12}p}{2/3} = \frac{3}{8}+\frac{1}{8}p.
	\]
	So, among the co-safety GPR causes of the form $\rCause_p$, there is no recall-optimal one. For  $p$ tending to $1/5$ from below, the recall always increases.
	Note also that an $\varepsilon$-recall-optimal co-safety GPR cause for $\varepsilon>0$ must take a very complicated form. It has to select paths of the form $\init \, a^k \, c$ that have probability $\frac{1}{3\cdot 2^k}$ such that there probability adds up to a value less than, but close to $\frac{1}{15}$.
	\Ende
\end{exa}

We have seen that for recall-optimal (and hence coverage ratio-optimal) SPR causes for an effect given by $\cA_{\rEff}$, we can provide a characterization of the canonical cause.
How to compute this cause, however, is unclear as we do not know how to check the co-safety SPR condition and as there might be some paths ending in a given state in the product of $\cM$ and $\cA_{\rEff}$ that belong to the canonical cause while other paths ending in that state do not.
For co-safety GPR causes, we have even seen that there might be no (non-trivial) optimal causes and that causes close to the optimum can be required to take a very complicated shape.
As the f-score is a more involved quality measure than the recall, we cannot expect that the search for f-score optimal causes is simpler. 
It seems to be likely that the situation is at least as bad as for recall-optimal causes if not worse.

\section{Conclusion}

\label{sec:conclusion}

In this work we formalized the probability-raising principle in MDPs and studied several quality notions for probability-raising causes.
We covered fundamental algorithmic problems for both the strict (local) and global view, where we considered a basic state-based setting in which cause and effect are given as sets of states.
We extended this setting to $\omega$-regular path properties as effects in two ways.
In a more simple setting we kept causes as sets of states and in a more general approach considered co-safety path properties as causes.

\subsection*{Strict vs. Global probability raising}
In our basic setting of state-based cause-effect relations, our results indicate that GPR causes are more general overall by leaving more flexibility to achieve better quality measures, while algorithmic reasoning on SPR causes is simpler.
This changed when extending the framework by considering $\omega$-regular effects given by a deterministic Rabin automaton.
Our results mainly stem from a polynomial reduction from $\omega$-regular effects to reachability effects (Lemma \ref{lem:regular-effect-wmin}).
The caveat here is that the strict PR condition translates to a global PR condition after this transformation, which increases the algorithmic complexity of reachability SPR causes to the level of reachability GPR causes.
Thus, the strict probability-raising loses its advantage over the global perspective.
Furthermore, when considering causes as co-safety path properties we observe increasing difficulties to handle strict probability-raising.
This stems from an underlying problem in the approach of strict probability-raising applied to path properties.
As we consider cause-effect relations between these properties, it is somewhat unnatural to require each individual path to raise the probability of the effect property.
Rather, it is more natural to say a path property as a whole causes another one, instead of saying all possible realizations of a path property cause another one.
This means that co-safety GPR causes also seem  more natural than co-safety SPR causes  from a philosophical standpoint.

\subsection*{Non-strict inequality in the PR conditions}
The approach of probability-raising within this work is in line with the classical notion in literature that uses a strict inequality in the PR condition.
As a consequence causes might not exist (see Example \ref{ex:no-prob-raising-cause}).
However, relaxing the PR condition by only requiring a non-strict inequality would apparently be a minor change that broadens the choice of causes.
Indeed, the proposed algorithms for checking the SPR and GPR condition for reachability effects (Section~\ref{sec:check}) can easily be modified for the relaxed definition.
As the algorithms of both extended settings discussed in Section~\ref{sec:regular} stem mainly from a reduction to reachability effects this also holds for reachability and co-safety causes of regular effects.
However, a non-strict inequality in the PR condition would lead to a questionable notion of causality, as e.g. $\{\init\}$ would always be a recall- and ratio-optimal cause.
Thus, other side constraints are needed in order to make use of the relaxed PR condition.
E.g., requiring the non-strict inequality for all schedulers that reach a cause with positive probability and also requiring the existence of a witnessing scheduler for the PR condition with strict inequality might be a useful alternative definition which agrees with Def.~\ref{def:GPR} for Markov chains.

\subsection*{Relaxing the minimality condition}
As many causality notions in the literature include some minimality constraint, we included the condition $\Pr^{\max}_\cM(\neg \Cause \until c)>0$ for all states of $\Cause$ in the state-based setting and for reachability PR causes of regular effects.
However, this requirement could be dropped without affecting the algorithmic results presented here.
This can be useful when the task is to identify components or agents which are responsible for the occurrences of undesired effects.
In these cases the cause candidates are fixed (e.g., for each agent $i$, the set of states controlled by agent $i$), but some of them might violate the minimality condition.

\subsection*{Future directions}
In this work we considered type-like causality where cause-effect relations are defined within the model without needing an actual execution that shows the effect.
Hence, causes are considered in a forward-looking manner.
Notions of probabilistic backward causality that take a concrete execution of the system into account  and considerations on PR causality with external interventions as in Pearl's do-calculus \cite{Pearl09} are left for future work.

\section*{Acknowledgments}
\noindent We would like to thank Florian Funke and Clemens Dubslaff for their helpful comments and feedback on the topic of causality in MDPs. 
We especially thank Simon Jantsch for pointing out improvements of prior results (wrt. \cite{fossacs2022}).

\bibliographystyle{alphaurl}
\bibliography{lit}

\end{document}